\documentclass[11pt]{article}

\newif\ifFULL
\FULLtrue

\usepackage{booktabs} % For formal tables

\usepackage{complexity}
\usepackage{geometry}
\usepackage{amssymb}
\usepackage{url}
\usepackage{amsmath}
\usepackage{enumitem}
\usepackage{amsthm}
\usepackage[square,numbers]{natbib}

\usepackage{enumerate}
\usepackage[small,it]{caption}
\newcommand {\ignore} [1] {}

\usepackage{comment}
\usepackage{float}
\usepackage{Pgfplots}
\usepackage[colorlinks=true]{hyperref}
\hypersetup{
    linkcolor=[rgb]{0.4,0,0.6},
    citecolor=[rgb]{0, 0.4, 0},
    urlcolor=[rgb]{0.6, 0, 0}
}

\usepackage[ruled]{algorithm2e}

\SetAlFnt{\small}
\SetAlCapFnt{\small}
\SetAlCapNameFnt{\small}
\SetAlCapHSkip{0pt}
\IncMargin{-\parindent}

\newtheorem{theorem}{Theorem}[section]

\newtheorem{corollary}[theorem]{Corollary}
\newtheorem{lemma}[theorem]{Lemma}

\newtheorem{definition}[theorem]{Definition}

\def\squarebox#1{\hbox to #1{\hfill\vbox to #1{\vfill}}}

\newcommand{\ceil}[1]{\lceil {#1}\rceil}
\newcommand{\cardinality}[1]{\left\vert{#1}\right\vert}
\newcommand{\comp}[1]{\bar{#1}}
\newcommand{\expect}[2]{\mathop{\mathbb{E}}_{#1} \left[ #2 \right]}
\newcommand{\var}{\mathrm{Var}}
\newcommand{\q}{q}
\newcommand{\val}{v}
\newcommand{\vals}{\mathbf{v}}
\newcommand{\reals}{\mathbb{R}}
\newcommand{\naturals}{\mathbb{N}}
\newcommand{\low}{L}
\newcommand{\high}{H}
\newcommand{\addedBuyers}{h}
\newcommand{\largeConstReg}{K}
\newcommand{\dist}{D}
\newcommand{\distM}{\mathcal{D}}
\newcommand{\distMtail}{\distM_{\textsc{Tail}}}
\newcommand{\distMcore}{\distM_{\textsc{Core}}}
\DeclareMathOperator{\supp}{supp}
\newcommand{\virtVal}{\varphi}
\newcommand{\revCurve}{\mbox{rev}}
\newcommand{\demCurve}{V}

\newcommand{\cutoff}{\mathcal{T}}
\newcommand{\ind}{\mathbb{I}}
\newcommand{\dens}{f}
\newcommand{\rev}[1]{#1Rev}

\newcommand{\BREV}{\textsc{\rev{B}}}
\newcommand{\SREV}{\textsc{\rev{S}}}

\newcommand{\REV}{\textsc{\rev{}}}
\newcommand{\VCG}{\textsc{{\small VCG}}}
\newcommand{\BVCG}{\textsc{{\small BVCG}}}

\newcommand{\VAL}{\textsc{Val}}
\newcommand{\event}{\mathcal{A}}
\newcommand{\eventB}{\mathcal{B}}
\newcommand{\events}{\mathcal{L}}
\renewcommand{\M}{\mathcal{M}}
\newcommand{\pay}{{\bf p}}
\newcommand{\payL}{{\pi}}

\newcommand{\alloc}{ \pi}
\newcommand{\piBar}{ \bar{\pi}}

\usepackage{setspace}
\usepackage[titles]{tocloft}
\definecolor{MyGray}{rgb}{0.8,0.8,0.8}

\newcommand{\squishlist}{
    \begin{list}{$\bullet$}
        { \setlength{\itemsep}{0pt}      \setlength{\parsep}{3pt}
            \setlength{\topsep}{2pt}       \setlength{\partopsep}{0pt}
            \setlength{\leftmargin}{1.5em} \setlength{\labelwidth}{1em}
            \setlength{\labelsep}{0.5em} } }

    \newcommand{\squishend}{
\end{list}  }

\begin{document}
% EC Format comment out:
\ifFULL
\title{99\% Revenue via Enhanced Competition}
\else
\title{99\% Revenue via Enhanced Competition\\{\large(Extended abstract)}}
\fi

\author{
	Michal Feldman\thanks{Computer Science, Tel-Aviv University, and Microsoft Research. \texttt{michal.feldman@cs.tau.ac.il}.
		This work was partially supported by the European Research Council under the European Union's Seventh Framework Programme (FP7/2007-2013) / ERC grant agreement number 337122, and by the Israel Science Foundation (grant number 317/17).
		} 
	\and 
	Ophir Friedler\thanks{Computer Science, Tel-Aviv University. \texttt{ophirfriedler@gmail.com}.}
	\and 
	Aviad Rubinstein\thanks{Harvard University, aviad@seas.harvard.edu. This work was supported by the Rabin Postdoctoral Fellowship. Part of the work was also done while the author was a student at UC Berkeley (supported by a Microsoft Research PhD Fellowship) and a visitor at Tel-Aviv University (supported by ERC 337122)}}

\maketitle

% !TeX root = main99revenue.tex

\begin{abstract}
A sequence of recent studies show that even in the simple setting of a single seller and a single buyer with
additive, independent valuations over $m$ items, the revenue-maximizing mechanism is prohibitively complex.
This problem has been addressed using two main approaches:
\begin{itemize}[leftmargin=*]
\item {\em Approximation:} the best of two simple mechanisms (sell each item separately, or sell all the items as one bundle) gives $1/6$ of the optimal revenue \cite{babaioff2014simple}.
\item {\em Enhanced competition:} running the simple VCG mechanism with additional $m$ buyers extracts at least the optimal revenue in the original market \cite{eden2016competition}.
\end{itemize}
Both approaches, however, suffer from severe drawbacks:
On the one hand, losing $83\%$ of the revenue is hardly acceptable in any application.
On the other hand, attracting a linear number of new buyers may be prohibitive.
Our main result is that by combining the two approaches one can achieve the best of both worlds.
Specifically, for any constant $\epsilon$ one can obtain a $(1-\epsilon)$ fraction of the optimal revenue by running simple mechanisms --- either selling each item separately or selling all items as a single bundle --- with substantially fewer additional buyers: logarithmic, constant, or even none in some cases.
\end{abstract}

\thispagestyle{empty}\maketitle\setcounter{page}{0}

\newpage

\ifFULL
{\setstretch{0.99}
\tableofcontents}
\newpage
\fi

\section{Introduction}
% !TeX root = main99revenue.tex
\label{SEC:INTRO}

The scenario of a buyer with an additive, independent valuation over $m$ items has become the paradigmatic setting for studying optimal (revenue-maximizing) mechanisms.
In this setting, the buyer's valuation is drawn from a distribution $\dist$ that is known to the seller, and the seller wishes to design a selling mechanism that extracts as much revenue as possible.

By now it is well known that the optimal mechanism requires randomization \cite{HN13}, infinite, uncountable menus \cite{DaskalakisDT16}, is non-monotone \cite{HartR12}, and computationally intractable~\cite{daskalakis2014complexity}%
\footnote{Similar undesired properties have been observed with respect to the optimal mechanism in additional related (multi-dimensional) models, e.g., unit-demand buyers \cite{briest2010pricing, ChenDOPSY15,ChenDPSY14,RochetC98}}.
Thus it is mostly interesting as a theoretic benchmark to which one can compare more plausible mechanisms (much like the way an offline optimum serves as a benchmark in online settings).
In recent years, two main approaches have been taken with respect to this challenge, both of which proposed simple mechanisms and measured their performance against the theoretic optimum:

The first line of work approaches this problem through the lens of {\em approximation} \cite{chawla2007algorithmic,chawla2010multi,kleinberg2012matroid,
hart2012approximate,LiY13,
babaioff2014simple,ChawlaMS10,Yao15,rubinstein2015simple,BateniDHS15,MS15-production_costs,cai2016duality,CaiZ16,chawla2016mechanism,EdenFFTW16a,Yao17-monotonicity}.
In particular, the breakthrough result of~\cite{babaioff2014simple} shows that
the better of two simple mechanisms --- selling each
item separately or selling all items together in a grand bundle ---
obtains at least $1/6$ (but at most $1/2$~\cite{rubinstein2016computational}) of the optimal revenue.

While a constant factor approximation algorithm may sound appealing to an algorithm designer, losing 83\% (or even 50\%) of the revenue is simply unacceptable.
Indeed, it will be difficult to convince a merchant who hopes to make $\$10K$ in revenue to sell her merchandise by a mechanism that would guarantee her $\$1.7K$.\footnote{There are still many great reasons to study constant-factor approximations in mechanism design; see, e.g., \cite{hartline2013mechanism} for an excellent discussion.}
On the other hand, a seller might be willing to compromise on optimality if guaranteed 99\% of the optimal revenue.
(For example, merchants around the world pay small fees to credit card companies in return for simple selling mechanisms).
We therefore believe that the most interesting agenda here should be obtaining 99\% of the optimal revenue.
(More generally, we are interested in mechanisms whose revenue is arbitrarily close to optimum; i.e., $(1-\varepsilon)$-fraction of the optimal revenue for any constant $\varepsilon$.)

The second approach is to enhance the competition for the merchandise by increasing the population of potential buyers \cite{BulowK96,RTY15,eden2016competition,LiuP2017}.
The state of the art for additive buyers is by
Eden {\em et al.} \cite{eden2016competition},
who showed that adding $m$ additional buyers is sufficient to recover the original optimal revenue with a simple mechanism. It is also known that at least $\Omega(\log m)$ additional buyers are required to achieve this benchmark.
This result from~\cite{eden2016competition} generalizes the seminal work of
Bulow and Klemperer \cite{BulowK96}
who showed that for a market with a single item, under a regularity assumption, running the second price auction with one additional buyer extracts at least as much revenue as the original optimal revenue.
However, when $m$ is large, adding $m$ additional buyers may be prohibitive.
We therefore believe that the most interesting question here is whether the linear dependence on $m$ is necessary.

To summarize, the ``optimal" benchmark is intractable. The approximation approach is stuck at a $1/6$-approximation (forfeiting 83\% of potential revenue). And if we wish to follow the enhanced competition approach, the best known bound on the number of additional buyers is linear in the number of items. Have we reached a dead end?

\subsection{Our Contribution}

We show that one can combine the two approaches (of approximation and enhanced competition) in a way that achieves the best of both worlds. 
We establish a host of results for various settings, but they all convey one theme: in order to obtain revenue that is very close to optimum, there is no need to recruit a linear number (in $m$) of additional buyers; that is:

\vspace{2mm}

\noindent {\bf Main take away (informal):} {\em
A seller can obtain 99\% of the optimal revenue in a simple mechanism (selling each item separately or selling all items together in a singe bundle) with substantially fewer additional buyers --- logarithmic (in $m$), constant, or even none in some cases.
}

\vspace{2mm}

All of our results apply to the paradigmatic scenario of a seller who sells $m$ items to a single buyer with additive, independent valuations over the items, with a known prior. Some of our results extend to the more general scenario of $n$ i.i.d.~buyers, namely where buyers' values are drawn independently according to the same distribution. The induced product distribution is then denoted by $\distM$.
Our first set of results consider the simple mechanism that sells each item separately.

\begin{theorem} \label{thm:SREV-one-buyer}
For every constant $\varepsilon>0$, selling each item separately to $O(\log m)$ i.i.d.~buyers extracts at least $(1-\varepsilon)$-fraction of the optimal revenue achievable by a single buyer.
\end{theorem}

This result improves upon the $O(m)$ bound shown in \cite{eden2016competition}, at the loss of $\varepsilon$ fraction in revenue.
In fact, this result can be extended even to a setting with $n$ i.i.d.~buyers and $m$ items, as follows:

\begin{theorem} \label{thm:SREV-many-buyers} (implies Theorem \ref{thm:SREV-one-buyer})
For a setting with $n$ i.i.d.~buyers and $m$ items, for every constant $\varepsilon>0$, selling each item separately achieves at least $(1-\varepsilon)$-fraction of the optimal revenue if we increase the number of buyers by a factor of $O(\log(2+m/n))$,
and this is tight up to constant factors.
Moreover, if $m = o(n)$, then selling each item separately achieves $(1-\varepsilon)$-fraction of the optimal revenue even with no additional buyers.
\end{theorem}

Theorem~\ref{thm:SREV-many-buyers} essentially fully characterizes (up to constant factors) the number of additional buyers necessary for achieving $(1-\varepsilon)$ of the optimal revenue.
In particular, we consider three different regimes of $m$ and $n$:
For $m=\omega(n)$ we prove that increasing the number of buyers by a factor of $O(\log({m/n}))$ is both necessary
\ifFULL
(Theorem~\ref{thm:lb})
\fi
and sufficient%
\ifFULL
~(Theorem~\ref{thm:ManyAdditivebuyers})%
\fi
.
For $m= \Theta(n)$, our new lower bound implies that the previous results of~\cite{eden2016competition} (who showed that a linear number of additional buyer suffice) are essentially tight.
Finally, for $m=o(n)$, we show that {\em no} additional buyers are necessary%
\ifFULL
~(Theorem~\ref{thm:ManyBuyersNoNeedForMore}).
\else
. See Section~\ref{sec:techniques} for details.
\fi
We note that our lower bound generalizes the special case of $\Omega(\log m)$-factor for the case
of a single buyer in~\cite{eden2016competition} %
\footnote{The lower bound in
\cite{eden2016competition} was proven for mechanisms that target
100\% of the optimal revenue, but it can be easily extended to
mechanisms that target 99\% or any constant fraction.}.

Let us return to the single buyer setting. Theorem~\ref{thm:SREV-many-buyers} implies that we can recover $(1-\varepsilon)$-fraction of the optimal revenue by adding $O(\log m)$ buyers.
However, one may argue that for a large value of $m$, $O(\log m)$ is still too large.
We address this issue by showing that the better of selling items separately and selling the grand bundle requires only a constant number of additional buyers\footnote{
There is no contradiction to the $\Omega(\log m)$ lower bound, which applies only for selling items separately.}.

\begin{theorem}  \label{thm:SREVBREV}
For every constant $\varepsilon>0$, the better of selling each item separately and selling the grand bundle, to a constant number of i.i.d.~buyers, extracts at least $(1-\varepsilon)$-fraction of the optimal revenue achievable from one buyer.
\end{theorem}

Up until now, we concentrated on {\em prior-dependent} mechanisms; namely, mechanisms that use the knowledge of the distribution of values.
Our work, however, contributes also to the literature on {\em prior-independent} mechanisms.
As in previous literature on prior-independent mechanisms, to achieve any meaningful result, we assume that the underlying single-dimensional distributions are regular (note that this does not generally imply regularity of the grand bundle's distribution).
Since bidders are additive, the prior-independent VCG mechanism simply runs the second price auction for each item simultaneously.
Thus,
Theorem~\ref{thm:SREV-many-buyers} combined with the original result of
Bulow and Klemperer \cite{BulowK96}
immediately implies the following corollary:
\begin{corollary}
If $\distM$ is a product of regular distributions, then for every constant $\varepsilon>0$, running the VCG mechanism with a factor $O(\log (2+m/n))$ increase in the number of buyers (and with no additional buyers in the case of $m=o(n)$) extracts at least $(1-\varepsilon)$-fraction of the original optimal revenue.
\end{corollary}

It would be highly desirable to obtain such an analog to Theorem~\ref{thm:SREVBREV}.
An immediate barrier, however, is that the seller must know in advance whether to sell the items separately or sell the grand bundle.
How can the seller determine which one of these mechanisms to run in the absence of a prior?
This barrier is overcome by the surprising result that when $\distM$ is a product of regular distributions, the seller {\em never} needs to sell items separately; selling the grand bundle is always the ``correct'' strategy:

\begin{theorem} \label{thm:BREVregular}
If $\distM$ is a product of regular distributions, then for every constant $\varepsilon>0$, selling the grand bundle in a second price auction to a constant number of i.i.d.~additive buyers extracts at least $(1-\varepsilon)$-fraction of the optimal revenue achievable from one buyer.
\end{theorem}

\subsection{Relation to work on approximate revenue maximization}

As already mentioned, our work is related to and inspired by a long line of research that aims at understanding what fraction of the optimal revenue can be guaranteed with simple mechanisms (without adding buyers), including  \cite{LiY13,babaioff2014simple, GoldnerK16} and references therein.
It is interesting to note that obtaining a $0.99$-approximation to the optimal revenue with $k$ additional buyers has implications on approximation results without additional buyers.
In particular, if one can obtain a $0.99$-approximation to the optimal revenue using mechanism $\mathcal{M}$ with $k$ additional buyers, then, by revenue submodularity, this immediately implies a $0.99/(k+1)$-approximation of the optimal revenue with a single buyer.
Hence, our Theorem~\ref{thm:SREV-one-buyer} implies the result of~\cite{LiY13} that selling each item separately is an $\Omega(1/\log m)$-approximation of the optimal revenue. Similarly, our Theorem~\ref{thm:SREVBREV} implies the main result of~\cite{babaioff2014simple} that the better of selling separately and selling the grand bundle yields a constant fraction of the optimal revenue.
Of course, the reverse implication is not true, since the seller only has a single copy of each item to allocate to all buyers.
Therefore, while our analysis builds on the techniques of \cite{LiY13,babaioff2014simple}, it is significantly more challenging.

Recently,
Goldner and Karlin \cite{GoldnerK16}
built on the result of
Babaioff {\em et al.} \cite{babaioff2014simple},
and showed that for regular distributions, the prior-independent mechanism which sells either each item separately or the grand bundle, and sets the price according to a single sample from the valuations distribution, obtains a constant fraction of the optimal revenue.
As a corollary of our Theorem~\ref{thm:BREVregular}, we can obtain the following qualitative strengthening of~\cite[Corollary 2]{GoldnerK16} for the case of a single buyer:
\begin{corollary}
    If $\distM$ is a product of regular distributions, taking a single sample of the value of the grand bundle and selling the grand bundle for that price (to a single buyer) extracts at least $\Omega(1)$ fraction of the optimal revenue.
\end{corollary}

For the special case where the single-dimensional distributions satisfy the monotone hazard rate (MHR) condition~\footnote{Roughly speaking, MHR distribution (a special case of regular distributions) has a tail that is thinner than that of an exponential distribution.}, Cai and Huang~\cite{CaiH13} gave a PTAS for the revenue maximizing auction.
En route to obtaining their computational result, they prove the following structural lemma, which holds even for heterogeneous buyers. 99\% of the optimal revenue can be obtained by one of the following simple mechanisms: (i) selling every item separately; or (ii) selling all but a constant number of items via a VCG-like mechanism.
In other words, they show that when the single-dimensional distributions satisfy the MHR condition, no additional buyers are necessary to obtain 99\% of the optimal revenue with mechanisms that are simple (in the above sense).
Moreover, inspired by \cite{CaiD11},
they show that for MHR distributions, when the number of buyers is larger than some constant, selling items separately obtains 99\% of the {\em social welfare}.
(In contrast, note that for regular distributions the ratio of social welfare to revenue may be unbounded, even when adding any number of buyers.)

% !TeX root = main99revenue.tex
\section{Model}
We consider a setting in which a monopolist seller sells a set $[m] = \{1, 2, \ldots, m\}$ of heterogeneous items to $n$ additive buyers.
Buyer $i$'s value is additive if there exist values $\val^{i}_{1}, \ldots, \val^{i}_{m}$ such that buyer $i$'s value for a set of items $A$ is $\sum_{j\in A}{\val^{i}_{j}}$.
The seller does not know the buyers' values, but knows the distribution from which they are sampled.
For every item $j$, the value of each agent $i$ for item $j$ is independently sampled from a single-dimensional distribution $\dist_j$.

A mechanism is given by a pair of an allocation function $\alloc$, and a payment function $\pay$.
The mechanism receives a valuation profile $\vals = \{\val^i_j\}_{i,j}$ as input.
Based on $\vals$, the allocation function determines the (possibly random) allocation of items to buyers, and the payment function determines the payment $\pay^{i}$ of every buyer $i$.
Buyers are quasi-linear; namely, buyer $i$'s utility from a mechanism that allocates her each item $j$ with probability $\alloc^{i}_{j}$ and charges her a payment $\pay_i$ is $\sum_{j}{\alloc^{i}_{j} \cdot \val^{i}_{j}} - \pay^{i}$.

A mechanism is Bayesian Individually Rational (BIR) if every buyer's expected utility (over the randomness of the mechanism and other buyers' values, assuming they are drawn from $\times_j \dist_j$) is non-negative.
A mechanism is Bayesian Incentive Compatible (BIC) if every buyer's expected utility is maximized when the buyer reports her valuation truthfully.
Throughout the paper, a BIR-BIC mechanism is termed a {\em truthful} mechanism.
The expected revenue of a truthful mechanism is $\expect{}{\sum_{i}{\pay^i(\vals)}}$, where for every buyer $i$ and item $j$ buyer $i$'s value for item $j$ is drawn independently from $\dist_j$.
The {\em optimal revenue} is the optimal%
\footnote{In general for distributions of infinite support, the optimum revenue may not be achieved by any mechanism (i.e. it is the supremum of all revenues achievable by truthful mechanisms). This is not so important for our purposes since we are only trying to approximate the optimum.}
 expected revenue among all truthful mechanisms.

\ifFULL
For any single dimensional distribution $\dist$ and any $q\in [0,1]$, we assume w.l.o.g. that there exists $v(q) \in \mathbb{R}$ such that $\Pr_{x\sim \dist}[x \geq v(q)] = q$. When the distribution is continuous this is true by the intermediate value theorem. When the distribution has point masses, we can smooth it with an infinitesimal perturbation;
see e.g.,~\cite{rubinstein2015simple} for a formal discussion.
We note that when $\dist$ is regular, the perturbed distribution may not be regular.
However, it will not be hard to see that our proof in Section~\ref{SEC:BREV} will also work with distributions that are infinitesimally-close to regular distributions.
\fi

% !TeX root = main99revenue.tex

\section{Overview of Proofs}\label{sec:techniques}
Our techniques build upon the now-standard approach for approximately optimal revenue analysis:
Separately reason about revenue contribution from rare events of extremely high value (``tail'' events),
and revenue contribution from lower values (``core'' events).
This approach is made formal via the core-tail decomposition framework of Li and Yao \cite{LiY13}.
Separation between core and tail events is done by setting a cutoff for each item,
and then proving a ``core-tail decomposition'' lemma that upper bounds
the optimal revenue by the sum of total value from core events in which items are below their cutoffs (a.k.a. contribution from the core), and the
total revenue from tail events in which items are above their cutoffs (a.k.a. contribution from the tail).

For approximation results, the next step would typically be to show that simple mechanisms approximate both the contribution from the core and the tail, hence the best simple mechanism approximates the optimal revenue.
However, when targeting 99\% of the optimal revenue, even if one shows that simple mechanisms fully recover the contribution from the core and the tail, it does not yet imply that simple mechanisms recover the {\em sum} of the contributions. Hence this approach alone cannot yield better than a 1/2-approximation.
Instead, we carefully set the cutoffs so that the contribution from the core is almost fully recovered using simple mechanisms with either $O(1)$ buyers (Theorems~\ref{thm:SREVBREV} and~\ref{thm:BREVregular}) or $O(\log m)$ buyers (Theorem~\ref{thm:SREV-many-buyers}), and the contribution from the tail is only a tiny fraction of the revenue from simple mechanisms with $O(1)$ buyers.

Before delving into the specificities of each result, we provide a few general notes about the tail and the core.

\paragraph{Tail}
Since tail events are rare, they are approximately ``separable'' across items and across buyers.
First, since the probability of two or more items in the tail of any buyer is low, 
the revenue contribution from tail items is roughly ``separable'' across items,
i.e., approximating revenue contribution of tail events by the revenue contribution from selling items separately loses only a moderate factor.
Furthermore, each buyer is likely to have tail valuations for disjoint subsets of items (``separability across buyers'').
Thus in an enhanced competition setting with more buyers we can simultaneously serve all of them, quickly increasing the revenue.

\paragraph{Core} Since the value of core items is bounded, a concentration bound typically suggests a near-optimal grand bundle price whenever the sum of item values is significantly larger than the revenue from selling items separately.

\vspace{0.2in}
In what follows we elaborate on cases where interesting artifacts arise in the analysis.
We use \REV, \SREV, and \BREV{} to denote the optimal revenue from any mechanism, the optimal revenue from selling each item separately, and the optimal revenue from selling the grand bundle, respectively.

\paragraph{Theorem~\ref{thm:SREVBREV} ($\max\{\SREV,\BREV\}$)}
An optimistic approach for proving Theorem~\ref{thm:SREVBREV} would use the same cutoff for all items, and apply
a concentration bound (Chebyshev's inequality) to suggest a bundle price that almost fully recovers the contribution from the core with constant probability, then improve this probability to almost $1$ by considering additional buyers.
Unfortunately, such strong concentration does not hold in general: a small number of items may still exhibit a large variance, even inside the core.
To overcome this challenge, we separate the set of items to items that exceed their cutoffs with constant probability, which we call ``high items'', and the remaining we call ``low items''.
We then use a concentration bound only for the sum of low items to get a good bundle price for the low items.
We set the cutoff so that the number of high items is a constant, hence
with probability close to $1$, one of our additional $O(1)$ buyers simultaneously exceeds the cutoff of all the ``high items''.
We conclude that $\BREV$ with $O(1)$ buyers recovers almost all of the contribution from the core with probability almost $1$.
\ifFULL
The proof appears in Section~\ref{SEC:SREVBREV}.
\fi

\paragraph{Theorem~\ref{thm:BREVregular} ($\BREV$, regular distributions)}
Single dimensional regular distributions are appealing since they have a ``small tail'' property.
While the grand bundle distribution need not be regular even when the individual item distributions are, we show that
the underlying regularity of the individual items still maintains some well behaved properties that we can exploit.
Specifically, we can set cutoffs so that in the resulting core-tail decomposition, the contribution from the tail is significantly smaller than the contribution from the core.
This is important since we are only allowed to sell the grand bundle, while the tail is typically covered by selling items separately.

It is now left to guarantee that the core is almost fully recovered.
The challenge is that a small number of outlier items may have very large variance, ruining the naive concentration bound argument.
(In Theorem~\ref{thm:SREVBREV} we sometimes sold the outliers separately, but here we are only allowed to sell the grand bundle.)
We show that given a large (but still constant) number of additional buyers,
two of them are likely to have high values simultaneously for all of the outlier items.
We can then use a concentration bound to suggest a good price for the remaining items, so that VCG on the grand bundle gives high revenue. 
\ifFULL
The proof appears in Section~\ref{SEC:BREV}.
\fi

\paragraph{Theorem~\ref{thm:SREV-many-buyers} ($\SREV$), $m=\omega(n)$ regime}
Instead of a concentration bound for the grand bundle (which is irrelevant when considering $\SREV$) we further separate the contribution from the core to the contributions from lower and higher values.
We then show that the contributions of lower values to the core can be almost fully recovered with probability almost $1$ by $\SREV$ with $O(n \cdot \log (m/n))$ buyers,
while the contributions of higher values to the core form only a tiny fraction of the revenue that can be extracted by $\SREV$ with $O(n \cdot \log (m/n))$ buyers.
\ifFULL
The proof appears in Section~\ref{SEC:SREV}.
\fi

\paragraph{Theorem~\ref{thm:SREV-many-buyers} ($\SREV$),  $m=o(n)$ regime}
We show that selling each item separately achieves $99\%$ of the optimal auction without adding {\em any buyers at all}.
Proving this result requires yet new ideas.
The intuition is simple: with so many buyers, we can afford to set the cutoff much higher -- so high that {\em most buyers } have at most one item in the tail --
while the contribution from the core can be almost fully recovered by selling each item to a tiny fraction of the buyers.
Therefore, we can set aside a small subset of special buyers, and offer the items at high (tail) prices to all other buyers; we can then recover the rest of the revenue by auctioning the items not previously sold to the special buyers.
(Note that this mechanism is for analysis purposes only --- once we establish guarantees for any mechanism for selling items separately, we can use Myerson's optimal mechanism for selling items separately.)
This part of the analysis also takes into account the probability that an item was sold in tail events.

Formalizing the above intuition is quite subtle.
With high probability, there are still {\em some} buyers with multiple items in the tail -- and each of those items has many other buyers whose valuations are in the tail, etc.
Thus even after the core-tail decomposition, we have to reason about a multiple buyer, multiple item setting.
To cope with this difficulty, we consider a bipartite graph of buyers and items, where we draw an edge $\{i, j\}$ whenever buyer $i$'s value of item $j$ is in the tail.
We then argue that we can analyze the revenue from each connected component separately.
In particular, while there are, w.h.p., large components in the graph (with $\omega(1)$ buyers and items), we use simple ideas from percolation theory to bound the expected size of each connected component.
\ifFULL
The proof appears in Section~\ref{SEC:MOREBUYERS}.
\fi

\ifFULL
\section{Preliminaries}
% !TeX root = main99revenue.tex
\label{SEC:PRELIM}

A monopolist seller sells a set $[m] = \{1, 2, \ldots, m\}$ of heterogeneous items to $n$ additive buyers.
We assume the seller is able to attract more buyers.
An additive buyer $i$ that values each item $j$ at $\vals_j^{i}$ values a set of items $A$ at $\sum_{j\in A}{\vals_{j}^{i}}$.
For a set of items $A \subseteq [m]$, let $\comp{A} = [m]\setminus A$.
We use $\vals^{i}_{A} = \{\vals^{i}_j\}_{j \in A}$, therefore $\vals^{i} = (\vals^{i}_{A}, \vals^{i}_{\comp{A}})$.
Also, we use $\vals_j = \{\vals_j^i\}_{i \in [n]}$,
    and more generally, for a set $S \subseteq [n]\times [m]$,
    we use $\vals_S = \{\vals_j^i \}_{(i, j) \in S}$
    and
    $\vals_{-S} = \{\vals_j^i \}_{(i, j) \not \in S}$.
Every buyer $i$ is quasi-linear, i.e., in a randomized outcome that allocates item $j$ to buyer $i$ with probability $\alloc_{j}^{i}$, and charges a payment $\pay^i$, $i$'s utility is $\sum_{j}{\alloc_{j}^{i} \cdot \vals_{j}^{i}} - \pay^i$.
The seller does not know $\vals_j^i$, but has a prior-distribution $\dist_j$ with density $\dens_j$ for each item $j$, i.e.,
$\vals_j^i$ is drawn from $\dist_j$.
Let $\distM = \times_{j\in [m]}\dist_j$.
Therefore, in a setting with $n$ i.i.d. buyers, each drawn from $\distM$, the prior distribution is $\distM^n$%
\footnote{By $\distM^n$ we mean the product distribution  $\times_{i \in [n]}\distM$.}.

Let $\M$ be a truthful mechanism, i.e., Bayesian Individually Rational (BIR)%
\footnote{A mechanism is Bayesian Individually Rational (BIR) if every buyer's
expected utility is non-negative.} and Bayesian Incentive Compatible (BIC)%
\footnote{A mechanism is Bayesian Incentive  Compatible (BIC) if every buyer's
    expected utility (over the randomness of $\M$ and the other buyers' values according to the prior distribution) is maximized by truth-telling.},
 with an
allocation function $\alloc$ and a payment function $\pay$, i.e.,
given the submitted bids $\vals$, each buyer $i$ is allocated item
$j$ with probability $\alloc_{j}^{i}(\vals)$ and pays
$\pay^i(\vals)$. Let $\REV_{\M}(\cdot)$ be mechanism $\M$'s expected
revenue, e.g., $\REV_{\M}(\distM^\addedBuyers)$ is mechanism $\M$'s
expected revenue from $\addedBuyers$ buyers, where each buyer's
valuation $\vals^i$ is drawn i.i.d. from $\distM$, i.e.,
$\REV_{\M}(\distM^{\addedBuyers}) = \expect{\vals \gets
\dist^h}{\sum_{i\in [\addedBuyers]}{\pay^i(\vals)}}$. Let
$\REV(\cdot)$ be the optimal expected revenue by any truthful
mechanism, e.g., $\REV(\distM^{\addedBuyers})$ is the optimal
expected revenue by any truthful mechanism for $\addedBuyers$ i.i.d.
buyers drawn from $\distM$. Similarly, $\SREV(\cdot)$ is the optimal
expected revenue when selling the items separately, and
$\BREV(\cdot)$ is the optimal expected revenue when selling all
items in a bundle. Let $\VAL(\cdot)$ be the expected value of the
optimal allocation. Since we consider additive, independently drawn
buyers, $ \VAL(\distM^{\addedBuyers}) = \sum_j{\expect{}{\max_{i \in [h]}
\vals_{j}^{i}}}$.

For a single dimensional distribution $\dist$, and a number $p \in
[0, 1]$, let $\REV_p(\cdot)$ denote the revenue of an optimal
mechanism, among all
truthful mechanisms that {\em sell with ex-ante probability of at
most $p$}, i.e., mechanisms with allocation function $\alloc$ that
satisfies $ \expect{\vals}{\sum_{i} \alloc^i(\vals) } \leq
p.$

Throughout the analysis, for
$\addedBuyers$ that is not an integer, we slightly abuse notation
and use, e.g., $\dist^{\addedBuyers}$ instead of $\dist^{\lceil \addedBuyers \rceil}$ to denote the product distribution
$\times_{i \in \lceil \addedBuyers
\rceil} \dist$.
For a random variable $X$ drawn from distribution $\dist$,
we use $\REV(X)$ and $\REV(\dist)$ interchangeably (and similarly
for $\SREV, \BREV$, etc.). Also, let $\ind\left[ \mathcal{E}
\right]$ be the indicator random variable that equals $1$ when event
$\mathcal{E}$ occurs, and $0$ otherwise.

Finally, we will use the following previously shown lemma.
\begin{lemma} \cite{babaioff2014simple} \label{lem:revSrev}
    For any $n \cdot m$ dimensional distribution $\distM$ (i.e., $n$ buyers and $m$ items),
    $\REV(\distM)\leq n \cdot m\cdot \SREV(\distM)$.
\end{lemma}

\subsection{Lemmas for single dimensional distributions.}
% !TeX root = main99revenue.tex
\label{SEC:SINGLE_DIMENSIONAL_LEMMAS}
We recall and develop tools for our analysis in the remaining sections.
Let $\val$  be drawn from a single dimensional distribution $\dist$ (i.e.,  $\val \gets \dist$),
and consider some cutoff $\cutoff \in \reals$.

\begin{lemma}  \cite{hart2012approximate}\label{lem:condRevIndRev}
    $
    \REV(\val \vert \val >\cutoff) =  \frac{\REV(\val \cdot \ind\left[\val >\cutoff \right])}{\Pr\left[\val > \cutoff \right]}
    $.
\end{lemma}
\begin{lemma} \label{lem:tail1}
    \cite{LiY13} Let $\gamma > 0$.
    Then $\Pr\left[\val > \gamma \cdot \REV(\dist) \right] \leq \frac{1}{\gamma}$.
\end{lemma}
Let $\var(\dist)$ denote the variance of $\dist$.
\begin{lemma}\cite{LiY13} \label{lem:boundVar}
    Let $\gamma >0$, and supposed that both $\REV(\dist) \leq \gamma$ and that the support of $\dist$ is in $[0, t \gamma]$. Then $\var(\dist) \leq (2t-1)\gamma^2$.
\end{lemma}

The following lemma shows that for $\delta >0$ that is not too small, the optimal revenue from $\ceil{1/\delta}$ i.i.d. buyers drawn from $\dist$ is
at least a $1/(2\delta)$ factor larger than the optimal revenue that can be extracted from a buyer with value distributed according to $\val \cdot \ind\left[\val >\cutoff \right]$.
This lemma will be used
to bound the contribution (both to the core and tail) of higher values.
\begin{lemma} \label{lem:addbuyers}
Suppose that $\delta \geq 2\cdot  \Pr\left[\val > \cutoff \right]$. Then
    $
    \REV(\val \cdot \ind\left[\val >\cutoff\right] ) \leq  2\delta\cdot  \REV(\dist^{1/\delta})
    $.
\end{lemma}
\begin{proof}
    Since $\val \cdot \ind\left[\val >\cutoff\right]$ is single-dimensional, there exists
    a price $\pi \geq \cutoff$ that extracts the optimal revenue, i.e.,
    $
    \pi \cdot \Pr\left[\val >\pi \right]
    =
    \pi \cdot \Pr\left[\val  \cdot \ind\left[\val >\cutoff\right]>\pi \right]
    =
    \REV(\val \cdot \ind\left[\val >\cutoff\right] )$.
    For every $i\in [\delta^{-1}]$, let $\val^{i}$ be drawn from $\dist$.
    Therefore,
    $
        \REV(\dist^{1/\delta}) \geq \pi \cdot \Pr\left[\max_{i\in [\delta^{-1}]} \val^{i} > \pi \right]
    $.
    By Bonferroni inequalities (the inclusion-exclusion principle),
    $$
    \Pr\left[\max_{i\in [\delta^{-1}]} \val^{i} > \pi \right]
    \geq
    \delta^{-1}\cdot \Pr\left[ \val > \pi \right] - \binom{\delta^{-1}}{2}\cdot \Pr\left[ \val > \pi \right]^{2}
    \geq
    (1-\tfrac{1}{4})\cdot \delta^{-1}\cdot \Pr\left[ \val > \pi \right],
    $$
    where the last inequality follows by the lemma's condition: $ 1/2 \geq \delta^{-1}\cdot \Pr\left[\val > \cutoff \right]\geq \delta^{-1}\cdot \Pr\left[\val > \pi \right]$.
    Therefore, we conclude that
    $
    \REV(\dist^{1/ \delta}) \geq ( \frac{4}{3}\cdot \delta)^{-1}\cdot \pi \cdot \Pr\left[ \val > \pi \right] = (\frac{4}{3}\cdot \delta)^{-1} \REV(\val \cdot \ind\left[\val >\cutoff \right] )
    $.
    For ease of exposition, we relax the $4 / 3$ to $2$.
\end{proof}

The following lemma will be used to cover the contribution of values from the core.
\begin{lemma} \label{lem:addbuyersAlmost}
    Let $\cutoff_\alpha$ satisfy $\Pr_{\val \gets \dist}\left[\val >\cutoff_\alpha \right] \geq \alpha$.
    For any $\addedBuyers \geq 1$,
    $$
    \REV(\dist^{\addedBuyers}) \geq (1-e^{-\alpha\cdot \addedBuyers}) \cdot \cutoff_\alpha.
    $$
\end{lemma}
\begin{proof}
    \begin{align*}
    \REV(\dist^{\addedBuyers})
    \geq &
    \cutoff_\alpha \cdot \Pr\left[\max_{i\in [\addedBuyers]} \val^i > \cutoff_\alpha \right]
    \\
    = &
    \cutoff_\alpha \cdot (1-\Pr\left[\val \leq \cutoff_\alpha \right]^{\addedBuyers})
    \\
    \geq &
    \cutoff_\alpha \cdot (1-(1-\alpha)^{\addedBuyers})
    \\
    \geq &
    \cutoff_\alpha \cdot (1-e^{-\alpha \cdot \addedBuyers})
    \\
    \end{align*}
\end{proof}

The following lemma is a specified analog to Lemma~\ref{lem:addbuyersAlmost} that bounds the contribution from values in the core using the second price auction, and its proof is similar to the proof of Lemma~\ref{lem:addbuyersAlmost}.
Let $\REV_{\VCG}$ denote the revenue from the second price auction.
\begin{lemma} \label{lem:VCGAlmost}
    Let $\cutoff_\alpha$ satisfy $\Pr_{\val \gets \dist}\left[\val >\cutoff_\alpha \right] \geq \alpha$.
    For any $\addedBuyers \geq 1$,
    $$
    \REV_{\VCG}\left(\dist^{2\addedBuyers} \right)
    \geq
    (1-e^{-\alpha\cdot \addedBuyers})^2 \cdot
    \cutoff_\alpha
    $$
\end{lemma}
\begin{proof}
    \begin{align*}
    \REV_{\VCG}(\dist^{2\addedBuyers})
    \geq &
    \cutoff_\alpha \cdot \Pr\left[\exists i, j \in [2\addedBuyers], i \neq j : \val^i, \val^j > \cutoff_\alpha \right] \\
    \geq &
    \cutoff_\alpha \cdot \Pr\left[\max_{i \in [\addedBuyers]} \val^i > \cutoff_{\alpha}, \max_{j \in [\addedBuyers]} \val^j > \cutoff_\alpha \right] \\
    = &
    \cutoff_\alpha \cdot \left(1-\Pr\left[\val \leq \cutoff_\alpha \right]^{\addedBuyers}\right)^2 \\
    = &
    \cutoff_\alpha \cdot \left(1-(1-\Pr\left[v > \cutoff_\alpha \right])^{\addedBuyers}\right)^2 \\
    \geq &
    \cutoff_\alpha \cdot \left(1-(1-\alpha)^{\addedBuyers}\right)^2 \\
    \geq &
    \cutoff_\alpha\cdot (1-e^{-\alpha \cdot \addedBuyers})^2 \\
    \end{align*}
\end{proof}

\subsubsection{Regular distributions.} \label{sec:regularDist}
For ease of exposition, we assume $\dist$ has a density $\dens$ and is strictly increasing.
Nevertheless, our results extend to arbitrary regular distributions.
A distribution $\dist$ with density $\dens$ is {\em regular} if $\val - \left(1-\dist(\val)\right)/\dens(\val)$ is non-decreasing in $\val$.
Our analysis is done in quantile space, which is defined below.
\begin{definition} (Quantiles, demand curve, and revenue curve.)
    \begin{itemize}
        \item
        Let $\q(\val) = \Pr_{x \gets \dist}\left[x \leq \val \right] = \dist(\val)$ be the {\em quantile} of $\val$.
        \item
        Let $\demCurve(\q) = \val$ for which $\dist(\val) = \q$, i.e.,
        $\demCurve(\q) = \dist^{-1}(\q)$ is the {\em demand curve} of $\dist$.
        \item
        Let $\revCurve(\q) = (1-\q) \cdot \demCurve(\q)$, i.e., the revenue from the posted price $\demCurve(\q)$ that sells w.p. $1-\q$.
    \end{itemize}
\end{definition}
Observe that $\demCurve$ is increasing in $\q$.
By change of variables, the expected value $\val$ can be computed as follows:
\footnote{By $v = \demCurve(\q)$, we get $dv = \demCurve'(\q) d\q = [\dist^{-1}]'(\q) = \frac{1}{\dens(\dist^{-1}(\q))}= \frac{1}{\dens(\demCurve(\q))}$,
    so $\val \dens(\val)d\val = \demCurve(\q)\frac{\dens(\demCurve(\q))}{\dens(\demCurve(\q))}d\q$}
$\expect{}{\val} = \int_{0}^{\infty}{\val \dens(\val)d\val} = \int_{0}^{1}{\demCurve(\q)d\q} = \expect{\q\gets U[0, 1]}{\demCurve(\q)}$.
The following lemma is a well known characterization of regular distributions.
\begin{lemma} \label{lem:regularIsConcave} \cite{Myerson81}
    A distribution $\dist$ is regular if and only if $\revCurve(\q)$ is concave, i.e.,
    for every $\alpha, \beta, \gamma \in [0,1]$, it holds that
    $\revCurve(\gamma  \cdot \alpha + (1-\gamma)\beta) \geq \gamma \cdot \revCurve(\alpha) + (1-\gamma) \cdot \revCurve(\beta)$.
\end{lemma}
The following corollary is the only way we use regularity (and is only used to prove Lemma~\ref{lem:tailIsSmallRegular}).
\begin{corollary}\label{cor:regularIsSublinear}
    If $\dist$ is regular and $0 \leq \beta \leq \alpha \leq 1$, then $\revCurve(\beta) \geq \alpha \cdot \revCurve(\beta) \geq \beta \cdot \revCurve(\alpha)$.
\end{corollary}
\begin{proof}
    $\revCurve(\beta) = \revCurve(\frac{\beta}{\alpha}\cdot \alpha + (1-\frac{\beta}{\alpha}) \cdot 0)
    \geq
    \frac{\beta}{\alpha}\cdot \revCurve(\alpha) + (1-\frac{\beta}{\alpha}) \cdot \revCurve(0)
    \geq
    \frac{\beta}{\alpha}\cdot \revCurve(\alpha)
    $
\end{proof}

The following lemma shows that if the probability of exceeding the cutoff $\cutoff$ is $\varepsilon$, then $\cutoff$ is a price that provides a $1-\varepsilon$ approximation to the optimal revenue from the random variable $\val \cdot \ind\left[\val > \cutoff \right]$.
\begin{lemma}\label{lem:tailIsSmallRegular}
    For a cutoff $\cutoff = \demCurve(1-\varepsilon)$,
    it holds that
    $
    \revCurve(1-\varepsilon) \geq  (1-\varepsilon) \cdot \REV(\val\cdot \ind\left[\val > \cutoff \right])
    $.
\end{lemma}
\begin{proof}
    Since $\val \cdot \ind\left[\val >\cutoff\right]$ is single-dimensional, there exists a price $\pi \geq \cutoff$ that achieves the optimal expected revenue for this random variable. Let $\beta$ satisfy $\pi = \demCurve(1-\beta)$,
    then we have
    $\REV(\val \cdot \ind\left[\val >\cutoff\right] )
    =
    \revCurve(1-\beta)$.
    Since $\pi \geq \cutoff$, we get that $\demCurve(1-\beta) \geq \cutoff = \demCurve(1-\varepsilon)$.
    By monotonicity of $\demCurve$ we have $1-\beta \geq 1- \varepsilon$, and by
    Corollary~\ref{cor:regularIsSublinear}
    $
    \revCurve(1-\varepsilon)
    \geq
    (1-\varepsilon)\revCurve(1-\beta)
    $,
    as required.
\end{proof}

The following lemma relates the contribution from the core to the contribution from the tail.
\begin{lemma} \label{lem:CoreToTailRegular}
    For a cutoff $\cutoff = \demCurve(1-\varepsilon_1)$,
    for every $0 < \varepsilon_2 < 1$ and $k \in \naturals$
    such that $\varepsilon_1 \cdot \varepsilon_2^{-k} < 1$ it holds that:
    $$
    \expect{\val \gets \dist}{\val \cdot \ind\left[\val \leq \cutoff \right]}
    \geq
    k \cdot (1 - \varepsilon_2) \cdot (1- \varepsilon_1 \cdot \varepsilon_2^{-k}) \REV(\val\cdot \ind\left[\val > \cutoff \right])$$
\end{lemma}
\begin{proof} First observe that
    $
    \expect{\val \gets \dist}{\val \cdot \ind\left[\val \leq \cutoff \right]}
    =
    \expect{\q\gets U[0, 1]}{\demCurve(\q)\cdot \ind\left[\q \leq 1-\varepsilon_1 \right]}
    $. By segmenting the expectation over $[0, 1-\varepsilon_1]$ we get
    \begin{align*}
    \expect{\q\gets U[0, 1]}{\demCurve(\q)\cdot \ind\left[\q \leq 1-\varepsilon_1 \right]}
    \geq &
    \sum_{\ell=0}^{k-1}\expect{}{\demCurve(\q) \cdot \ind\left[  1-\varepsilon_1 \cdot \varepsilon_2^{-(\ell+1)} < \q \leq 1- \varepsilon \cdot \varepsilon_2^{-\ell}\right]} \\
    \geq &
    \sum_{\ell=0}^{k-1}{\left(\varepsilon_1 \cdot \varepsilon_2^{-(\ell+1)} - \varepsilon_1 \cdot \varepsilon_2^{-\ell}\right)\demCurve(1-\varepsilon_1 \cdot \varepsilon_2^{-(\ell+1)}) } \\
    = &
    (1 - \varepsilon_2)\sum_{\ell=0}^{k-1}{\varepsilon_1 \cdot \varepsilon_2^{-(\ell+1)} \cdot \demCurve(1-\varepsilon_1 \cdot \varepsilon_2^{-(\ell+1)}) } \\
    = &
    (1 - \varepsilon_2)\sum_{\ell=1}^{k}{\revCurve(1-\varepsilon_1 \cdot \varepsilon_2^{-\ell})}
    \end{align*}
    In the first inequality we reduce values in $[0, 1-\varepsilon_1 \cdot \varepsilon_2^k]$ to $0$.
    The second inequality follows by monotonicity of $\demCurve$.
    The fact that $\varepsilon_1 \cdot \varepsilon_2^{-k} \geq \varepsilon_1 \cdot \varepsilon_2^{-\ell}\geq \varepsilon_1$ for every $\ell$, combined with  Lemma~\ref{lem:tailIsSmallRegular} implies that
    $\revCurve(1-\varepsilon_1 \cdot \varepsilon_2^{-k}) \geq (1-\varepsilon_1 \cdot \varepsilon_2^{-\ell}) \REV(\val\cdot \ind\left[\val > \cutoff \right])$, which completes the proof:
    $$
    \expect{\val \gets \dist}{\val \cdot \ind\left[\val \leq \cutoff \right]}
    \geq
    (1 - \varepsilon_2) \cdot k \cdot (1-\varepsilon_1 \cdot \varepsilon_2^{-k}) \REV(\val\cdot \ind\left[\val > \cutoff \right]).
    $$
\end{proof}

\subsubsection{Mechanisms with restricted probability of sale.} \label{sec:mechRestPr}

We describe optimal single item mechanisms with restricted
probability of sale. The lemmas established in this section are used in Section~\ref{SEC:MOREBUYERS}.
In particular, we formally show that a naive generalization of Myerson's optimal auction for the unrestricted case is also optimal for the restricted case.

Fix a number $p \in [0,1]$ and $x$ i.i.d. buyers, each buyer $i$ with value $\vals^i$ drawn from a single dimensional distribution $\dist$.
A mechanism with allocation function $\alloc$ that satisfies
$
\expect{\vals}{\sum_{i \in [x]} \alloc^i(\vals) } \leq p,
$
is said to {\em sell with ex-ante probability of at most $p$}.

Recall that by Myerson's theory, a mechanism with allocation rule $\alloc$ and payment $\pay$ is truthful if and only if each buyer $i$'s allocation rule $\alloc^i$ is monotone,
and buyer $i$'s payment function is
fully determined by $\alloc^i$ via the equality $\pay^i(\vals) = \vals^i \cdot \alloc^i(\vals) - \int_{z=0}^{\vals^i}\alloc^i(z, \vals^{-i})dz$.
Furthermore,
Myerson defines {\em virtual value} functions%
\footnote{When $\dist$ has a density $\dens$, its virtual value function is defined by $\virtVal(\val) = \val - \frac{1-\dist(\val)}{\dens(\val)}$.
See e.g., \cite{elkind2007designing, cai2016duality} for the general case.} $\virtVal^\dist : \supp(\dist) \rightarrow \reals$ for single dimensional distributions $\dist$ (we henceforth drop the superscript $\dist$).
The virtual value function is particularly useful because one can re-amortize the expected payment to a term that uses the virtual value functions, and this term can be optimized point wise.
This becomes apparent in Myerson's payment identity:
\begin{lemma} \cite{Myerson81} (Payment identity)
    In every truthful mechanism with allocation function $\alloc$ and payment function $\pay$, for every buyer $i$ it holds that
    $\expect{}{\virtVal(\vals^i)\cdot \alloc^i(\vals)}  = \expect{}{\pay^i(\vals)}$
\end{lemma}
A regular distribution is a distribution whose corresponding virtual value function $\virtVal$ is monotone.
For irregular distributions, Myerson defines yet another transformation that transforms the virtual value function to an {\em ironed} virtual value function $\bar{\virtVal}$
(the details are not so important for our application; for more details see e.g.,~\cite{Myerson81,hartline2013mechanism}).
The key point is that ironed virtual value functions are always monotone, and furthermore, maintain the following property:

\begin{theorem} \label{thm:paymentToSurplus}
    (\cite[Theorem 3.12]{hartline2013mechanism}) For every truthful mechanism with allocation function $\alloc$, for every buyer $i$,
    $
    \expect{}{\virtVal(\vals^i)\cdot \alloc^i(\vals)}
    \leq
    \expect{}{\bar{\virtVal}(\vals^i)\cdot \alloc^i(\vals)},
    $
    with equality if $\alloc^i(\cdot, \vals^{-i})$ is constant on ironed intervals.
\end{theorem}

In order to construct the optimal mechanism with restricted probability of sale, we will need the following definitions.
Let $\vals^* = \max_{i \in [x]} \vals^i$, and
let
$S = \{y \geq 0 :  \Pr[\bar{\virtVal}(\vals^*) \geq y ] \leq p \}$
and
$
\phi = \inf S.
$

Note that if $\phi \in S$, then it must be that $\phi =0$ or that $\Pr[\bar{\virtVal}(\vals^*) \geq \phi ] = p$.
 \footnote{Otherwise $\Pr[\bar{\virtVal}(\vals^*) \geq \phi ] = p - \epsilon$ but for every $0 < a< \phi$ we have $\Pr[\phi > \bar{\virtVal}(\vals^*) \geq a ] > \epsilon$.}
Otherwise, it must be that $\Pr[\bar{\virtVal}(\vals^*) > \phi ] \leq p$,
\footnote{Otherwise $\Pr[\bar{\virtVal}(\vals^*) > \phi ] = p + \epsilon$ but for every $a > \phi$ we have $\Pr[a > \bar{\virtVal}(\vals^*) > \phi ] \geq \epsilon$.}
but
$\Pr[\bar{\virtVal}(\vals^*) \geq \phi ] > p $,
i.e,
$
\Pr[\bar{\virtVal}(\vals^*) = \phi ] = \Pr[\bar{\virtVal}(\vals^*) \geq \phi ] - \Pr[\bar{\virtVal}(\vals^*) > \phi ] > 0.
$

Consider the mechanism $\M$ that sells a single item to $x$ buyers as follows.
If $\phi \in S$, the mechanism allocates the item to a random buyer $i$ among those with $\bar{\virtVal}(\vals^i) = \bar{\virtVal}(\vals^*)$
whenever $\bar{\virtVal}(\vals^*) \geq \phi$, and otherwise does not allocate.

If $\phi  \not \in S$, then the mechanism allocates the item to a random buyer $i$ among those with $\bar{\virtVal}(\vals^i) = \bar{\virtVal}(\vals^*)$
whenever $\bar{\virtVal}(\vals^*) > \phi$,
and whenever $\bar{\virtVal}(\vals^*) = \phi$ the mechanism draws a Bernoulli random variable that equals $1$ w.p. $\frac{p - \Pr[\bar{\virtVal}(\vals^*) > \phi]}{\Pr[\bar{\virtVal}(\vals^*) = \phi]}$,
and if the variable equals $1$ the mechanism allocates to a random buyer $i$ among those with
$\bar{\virtVal}(\vals^i) = \phi$.

\begin{lemma}
    \label{lem:OptRevWithRest}
    For $x$ i.i.d. buyers drawn from a single dimensional distribution $\dist$, and $p, \phi, \vals^*$ as above, the mechanism $\M$ described above sells with ex-ante probability of at most $p$, and extracts expected revenue of:
    \begin{align*}
        \expect{}{\bar{\virtVal}(\vals^*) \cdot \ind[ \bar{\virtVal}(\vals^*) > \phi  ] } + \phi \cdot \left(p - \Pr[\bar{\virtVal}(\vals^*) > \phi]\right).
    \end{align*}
    Moreover, every buyer contributes the same amount of expected revenue.
\end{lemma}
\begin{proof}
    By definition of the mechanism $\M$,
    all the allocation rules depend on values of $\bar{\virtVal}$, which implies that they are constant on ironed intervals,
    and since buyers are drawn from the same distribution, every buyer $i$ contributes the same amount of expected ironed virtual surplus $\expect{}{\bar{\virtVal}(\vals^i)\cdot \alloc^i(\vals)},$
    which, by Theorem~\ref{thm:paymentToSurplus} and Myerson's payment identity,
    \footnote{For more details, see e.g. \cite[Equation~(3.3)]{hartline2013mechanism}.}
     is also the contribution to the mechanism's expected revenue.

    Whenever $\phi \in S$,
    we saw that $\phi = 0$ or $\Pr[\bar{\virtVal}(\vals^*) \geq \phi ] = p$.
    The mechanism's ironed virtual surplus in this case is exactly
    $$
        \expect{}{\bar{\virtVal}(\vals^*) \cdot \ind[ \bar{\virtVal}(\vals^*) \geq \phi  ] }
        =
        \expect{}{\bar{\virtVal}(\vals^*) \cdot \ind[ \bar{\virtVal}(\vals^*) > \phi  ] }  +  \phi \cdot \Pr[\bar{\virtVal}(\vals^*) = \phi].
    $$
    If $\phi = 0$, the statement holds, and  if $\Pr[\bar{\virtVal}(\vals^*) \geq \phi ] = p$,
    then
    $$
    p - \Pr[\bar{\virtVal}(\vals^*) > \phi ] = \Pr[\bar{\virtVal}(\vals^*) \geq \phi ] - \Pr[\bar{\virtVal}(\vals^*) > \phi ]
    =
    \Pr[\bar{\virtVal}(\vals^*) = \phi ]
    $$
    as required.

    Whenever $\phi \not \in S$, the mechanism's ironed virtual surplus is exactly
    $$
    \expect{}{\bar{\virtVal}(\vals^*) \cdot \ind[ \bar{\virtVal}(\vals^*) > \phi  ] } + \Pr[\bar{\virtVal}(\vals^*) = \phi] \cdot \frac{p - \Pr[\bar{\virtVal}(\vals^*) > \phi]}{\Pr[\bar{\virtVal}(\vals^*) = \phi]} \cdot \phi.
    $$
    As required.
\end{proof}

We are now ready to prove the optimality of $\M$ among mechanisms with restricted probability of sale.
\begin{lemma} \label{lem:OptSymHighpay}
    For $\dist, x, p$ as above, the mechanism $\M$ above is a revenue maximizing mechanism among all truthful mechanisms that sell to $x$ i.i.d. buyers drawn from $\dist$ with ex-ante probability of at most $p$, i.e.,
    $
    \REV_p(\dist^{x}) = \REV_{\M}(\dist^{x}).
    $
\end{lemma}

\begin{proof}
    Let $\vals^i \gets \dist$ for all $i \in [x]$.
    Note that since the buyers are all drawn from $\dist$, they all have the same virtual function and ironed virtual function.

    Fix a truthful mechanism $\M'$ that sells w.p. at most $p$. Let $\alloc^i(\vals)$ be its allocation function for buyer $i$.
    By Myerson's theory, $\M'$ is truthful if and only if $\alloc^i$ is monotone w.r.t. $\vals^i$ for every buyer $i$.
    Clearly, if $\alloc^i(\vals) > 0$ for some $\virtVal(\vals^i) < 0$, setting  $\alloc^i(\vals) = 0$ whenever $\virtVal(\vals^i) < 0$
    does not increase the ex-ante probability of sale, and does not decrease the expected virtual surplus,
    and maintains monotonicity of $\alloc^i(\cdot, \vals^{-i})$.
    Therefore, we may consider w.l.o.g. only mechanisms with $\alloc^i(\vals) = 0$ for all $\virtVal(\vals^i) < 0$,
    i.e., mechanisms that count only non-negative virtual surplus.

    \begin{align*}
    \REV_{\M'}(\dist^{x})
     =
    \expect{\vals}{\sum_{i \in [x]} \virtVal(\vals^i) \cdot \alloc^i(\vals)}
     \leq
    \expect{\vals}{\sum_{i \in [x]} \bar{\virtVal}(\vals^i) \cdot \alloc^i(\vals)}
     \leq
    \expect{\vals}{\bar{\virtVal}(\vals^*) \cdot \sum_{i \in [x]}\alloc^i(\vals) }.
    \end{align*}
    Where the first equality is Myerson's payment identity,
    and the first inequality holds by Theorem~\ref{thm:paymentToSurplus}.
    Let $\hat{\alloc} = \expect{}{\sum_{i \in [x]}\alloc^i(\vals) }$, i.e., the total probability of sale of mechanism $\M'$.
    We know that $\hat{\alloc} \leq p$.
    Recall that by definition of $\phi$, it holds that $\Pr[\bar{\virtVal}(\vals^*)  > \phi ] \leq p$.
    Observe that
    $$
    \expect{}{\bar{\virtVal}(\vals^*) \cdot \sum_{i \in [x]}\alloc_i(\vals) }
    \leq
    \expect{}{\bar{\virtVal}(\vals^*) \cdot \ind[ \bar{\virtVal}(\vals^*) > \phi  ] } +  \phi \cdot \left(p - \Pr[\bar{\virtVal}(\vals^*)  > \phi ]\right),
    $$
    because the former counts non-negative values of $\bar{\virtVal}(\vals^*)$ that accumulate to at most $p$ mass, while the latter
    counts all the highest values of $\max \{\bar{\virtVal}(\vals^*), 0\}$, that accumulate to $p$ mass.
    The proof follows by Lemma~\ref{lem:OptRevWithRest}.
\end{proof}
\begin{corollary} \label{cor:RevSubmodular} Let $0 \leq \delta < 1$.
    Then
    $\REV_p(\dist^{x}) \leq \frac{1}{1-\delta} \REV_p(\dist^{(1-\delta)\cdot x})$
\end{corollary}
\begin{proof}
    Let $\M$ be the mechanism considered in Lemma~\ref{lem:OptRevWithRest} for $x$ buyers drawn i.i.d. from $\dist$.
    Let $\pay = \expect{}{\pay^i_{\M}(\vals)} $ be the expected payment of a buyer in $\M$.
    Fix some $\delta \in [0, 1]$, and
    consider the mechanism $\M_{\delta}$ that sells a single item to an arbitrary set of $(1-\delta) \cdot x$ buyers
    by sampling the remaining $\delta \cdot x$ buyers i.i.d. from $\dist$, and
    by executing $\M_{p}$ with the buyers' bids and the samples.
    Truthfulness of $\M_\delta$ follows by truthfulness of $\M$.
    The revenue from $\M_{\delta}$ is $(1-\delta) x \cdot \pay$, and its ex-ante probability of sale is at most $p$.
    Therefore,
    \begin{align*}
    \REV_{\M}(\dist^{x})
    =
    x \cdot \pay
    =
    (1-\delta) x \cdot \pay + \delta x \cdot \pay
    =
    \REV_{\M_\delta}(\dist^{(1-\delta) x}) + \delta \cdot \REV_{\M}(\dist^{x})
    \end{align*}
    Therefore by rearranging we get that
    \begin{align*}
    \REV_p(\dist^{x})
    =
    \REV_{\M}(\dist^{x})
    =
    \frac{1}{1-\delta} \REV_{\M_\delta}(\dist^{(1-\delta) x})
    \leq
    \frac{1}{1-\delta} \REV_{p}(\dist^{(1-\delta) x})
    \end{align*}
    As required.
\end{proof}

\begin{lemma} \label{lem:restMonotonicity}
    $
    \REV_p(\dist^{k})
    $
    is non-decreasing in $k$.
\end{lemma}
\begin{proof}
    Fix some $k > t$, and consider the mechanism $\M'$ for $k$ buyers that selects a set of $t$ buyers,
    and executes the mechanism $\M$ considered in Lemma~\ref{lem:OptRevWithRest} for $t$ buyers drawn i.i.d. from $\dist$, only on these $t$ buyers, and completely ignores the rest.
    By independence across buyers, $\M'$ extracts exactly the same revenue as $\M$, and sells with the same probability, i.e.,
    $\REV_{\M'}(\dist^{k}) = \REV_{\M}(\dist^{t}) =\REV_{p}(\dist^{t}).$
    Therefore,
    $
    \REV_{p}(\dist^{t}) = \REV_{\M'}(\dist^{k}) \leq \REV_{p}(\dist^{k})
    $
    as required.
\end{proof}

\begin{lemma}  \label{lem:restProbHighVals}
    Fix a cutoff $\cutoff$.
    Let $\vals \gets \dist^n$.
    Let $\hat{\dist}$ denote the distribution of $\vals^i \vert (\vals^i > \cutoff)$ for every buyer $i$.
    Then for every $k \leq n$ and $p \in [0, 1]$.
    \begin{align*}
    \REV_{p}\left(\hat{\dist}^{k}\right)
    \leq
    \frac{\REV_{p}(\dist^{n}) }{\Pr[ \cardinality{\{i : \vals^i > \cutoff \}} \geq k ] }
    \end{align*}
\end{lemma}
\begin{proof}
    Consider the following $n$ buyers mechanism $\M'$.
    The mechanism $\M'$ first receives bids, and in the first step removes all bids that are at most $\cutoff$, and remains with some $k' \leq n$ buyers.
    In the second step $\M'$ executes the mechanism $\M$ from Lemma~\ref{lem:OptRevWithRest} for $k'$ buyers drawn i.i.d. from $\hat{\dist}$.
    For every buyer, the participation rule (value more than $\cutoff$) is monotone, and every mechanism from Lemma~\ref{lem:OptRevWithRest} has a monotone allocation rule,
    therefore $\M'$ has a monotone allocation rule, and thus is truthful.
    Also, in the second step $\M'$ always executes a mechanism that sells with ex-ante probability at most $p$, hence $\M'$ also sells with ex-ante probability at most $p$.

    Let $\event' = \{i : \vals^i > \cutoff \}$.
    Denote by $\mathcal{R}'$ the revenue of $\M'$ where every $\vals^i$ is drawn from $\dist$, conditioned on the event $\cardinality{\event'} \geq k$.
    Conditioned on $\event'$, the revenue of the mechanism $\M'$ is by definition exactly that of $\M$ from Lemma~\ref{lem:OptRevWithRest} for $\cardinality{\event'}$ buyers drawn i.i.d. from $\hat{\dist}$,
    which, by Lemma~\ref{lem:OptSymHighpay} equals $\REV_p\left(\hat{\dist}^{\cardinality{\event'}}\right)$.
    Therefore by law of total expectation,
    \begin{align*}
    \mathcal{R}'
    =
    \expect{}{\REV_{\M'}(\vals)  \bigg\vert  k \leq  \cardinality{\event'}  }
    &=
    \expect{}{
        \expect{}{\REV_p\left(\hat{\dist}^{\cardinality{\event'}}\right) \bigg \vert \event', k \leq \cardinality{\event'} } \bigg \vert \event'
    }
    \geq
    \REV_p\left(\hat{\dist}^{k}\right)
    \end{align*}
    Where the inequality follows by monotonicity w.r.t. the number of buyers (Lemma~\ref{lem:restMonotonicity}).
Finally,
    observe that
    \begin{align*}
    \mathcal{R}'
     =
    \expect{\vals}{\REV_{\M'}(\vals) \bigg\vert  k \leq  \cardinality{\event'}  }
    \leq
    \frac{\expect{\vals}{\REV_{\M'}(\vals)}}
    {\Pr\left[ k \leq  \cardinality{\event'} \right]}
    =
    \frac{\REV_{\M'}(\dist^{n})}
    {\Pr\left[ k \leq  \cardinality{\event'} \right]}
    \leq
    \frac{\REV_p(\dist^{n}) }
    {\Pr\left[ k \leq  \cardinality{\event'} \right]}
    \end{align*}
    The last inequality follows by that $\M'$ sells with probability at most $p$.
    This completes the proof.
\end{proof}

\subsection{Partition of the support.}
\label{sec:coreTail}
Recall that for every buyer $i$ and item $j$, $\vals_j^i \gets \dist_j$, and that $\distM = \times_{j \in [m]} \dist_j$.
Every item $j$ is assigned a cutoff $\cutoff_j$.
If the value of an item is greater than its cutoff (i.e.,  $\vals_j^i > \cutoff_j$) then it is said to be in the tail w.r.t. $\vals^i$, and otherwise it is in the core w.r.t $\vals^i$.

For a set of items $A$, let $\distM^A$ be the distribution $\distM$ conditioned on $A$ being exactly the set of items in the tail.
Let $p_{A}$ be the probability of $A$ being
exactly the set of items in the tail, i.e.,
$p_{A} = \prod_{j \in A} \Pr\left[\vals^i_j > \cutoff_j \right] \cdot \prod_{j \not \in A } \Pr\left[\vals^i_j \leq \cutoff_j \right] $.
Let $\distM_C^A$ (resp., $\distM_T^A$) be the distribution $\distM^A$ restricted to just items in the core (resp., tail).

For every buyer $i$ consider some $\event^i \subseteq [m]$, and let
$\event = \{\event^i\}_{i\in [n]}$. Let $\event$ represent the
event that for every $i$ it holds that $\event^i$ is exactly the set
of items that are in the tail with respect to buyer $i$'s valuation, i.e., for every item $j\in \event^i$ it holds that
$\vals_j^{i} > \cutoff_j$, and for every $j \not \in \event^i$ it
holds that $\vals_j^{i} \leq \cutoff_j$.
Denote by $p_{\event}$ the
probability of event $\event$, i.e.,
$p_{\event}
= \prod_i p_{\event^i}$,
and let $\distM^{\event}$ be the distribution $\distM^n$ conditioned on the event $\event$,
i.e.,
$\distM^\event = \times_{i\in [n]}\distM^{\event^i}$.

The following lemma appeared previously and is included for completeness.
\begin{lemma} \label{lem:subdomainStiching}
    \cite{hart2012approximate} (Sub-domain stitching) Let $\events = \{\event : \forall i , \event^i \subseteq [m]  \}$ be
    all possible $\event$.
    Then
    $$
    \REV(\distM^n)
    \leq
    \sum_{\event \in \events}{p_{\event} \cdot \REV(\distM^{\event} )}
    $$
\end{lemma}
\begin{proof}
    Let $\M$ be the optimal mechanism for $n$ i.i.d. buyers with values $\vals$ drawn from $\distM^n$.
    Then
    $\REV_{\M}(\distM^{n})
    =
    \sum_{\event \in \events}{p_{\event} \cdot \REV_{\M}(\distM^{\event})}
    $.
    Taking the optimal mechanism for $\distM^{\event}$ only increases the right-most sum.
\end{proof}

\section{Many Items: $\SREV$ with $O(\lowercase{n \cdot \log \frac{m}{n} } )$ Buyers}
% !TeX root = main99revenue.tex
\label{SEC:SREV}

In this section we prove that for $m \gg n$, increasing the number of buyers by a factor of $O(\log(m/n))$ suffices to recover 99\% of the optimal revenue by selling the items separately.
\begin{theorem}[Theorem~\ref{thm:SREV-many-buyers}, case $m \gg n$] \label{thm:ManyAdditivebuyers}
    For any constant $\varepsilon > 0 $ there exists a constant $\delta = \delta(\varepsilon)>0$ such that whenever $m \geq {2n}/{\delta}$,
    $$
    (1-\varepsilon)\REV(\distM^{n}) \leq \SREV\left( \distM^{\frac{n\cdot \log \frac{m}{n} }{\delta }} \right)
    $$
\end{theorem}
This implies the following corollary for the special case of a single buyer:
\begin{corollary} [Theorem~\ref{thm:SREV-one-buyer}]\label{cor:OneAdditivebuyer}
    For any constant $\varepsilon > 0 $ there exists a constant $\delta(\varepsilon)>0$ such that whenever $m \geq 2 / \delta$:
    $$
    (1-\varepsilon)\REV(\distM) \leq \SREV\left( \distM^{\frac{\log m}{\delta }} \right)
    $$
\end{corollary}

\subsubsection*{Proof outline}
We first partition the domain of the valuations distribution into sub-domains (using the Subdomain Stiching Lemma (Lemma~\ref{lem:subdomainStiching})). We henceforth condition on the event (or sub-domain) $\event$ that describes which items are in the tail or core for which buyer (we still don't know their values within the tail/core).

We now describe a ``marginal mechanism lemma''
(Lemma~\ref{lem:marginalSub}) that shows that the optimal revenue
from $\distM^\event$, is bounded by the sum of expected item values
in the core in event $\event$, plus the revenue from selling to each
buyer the items that are in her tail in event $\event$ (conditioned
on them being in the tail). This implies a core-tail decomposition
lemma (Lemma~\ref{lem:coreDecompositionManyBuyers}) that bounds the
optimal revenue using two terms: the core and tail.

We bound the tail in Section~\ref{subsec:SREV:tail}.
In Lemma~\ref{lem:tailToSeperate} we show that the tail is bounded by the revenue from selling items separately and using only prices that are higher than the cutoffs.
Since the probability that two buyers are interested in the same item at a high price is extremely low, the supply constraint is hardly restricting. Therefore when we increase the number of buyers, the revenue increases almost linearly.
This is used in Lemma~\ref{lem:TailUpperBound} that shows that the tail is a tiny fraction of the revenue after multiplying the number of buyers by a (large) constant.

We bound the core in Section~\ref{subsec:SREV:core}.
We separate the core into regions of lower values and higher values.
Lower values are upper bounded by a maximum of two cutoffs, and are handled in two lemmata (\ref{lem:lowerValsManyBuyers}~and~\ref{lem:lowerValsManyBuyers:1}).

In Lemma~\ref{lem:HigherQuantileBoundManyBuyers} we bound the contribution from ``higher'' values (still in the core).
First, we observe that they lie in a bounded region: bounded from above since they are in the core, and bounded from below since we handle low values separately. We use this property to show that their expected value is at most $O(\log (m/n))$-factor larger than the revenue from selling each item separately at ``higher'' prices.
Now, similarly to the tail, since at high prices buyers are likely to want disjoint sets of items, the above revenue scales almost linearly when increasing the number of buyers.
Therefore, multiplying the number of buyers by $O(\log (m/n))$ suffices to extract revenue much larger than the contribution from high values.

Finally, we complete the proof in
Section~\ref{sub:ManyAdditivebuyers} by combining the core-tail
decomposition with the above arguments.

\subsection{Core-tail decomposition.}
The following lemma may be seen as a generalization of Lemma 24 in~\cite{hart2012approximate} for $n > 1$ buyers.
In~\cite{hart2012approximate}, the authors note that the proof of their Lemma 24 implies a generalization for $n >1$, by considering for each item the maximal value among all buyers' values.
This method still allows to partition the support of $\distM^n$ to only $2^{m}$ subsets (according to the set of items that have a maximal value above the cutoff).
We show a more fine-grained generalization that partitions the support of $\distM^n$ to $2^{m \cdot n}$ subsets,
thereby leading to a tighter analysis.

Recall that 
for some $\event^i$, $\distM^{\event^i}_T$ is  the distribution $\distM$ conditioned on $\event^i$ being exactly the set of items in the tail, restricted only to the items in the tail ($\event^i$).
\begin{lemma} \label{lem:marginalSub}
    (Marginal Mechanism on Sub-Domain) For $\event$ as above,
    $$
    \REV\left(\distM^{\event}\right)
    \leq
    \sum_{j\in [m]}{
        \expect{\vals \gets \distM^{n}}{
            \max_{i : j \not \in \event^{i}}\vals^{i}_j \bigg \vert \event }
    }
+
    \sum_{i\in [n]}{
        \REV\left(\distM^{\event^i}_{T}\right)
    }
    $$
\end{lemma}
\begin{proof}
    Let $\M$ be an optimal $n$ buyer mechanism for buyers that are drawn from $\distM^{\event}$, i.e.,
    \begin{align} \label{eq:marginalSub:opt}
    \REV_{\M}(\distM^{\event})
    =
    \REV\left(\distM^{\event}\right).
    \end{align}
    Let $\pay^i(\cdot )$ be buyer $i$'s payment function in mechanism $\M$,
    and let $\alloc_j^i(\cdot)$ be buyer $i$'s allocation function for item $j$ in mechanism $\M$.
    Consider the following one buyer mechanism $\M^i$
    for buyer $i$, with valuation $\vals^i \gets \distM^{\event^i}$.
    For each $\ell \neq i$ sample (privately) $\vals^{\ell} \gets \distM^{\event^{\ell}}$,
    and execute $\M$ with the buyer $i$'s submitted bids and the samples $\vals^{-i} = \{\vals^\ell\}_{\ell \neq i}$.
    Truthfulness of $\M^{i}$ follows by that $\M$ is truthful.
    By independence across buyers we have that
    \begin{align} \label{eq:marginalSub:oneBuyerRev}
    \REV_{\M^i}(\distM^{\event^{i}})
    =
    \expect{\vals^{-i}}{
        \expect{\vals^i}{\pay^i(\vals^i, \vals^{-i}) \bigg\vert \vals^{-i} } }
    =
    \expect{\vals \gets \distM^{\event}}{\pay^i(\vals)}
    .
    \end{align}
    By linearity of expectation and Equation~(\ref{eq:marginalSub:opt}) we have that:
    \begin{align} \label{eq:marginalSub:sumBuyers}
    \sum_{i \in [n]}{\REV_{\M_i}(\distM^{\event^{i}})}
    =
    \expect{\vals \gets \distM^{\event}}{\sum_{i\in [n]}{\pay^i(\vals)}}
    =
    \REV\left(\distM^{\event}\right).
    \end{align}
    Fix some $i \in [n]$, and some $\event^i \subseteq [m]$.

    Consider the following single buyer mechanism $\M(\event^i)$ for selling only items in $\event^i$:
    first (publicly) sample values $\vals^{i}_{\comp{\event^i}} = \{\val^i_{j}\}_{j \in \comp{\event^i}}$ according to the distribution $\distM^{\event^i}_C$.
    Then solicit from the buyer values for the items $\vals^i_{\event^i} = \{\val^i_{j}\}_{j \in \event^i}$.
    Finally, feed mechanism $\M^i$ the valuation $\vals^i = (\vals^i_{\event^i}, \vals^i_{\comp{\event^i}})$.
    Let $\alloc_j^i$ be the probability that $\M^i$ allocates item $j$ to the buyer, and let $\pay^i$ be the buyer's payment.
    Allocate each item $j \in \event^i$ with probability $\alloc^{i}_j$,
    and charge a payment of $\pay^i - \sum_{j\in \comp{\event^i}}{\alloc^{i}_j\cdot \vals^i_j}$.
    Truthfulness of $\M(\event^i)$ follows by that $\M^{i}$ is truthful.

    Conditioned on a fixed private sample $\vals^{-i}$ of mechanism $\M^i$, the revenue of $\M(\event^{i})$
        is the expected payment over the samples $\vals^{i}_{\comp{\event^i}} \gets \distM^{\event^i}_C$ and buyer $i$'s private values
        $\vals^i_{\event^i}\gets \distM^{\event^i}_T$ for items in $\event^i$.
        \begin{align*}
        \expect{\vals^i_{\comp{\event^i}}}{
            \expect{\vals^i_{\event^i}}{\pay^i(\vals^i, \vals^{-i}) - \sum_{j\in \comp{\event^i}}{\alloc^{i}_j(\vals^i, \vals^{-i}) \cdot \vals^i_j}} \bigg\vert \vals^i_{\comp{\event^i}}}
        =
        \expect{\vals^i \gets \distM^{\event^i}}{
            \pay^i(\vals) - \sum_{j\in \comp{\event^i}}{\alloc^{i}_j(\vals)\cdot \vals^i_j}},
        \end{align*}
        where the equality follows by independence across items.
        Therefore, by taking expectation over $\vals^{-i} \gets \times_{\ell \neq i}\distM^{\event^{\ell}}$, we get that the expected revenue of $\M(\event^{i})$ from a buyer $\vals^i \gets \distM^{\event^i}$ is
         \begin{align*}
         \REV_{\M(\event^{i})}(\distM^{\event^i}_T)
         & =
         \expect{\vals^{-i}}{
                \expect{\vals^i}{
                    \pay^i(\vals) - \sum_{j\in \comp{\event^i}}{\alloc^{i}_j(\vals)\cdot \vals^i_j}} \bigg\vert \vals^{-i}
                }
         & =
         \expect{\vals \gets \distM^{\event}}{\pay^i(\vals) - \sum_{j\in \comp{\event^i}}{\alloc^{i}_j(\vals)\cdot \vals^i_j}},
         \end{align*}
         where the last equality follows by independence across buyers.
    Therefore by linearity of expectation and Equation~(\ref{eq:marginalSub:oneBuyerRev}) we get that
    $$
    \REV_{\M(\event^{i})}(\distM^{\event^i}_T) =  \REV_{\M^i}(\distM^{\event^i})    -
    \expect{\vals \gets \distM^{\event}}{\sum_{j\in \comp{\event^i}}{\alloc^{i}_j(\vals)\cdot \vals^i_j}}.
    $$
    Summing over all $i \in [n]$, by Equation~(\ref{eq:marginalSub:sumBuyers}) we get that
    $$
    \sum_{i\in [n]}{\REV_{\M(\event^{i})}(\distM^{\event^i}_T)}
    =
    \REV\left(\distM^{\event}\right) -
    \expect{\vals \gets \distM^{\event}}{\sum_{i \in [n]}{\sum_{j\in \comp{\event^i}}{\alloc^{i}_j(\vals)\cdot \vals^i_j}}}.
    $$
    For each item $j$ and any buyers valuation profile $\vals$ it holds that  $\sum_{i}{\alloc_j^i(\vals)} \leq 1$.
    Therefore by reorganizing the above and replacing the revenue of mechanism $\M(\event^{i})$ with the optimal mechanism for the distribution $\distM^{\event^i}_T$ we get that
    $$
    \REV\left(\distM^{\event}\right)
    \leq
    \sum_{j \in [m]}{\expect{\vals \gets \distM^{\event}}{\max_{i : j \in \comp{\event^i}}{\vals^i_j}}}
    +
    \sum_{i\in [n]}{\REV(\distM^{\event^i}_T)},
    $$
    as required.
\end{proof}

The following core-tail decomposition upper bounds the optimal
revenue from $n$ buyers. In contrast to the decomposition in
\cite{babaioff2014simple}, our decomposition provides a tail that
relates to single buyer settings.

To simplify notation, once cutoffs $\cutoff_j$ are fixed for every $j$,
we use $\distMtail$ to denote the product distribution of the random variables $\{\val_j \cdot \ind\left[\val_j >\cutoff_j \right]\}_{j \in [m]}$,
and $\distMcore$ to denote the product distribution of the random variables
$\{\val_j \cdot \ind\left[\val_j \leq \cutoff_j \right]\}_{j \in [m]}$, where $\val_j \gets \dist_j$ for every $j$.
\begin{lemma} (Core-tail decomposition) \label{lem:coreDecompositionManyBuyers}
    $$
    \REV(\distM^{n})
    \leq
    \VAL(\distMcore^{n})
    +
    n \cdot \sum_{A \subseteq [m]}{p_A \cdot \REV(\distM^A_T)}
    $$
\end{lemma}
\begin{proof}
    By Lemma~\ref{lem:subdomainStiching} and Lemma~\ref{lem:marginalSub} We immediately get:
    \begin{align} \label{eq:coreDecompositionManyBuyers:1}
    \REV(\distM^{n})
    \leq
    \sum_{\event \in \events}{p_{\event}{
        \sum_{j\in [m]}{
                \expect{\vals \gets \distM^{n}}{
                    \max_{i : j \not \in \event^{i}}\val^{i}_j  \bigg \vert \event}
            }
        }
    }
    +
    \sum_{\event \in \events}{p_{\event}{   \sum_{i\in [n]}{
                \REV\left(\distM^{\event^i}_{T}\right)
            }
        }
    }
    .
    \end{align}
    Observe that by reordering of the right-most summations and independence across buyers, we get that, 
    $$
    \sum_{\event \in \events}{p_{\event}{   \sum_{i\in [n]}{
                \REV\left(\distM^{\event^i}_{T}\right)
            }
        }
    }
    =
    \sum_{i\in [n]}{\sum_{A \subseteq [m]}{p_{A} \cdot \REV(\distM^{A}_{T})
            }
        }
    =
    n \cdot \sum_{A \subseteq [m]}{p_{A} \cdot \REV(\distM^{A}_{T})
    }.
    $$

        Observe also that
        $\max_{i : j \not \in \event^{i}}\vals^{i}_j
    =
    \max_{i \in [n]}\lbrace\vals^{i}_j\cdot \ind\left[\vals^i_j \leq \cutoff_j\right] \rbrace,
    $
    therefore, by the law of total expectation, the left-most summations in Equation~(\ref{eq:coreDecompositionManyBuyers:1}) equals (after reordering)
    $$
    \sum_{j\in [m]}{\expect{\vals \gets \distM^n}{
                            \max_{i \in [n]}\{\val^{i}_j\cdot \ind\left[\vals^i_j \leq \cutoff_j\right] \} }
            }
    =
    \VAL(\distMcore^{n}).
    $$
\end{proof}

\paragraph{Cutoff setting.}
For every item $j$ fix $\cutoff_j$ to be so that
$\Pr_{\vals_j \gets \dist_j}\left[\vals_j > \cutoff_j\right]
=
m^{-1}
.$

\subsection{Tail}
\label{subsec:SREV:tail}
The main lemma of this section is
Lemma~\ref{lem:TailUpperBound} which states that the contribution from the tail is a tiny
fraction of the revenue from $\SREV$ with $n \cdot C$ buyers where $C$ is some constant.
Lemma~\ref{lem:tailToSeperate} shows (in the spirit of Proposition~1 in
\cite{babaioff2014simple}) 
that we can approximately recover the tail of one buyer ($\sum_{A \subseteq [m]}{p_A \cdot {\REV(\distM^A_T)}}$) by selling items separately, only using prices that are higher than the cutoffs. 
Recall that $\distMtail$ is the product
distribution of the random variables $\{\vals_j \cdot \ind
\left[\vals_j > \cutoff_j \right]\}_{j \in [m]}$, and therefore
$\SREV(\distMtail) = \sum_{j \in [m]}\REV(\vals_j \cdot \ind\left[\vals_j >\cutoff_j \right] )$.
\begin{lemma} \label{lem:tailToSeperate}
    $
    \sum_{A \subseteq [m]}{p_A \cdot {\REV(\distM^A_T)}}
    \leq
    2 \cdot
    \SREV(\distMtail)
    $
\end{lemma}
\begin{proof}
%Note that $\distM^A_T$ describes a single buyer setting.
The following holds:
    \begin{align*}
    \sum_{A}{p_A}{\REV(\distM^A_T)}
    \leq &
    \sum_{A}{p_A\cdot \cardinality{A} \cdot \SREV(\distM^A_T)} && \text{(Lemma~\ref{lem:revSrev})}\\
    = &
    \sum_{A}{p_A\cdot \cardinality{A} \cdot \sum_{j\in A}\REV(\vals_j | \vals_j >\cutoff_j)} && \text{(definition of $\SREV$)}\\
    = &
    \sum_{A}{p_A\cdot \cardinality{A} \cdot \sum_{j\in A}\frac{\REV(\vals_j \cdot \ind\left[\vals_j >\cutoff_j\right] )}{\Pr\left[\vals_j > \cutoff_j\right]}} && \text{(Lemma~\ref{lem:condRevIndRev})}\\
    = &
    \sum_{A}{\sum_{j\in A} \cardinality{A} \cdot \frac{p_A}{\Pr\left[\vals_j > \cutoff_j\right]} \cdot \REV(\vals_j \cdot \ind\left[\vals_j >\cutoff_j\right] )} \\
    = &
    \underbrace{\sum_{j\in [m]}\REV(\vals_j \cdot \ind\left[\vals_j >\cutoff_j\right] ) }_{= \SREV(\distMtail)}
    \cdot
    \underbrace{{\sum_{A: A\ni j}{\cardinality{A} \cdot \frac{p_A}{\Pr\left[\vals_j > \cutoff_j \right]}}}}_{
        \leq
        2
    } \\
    \end{align*}

    Observe that the rightmost sum is the expected size of the set in the tail, conditioned on $j$ being in the tail, i.e.,
    $
    1+\sum_{k \neq j}{\Pr[\vals_k > \cutoff_k]}
    =
    1+ \frac{m-1}{m}
    $ 
    (by definition of $\cutoff_k$)
    .
    The result follows by the definition of $\distMtail$ and $\SREV$'s additivity across items.
\end{proof}

We are now ready to show that $\SREV$ with
$n \cdot C$  buyers (where $C$ is a constant)
extracts {\em much more} revenue than the contribution from the tail.
\begin{lemma} \label{lem:TailUpperBound}
    Let $\delta \geq 2n/m$. Then
        $$
        n \cdot \sum_{A \subseteq [m]}{p_A\REV(\distM^A_T)}
        \leq
        4 \cdot \delta \cdot \SREV(\distM^{n/\delta})
        $$
\end{lemma}
\begin{proof}
    By lemma~\ref{lem:tailToSeperate}:
    $
    \sum_{A}{p_A}{\REV(\distM^A_T)}
    \leq
    2 \cdot \SREV(\distMtail)
    =
    2 \cdot \sum_{j \in [m]}\REV(\vals_j \cdot \ind\left[\vals_j >\cutoff_j \right] )
    $.
    Also, as we consider the case where ${\delta}/{n} \geq 2/m = 2 \cdot \Pr[\vals_j >\cutoff_j]$,
    we can apply
    Lemma~\ref{lem:addbuyers} to each $j$ and get
    $
    \REV(\vals_j \cdot \ind\left[\vals_j >\cutoff_j \right] ) \leq \frac{2\delta}{n} \cdot \REV(\dist_j^{n/\delta})
    $.
    Summing over all $j$ completes the proof.
\end{proof}

\subsection{Core}
\label{subsec:SREV:core}
For ease of exposition, let $\vals^*_j$ denote the random variable
$
\vals^*_j
=
\max_{i \in [n]}\{ \vals^i_j\cdot \ind\left[\vals^i_j \leq \cutoff_j\right]\}
$
where $\vals^i_j \gets \dist_j$ for each $i \in [n]$.
In this section we upper bound the core $\VAL(\distMcore^{n}) = \sum_{j\in [m]}{\expect{}{\vals^*_j}}$.
\paragraph{Cutoff setting in the core.}
Fix some constant $\varepsilon_1 \leq 1$.
For each item $j$, set $\cutoff_j'$ so that
$
\Pr_{\val\gets \dist_j}\left[\val  > \cutoff_j' \right]
=
\frac{\varepsilon_1}{n \cdot \log\frac{m}{n}}
$,
and set $\cutoff_j'' = \max \{\frac{n}{m}\cdot \cutoff_j, \cutoff_j'\}$.
In order to bound the contribution from the core, for each item we separate to ``tiny values" (at most $\frac{n}{m}\cutoff_j$), ``low values'' (at most $\cutoff_j'$), and ``higher values'' (above $\max \{\frac{n}{m}\cdot \cutoff_j, \cutoff_j'\}$).

In the following we bound the contribution of values that are at most $\cutoff_j''$, we provide separate bounds for tiny values (Section~\ref{subsec:tinyValuesSrev}) and low values (Section~\ref{subsec:lowValuesSrev}).

\subsubsection{Tiny values in core} \label{subsec:tinyValuesSrev}
We show in Lemma~\ref{lem:lowerValsManyBuyers} that ``tiny values" contribute a tiny fraction of the revenue that is achievable by increasing the number of buyers by a constant factor.
\begin{lemma} \label{lem:lowerValsManyBuyers}
    Let $\delta  \geq 2n/m $. Then for each item $j$
    $$
            \frac{n}{m} \cutoff_j
    \leq
    2\delta \cdot  \REV(\dist_j^{n/\delta})
    $$
\end{lemma}
\begin{proof}
    Observe that
    \begin{align*}
    \frac{n}{m} \cutoff_j
    =
    n \cdot \cutoff_j \cdot \Pr_{\vals^i_j \gets \dist_j}\left[\vals^i_j > \cutoff_j\right]
    \leq
    n \cdot \REV(\vals^i_j \cdot \ind \left[ \vals^i_j > \cutoff_j\right]).
    \end{align*}
    By Lemma~\ref{lem:addbuyers} for any $\delta /n \geq 2 / m$ it holds that
    $
    \REV(\vals^i_j \cdot \ind [ \vals^i_j > \cutoff_j ]) \leq  \frac{2\delta}{n} \cdot  \REV(\dist_j^{n/\delta}),
    $
    which completes the proof.
\end{proof}

\subsubsection{Low values in core} \label{subsec:lowValuesSrev}

In Lemma~\ref{lem:lowerValsManyBuyers:1} we show that the contribution of ``lower values" to the core can be almost completely covered by the revenue that is achievable by increasing the number of buyers by a factor of $O(\log \frac{m}{n})$.
\begin{lemma} \label{lem:lowerValsManyBuyers:1}
    Let $\delta< \varepsilon_1$. Then for each item $j$
    $$
            \cutoff_j'
    \leq
    (1-e^{- \frac{\varepsilon_1 }{\delta}})^{-1} \cdot \REV\left(\dist_j^{n\cdot \log \left( \frac{m}{n} \right) /\delta}\right)
    $$
\end{lemma}
\begin{proof}
    Directly by applying Lemma~\ref{lem:addbuyersAlmost} to the cutoff $\cutoff_j'$.
\end{proof}

Putting the last two lemmas together,
\begin{corollary} \label{cor:lowerValsManyBuyers}
    Let $2n/m \leq \delta \leq \min \{1/2, \varepsilon_1  \}$. Then for each item $j$
    $$
            \cutoff_j''
    \leq
    (1-e^{- \frac{\varepsilon_1 }{\delta}})^{-1} \cdot \REV\left(\dist_j^{n\cdot \log \left( \frac{m}{n} \right) /\delta}\right)
    $$
\end{corollary}
\begin{proof}
    Since $2\delta \leq 1$, Lemma~\ref{lem:lowerValsManyBuyers:1} always gives a weaker upper bound than
    Lemma~\ref{lem:lowerValsManyBuyers}.
    The result follows by the definition of $\cutoff''_j$ and Lemmata~\ref{lem:lowerValsManyBuyers}~and~\ref{lem:lowerValsManyBuyers:1}.
\end{proof}

\subsubsection{High values in core}
Lemma~\ref{lem:HigherQuantileBoundManyBuyers} shows that ``higher values'' in the core contribute a tiny fraction
of the revenue that is achievable by increasing the number of buyers by a factor of $O(\log \frac{m}{n})$.
\begin{lemma} \label{lem:HigherQuantileBoundManyBuyers}
Let $\delta < 2\varepsilon_1$. Then for each $j$
    $$
    \expect{}{\vals^*_j  \cdot \ind\left[\vals^*_j > \cutoff_j''  \right] }
    \leq
    8 \cdot \varepsilon_1 \cdot \REV\left(\dist_j^{\frac{n \cdot \log \frac{m}{n}}{\delta} } \right).
    $$

\end{lemma}
\begin{proof}
    First observe that since $\vals^*_j  \leq \cutoff_j$, computing the expectation gives
    \begin{align*}
    \expect{}{\vals^*_j  \cdot \ind\left[\vals^*_j > \cutoff_j''  \right] }
    & =
    \int_0^{\cutoff''_j}{\Pr\left[\vals^*_j  \cdot \ind\left[ \vals^*_j > \cutoff_j''  \right] \geq y \right] dy }
    +
    \int_{\cutoff''_j}^{\cutoff_j}{\Pr\left[\vals^*_j > y\right] dy}
    \\ & =
    \cutoff''_j \cdot \Pr\left[\vals^*_j  > \cutoff''_j\right]
    +
    \int_{\cutoff''_j}^{\cutoff_j}{\Pr\left[\vals^*_j > y\right] dy}
    \\ & \leq
    \REV(\vals^*_j \cdot \ind \left[\vals^*_j  > \cutoff''_j\right])
    +
    \int_{\cutoff''_j}^{\cutoff_j}{\Pr\left[\vals^*_j > y\right] dy}.
    \end{align*}
    To upper bound the left-most term,
        let $\pi^* \geq \cutoff_j''$ be the price for which $\pi^* \cdot \Pr[ \vals^*_j > \pi^*] = \REV(\vals^*_j \cdot \ind \left[\vals^*_j  > \cutoff''_j\right])$, and let $\val \gets \dist_j$. By union bound, we get that,
    \begin{align*}
        \pi^* \cdot \Pr[ \vals^*_j > \pi^*]
        \leq
        \sum_{i \in [n]} \pi^* \cdot \Pr[ \vals^i_j\cdot \ind\left[\vals^i_j \leq \cutoff_j\right] > \pi^*]
        \leq
        n \cdot \REV( \val \cdot \ind\left[\cutoff_j'' < \val  \leq \cutoff_j\right]).
    \end{align*}
        To upper bound the right-most term, let $\pi \geq \cutoff_j''$ be the price for which
    $$
    \pi \cdot \Pr \left[ \val \cdot \ind\left[\cutoff_j'' < \val \leq \cutoff_j\right] \right]
    =
    \REV( \val \cdot \ind\left[\cutoff_j'' < \val  \leq \cutoff_j\right])
    $$
    observe that by union bound (recall that $\vals_j^* = \max_{i \in [n]}\{ \vals^i_j \cdot \ind\left[\vals^i_j \leq \cutoff_j\right]\}$),
    \begin{align*}
    \int_{\cutoff''_j}^{\cutoff_j}{\Pr\left[ \vals_j^*> y\right] dy}
    & \leq
    n \cdot \int_{\cutoff''_j}^{\cutoff_j}{\Pr \left[\val \cdot \ind\left[\val \leq \cutoff_j\right] > y\right] dy} \\
    & =
    n \cdot \int_{\cutoff''_j}^{\cutoff_j}{\frac{y \cdot \Pr \left[\val \cdot \ind\left[\val \leq \cutoff_j\right] > y\right]}{y} dy} \\
    & \leq
    n \cdot \pi \cdot \Pr \left[ \val \cdot \ind\left[\cutoff_j'' < \val \leq \cutoff_j\right] > \pi \right] \int_{\cutoff''_j}^{\cutoff_j}{\frac{1}{y} dy}
    \end{align*}
    Since $\cutoff_j'' \geq \frac{n}{m} \cutoff_j$, it holds that
    $\int_{\cutoff''_j}^{\cutoff_j}{\frac{1}{y} dy} \leq \int_{\frac{n}{m} \cutoff_j}^{\cutoff_j}{\frac{1}{y} dy} = \log \left( \frac{m}{n} \right)$,
    we conclude that
    \begin{align*}
    \expect{}{\vals^*_j  \cdot \ind\left[\vals^*_j > \cutoff_j''  \right] }
    \leq
    n \cdot \left(1 + \log \left( \frac{m}{n} \right) \right) \cdot \REV( \val \cdot \ind\left[\cutoff_j'' < \val \leq \cutoff_j\right])
    \end{align*}
    And also $\REV( \val \cdot \ind\left[\cutoff_j'' < \val \leq \cutoff_j\right]) \leq \REV( \val \cdot \ind\left[\val > \cutoff_j'' \right])$
    by single dimensional revenue monotonicity.

    Since $\cutoff_j'' \geq \cutoff_j'$ it holds that
    $
    \Pr\left[\val >\cutoff_j'' \right]
    \leq
    \Pr\left[\val >\cutoff_j' \right]
    =
    \frac{\varepsilon_1}{n \cdot \log \frac{m}{n}}.
    $
    Therefore by applying Lemma~\ref{lem:addbuyers} with $\delta' = 2 \cdot \frac{\varepsilon_1}{n \log \frac{m}{n}}$ we obtain:
    \begin{align*}
    \REV(\val \cdot \ind\left[\val >\cutoff_j''\right] )
    \leq
    \frac{4 \cdot \varepsilon_1}{n \log \frac{m}{n}} \cdot \REV\left(\dist_j^{\frac{n \log \frac{m}{n}}{2 \cdot \varepsilon_1}} \right)
    \leq
    \frac{4 \cdot \varepsilon_1}{n \log \frac{m}{n}} \cdot \REV\left(\dist_j^{\frac{n \log \frac{m}{n}}{\delta} }\right)
    \end{align*}
    Where the last inequality follows from the fact that adding more buyers does not decrease revenue.
    Combining the above we get that:
    \begin{align*}
    \expect{}{\vals^*_j  \cdot \ind\left[\vals^*_j > \cutoff_j''  \right] }
    & \leq
    n \cdot \left(1 + \log \left( \frac{m}{n} \right) \right) \cdot \REV( \val \cdot \ind\left[\val  > \cutoff_j''\right]) \\
    & \leq
    2 \cdot n \cdot \log \left( \frac{m}{n} \right) \cdot \frac{4 \cdot \varepsilon_1}{n \log \frac{m}{n}} \cdot \REV\left(\dist_j^{\frac{n \log \frac{m}{n}}{\delta} }\right)
    \\ & =
    8 \cdot \varepsilon_1 \cdot \REV\left(\dist_j^{\frac{n \log \frac{m}{n}}{\delta}} \right)
    \end{align*}
\end{proof}

\subsubsection{Total contribution from core}

We are now ready to bound the contribution from the core.
Let $\varepsilon_2$ be so that $1+\varepsilon_2 = (1-e^{- \frac{\varepsilon_1 }{\delta}})^{-1}$.
\begin{lemma} \label{lem:CoreUpperBoundManyBuyers}
    Let $2n/m \leq \delta \leq \min \{1/2, \varepsilon_1  \}$. Then
    \begin{align*}
    \VAL(\distMcore^{n})
    \leq &
    \left(1+ \varepsilon_2 +  8 \cdot \varepsilon_1 \right) \cdot \SREV\left( \distM^{ \frac{n\cdot \log \frac{m}{n} }{\delta } } \right)
    \end{align*}
\end{lemma}
\begin{proof}
    It holds that
    \begin{align*}
    \VAL(\distMcore^{n}) = \sum_{j\in [m]}{\expect{}{\vals^*_j}}
    & \leq
    \sum_{j\in [m]}{
                    \cutoff_j''
          + \expect{}{\vals^*_j \cdot \ind\left[\vals^*_j > \cutoff_j'' \right]}}
    \\ & \leq
    \sum_{j\in [m]}{    (1 + \varepsilon_2 ) \cdot \REV\left(\dist_j^{\frac{n\cdot \log \frac{m}{n} }{\delta } }\right)   +  8 \cdot \varepsilon_1 \cdot \REV\left(\dist_j^{\frac{n \cdot \log \frac{m}{n}}{\delta} } \right)},
    \end{align*}
    where the last inequality follows by corollary~\ref{cor:lowerValsManyBuyers} and Lemma~\ref{lem:HigherQuantileBoundManyBuyers}.
    The assertion now follows by observing that the last expression equals 
    $
    \left(1+ \varepsilon_2 +  8 \cdot \varepsilon_1 \right) \cdot \SREV\left( \distM^{ \frac{n\cdot \log \frac{m}{n} }{\delta } } \right)
    $
    by the definition of $\SREV$.
\end{proof}

\subsection{Proof of Theorem~\ref{thm:ManyAdditivebuyers}}\label{sub:ManyAdditivebuyers}
\begin{proof}[Proof of Theorem~\ref{thm:ManyAdditivebuyers}]
The result follows by combining the core-tail decomposition
lemma~(\ref{lem:coreDecompositionManyBuyers}) with
Lemmas~\ref{lem:CoreUpperBoundManyBuyers}~and~\ref{lem:TailUpperBound}
for some constant $\delta$ such that $2n/m \leq \delta \leq \min \{1/2, \varepsilon_1  \}$:
\begin{align*}
\REV(\distM^{n})
& \leq
    \VAL(\distMcore^{n})
    +
    n \cdot \sum_{A \subseteq [m]}{p_A \cdot \REV(\distM^A_T)} \\
& \leq
\left(1+\varepsilon_2+
8 \cdot \varepsilon_1 + 4 \cdot \delta
\right)
\SREV\left( \distM^{\frac{n\cdot \log \frac{m}{n} }{\delta }} \right)
\end{align*}
\end{proof}

\section{Matching Lower Bound for Many Items}
In this section we prove that for a setting with $m = \Omega(n)$, even with $O\left(n \cdot  \ln \frac{m}{n}\right)$ buyers, 
selling separately still may only be a tiny fraction of the optimal revenue with $n$ buyers, thus Theorem~\ref{thm:ManyAdditivebuyers} is tight.
\begin{theorem}\label{thm:lb}
	Let $m = \Omega(n)$.
	There exists a distribution $\distM$ so that for every $\varepsilon > 0$ there exists some $\delta > 0$ so that
	\begin{align*}
		\SREV(\distM^{ \delta \cdot n \cdot \ln \frac{m}{n}}) \leq \varepsilon \cdot \REV(\distM^{n}).
	\end{align*}
\end{theorem}
\begin{proof}
	Denote each buyer $i$'s value for item $j$ is by $\vals^i_j$.
The equal revenue distribution (denoted by ER) on the support $[1,  M]$ is defined by 
$\mbox{ER}(x) = \Pr_\val[\val \leq x] = 1- \frac{1}{x}$ for $x < M$ and $\Pr_\val[\val \leq M] = 1$, 
so, 
$\Pr_\val[\val = M] = \frac{1}{M}$. 
Let $\mathcal{ER} = \times_{j \in [m]}\mbox{ER}$, i.e., $\vals^i_j \gets \mbox{ER}$ for every $i \in [n], j \in [m]$.
Note that $\expect{}{\vals^i_j} = \int_1^M \frac{x}{x^2}dx + 1  = \ln M + 1$, and  $\var(\vals^i_j) \leq \int_1^M \frac{x^2}{x^2}dx + \frac{M^2}{M} \leq 2M$.
For any price $p \in [1, M]$, the same revenue of $1$ is obtained ($p \cdot \Pr[\val > p] = p \cdot (1-(1-\frac{1}{p})) = 1$).
Therefore, combined with Corollary~\ref{cor:RevSubmodular} we conclude that 
$$
\SREV({\mathcal{ER}}^{\delta \cdot n \cdot \ln \left( \frac{m}{n} \right) }) 
= 
m \cdot \REV\left(\mbox{ER}^{\delta \cdot n \cdot \ln \left( \frac{m}{n} \right) }\right) 
\leq 
\delta \cdot m \cdot n \cdot \ln \left( \frac{m}{n} \right) \cdot \REV(\mbox{ER}) 
= \delta \cdot m \cdot n \cdot \ln \left( \frac{m}{n} \right).
$$

Consider the mechanism that orders the buyers (arbitrarily), and offers each buyer in her turn to purchase her favorite bundle of size 
$\frac{m}{4 \cdot n}$ 
at price 
$\payL$.

Fix a buyer $i$ (for simplicity of presentation we henceforth omit $i$ from $\vals^i_j$). Let $B$ be the set of available items upon buyer $i$'s arrival. 
Note that $B \geq m/2$.
Item $j$ is said to be ``large'' if $\vals_j > n$. Let $I_j$ be the indicator r.v. that equals $1$ when item $j$ is large.
This occurs w.p. $\frac{1}{n}$.
Therefore, the expected number of large items is $\frac{B}{n}$.
Also, the probability that there are at most $\frac{m}{4\cdot n}$ large items is, by Chernoff bounds:
\begin{align*}
\Pr [\sum_{j \in [B]} I_j \leq \frac{m}{4\cdot n} ] 
\leq 
\Pr [\sum_{j \in [B]} I_j \leq \frac{B}{2n} ] 
\leq 
\exp \left(-  \frac{B}{8 \cdot n }  \right) 
\leq 
\exp \left(-  \frac{m}{16 \cdot n }  \right) 
\end{align*}
Let $\eventB$ denote the event that $\sum_{j \in B} I_j > \frac{m}{4n}$, i.e., there are more than $\frac{m}{4n}$ large items.
Then $\eventB$ occurs w.p. at least $1- \exp \left(-  \frac{m}{16 \cdot n }  \right) $.
Let ${C}$ denote all subsets of $B$ of size $\frac{m}{4n}$.
Then the expected revenue contribution from buyer $i$ is at least:
\begin{align} \label{eq:lowerBoundRev}
\expect{}{\payL \cdot \ind[\exists S \in {C}: \sum_{j \in S}\vals_j \geq \payL] } 
\geq 
	\left(1- e^{-  \frac{m}{16 \cdot n }} \right) \cdot \payL \cdot \Pr\left[\exists S \in {C} : \sum_{j \in S}\vals_j \geq \payL  \bigg \vert \eventB \right]  
\end{align}
let ${L}$ be the set of all the subsets of $B$ that are of size more than $\frac{m}{4n}$.
For each $A \in {L}$, let $p_A$ be the probability that $A$ is exactly the set of large items, and let $A'$ be a subset of $A$ of size $\frac{m}{4n}$ (say, the $\frac{m}{4n}$ items with the lowest index). 
By the law of total probability
\begin{align} \label{eq:lowerBound:totalE}
	\Pr\left[ \exists S \subseteq B : \sum_{j \in S}\vals_j \geq \payL \bigg \vert \eventB \right]
	& =
	\sum_{A \in {L}}\Pr \left[\exists S \subseteq B : \sum_{j \in S}\vals_j \geq \payL \bigg \vert A \right] \cdot p_A \nonumber \\
	& \geq 
	\sum_{A \in {L}}\Pr \left[\sum_{j \in A'}\vals_j \geq \payL \bigg \vert A \right] \cdot p_A
\end{align}
Using the notation $\val_j  = \vals_j \vert (\vals_j > n)$, 
we have that 
$
\Pr \left[\sum_{j \in A'}\vals_j \geq \payL \bigg \vert A \right]
= 
\Pr[\sum_{j \in A'}\val_j \geq \payL].
$
The expectation of $\val_j$ is 
$
\expect{}{\vals_j \vert \vals_j \geq n }
= 
n \cdot \left(\int_n^{M}\frac{1}{x}dx  + 1 \right) 
=
n \cdot \left(\ln \frac{M}{n} + 1\right),
$ 
and therefore the expectation of $\frac{m}{4n}$ such items is 
$
\frac{m}{4} \cdot \left(\ln \frac{M}{n} + 1\right).
$
The variance of $\val_j$ is
\begin{align*}
\var(\val_j)  
&=
\expect{}{ \left(\vals_j -  n \cdot \left(\ln \frac{M}{n} + 1\right) \right)^2 \vert \vals_j \geq n} \\
& = 
n \cdot \left( \int_n^{M}\frac{\left(x -  n \cdot \left(\ln \frac{M}{n} + 1\right) \right)^2}{x^2} dx  + \frac{\left(M - n \cdot \left(\ln \frac{M}{n} + 1\right) \right)^2 }{M}\right) \\
& \leq 
n \cdot \left( 
	\int_n^{n \cdot \left( \ln \frac{M}{n} + 1\right)}\frac{\left(x -  n \cdot \left(\ln \frac{M}{n} + 1\right) \right)^2}{x^2} dx
	+
	\left(\int_{n \cdot \left(\ln \frac{M}{n} + 1\right)}^M \frac{\left(x -  n \cdot \left(\ln \frac{M}{n} + 1\right) \right)^2}{x^2}dx \right) 
	+ M 
	\right) \\
& \leq 
n \cdot \left( 
\left(n \cdot \left(\ln \frac{M}{n}\right) \right)^2 \cdot \left( \frac{1}{n} - \frac{1}{n \cdot \left( \ln \frac{M}{n} + 1\right)} \right) 
+
M
+ M 
\right) \\
& \leq 
n^2 \cdot \left(  \ln \frac{M}{n}  \right)^2
+
2Mn
\end{align*}

Therefore, the variance of the sum of $\frac{m}{4n}$ such items is at most 
$
\frac{m}{4} \cdot n \cdot \left(  \ln \frac{M}{n}  \right)^2
+
2Mm.
$
By Chebychev, 
\begin{align*}
\Pr \left[
\sum_{j \in A'}\val_j 
	\leq 
		\frac{m}{8} \cdot \left(\ln \frac{M}{n} + 1\right) 
\right] 
\leq 
4 \cdot \left(
\frac{\frac{m}{4} \cdot n \cdot \left(  \ln \frac{M}{n}  \right)^2
	+
	2Mm
	}
	{
		\left( \frac{m}{4} \cdot \left(\ln \frac{M}{n} + 1\right) \right)^2 
	}
\right)
\leq 
4 \cdot \left(
\frac{4n }{m} 
+
\frac{2Mm}{	\left( \frac{m}{4} \cdot \ln \frac{M}{n} \right)^2 }
\right)
\end{align*}
Combining with Equations~(\ref{eq:lowerBoundRev})~and~(\ref{eq:lowerBound:totalE}), 
we get that for $\payL = \frac{m}{8} \cdot \left(\ln \frac{M}{n} + 1\right) $, the contribution of buyer $i$ is at least 
\begin{align*}
	\left(1- \exp \left(-  \frac{m}{16 \cdot n }  \right) \right) \cdot \left( 1-
	\left(
	\frac{16 n }{m} 
	+
	\frac{8Mm}{	\left( \frac{m}{4} \cdot \ln \frac{M}{n} \right)^2 }
	\right)
	\right) \cdot 
	\frac{m}{8} \cdot \ln \frac{M}{n}. 
\end{align*}
Now set $M = m$, and note that for some (universal) constant $c$, if $m \geq c \cdot n$ then the above (that lower bounds the revenue contribution from buyer $i$) is at least $\frac{m}{16} \cdot \ln \frac{m}{n}$.
Hence the revenue from the above mechanism is at least $\frac{m\cdot n}{16} \cdot \ln \frac{m}{n}$,
while selling separately to $\delta \cdot n \cdot \ln \frac{m}{n}$ buyers extracts at most 
$
\delta \cdot n \cdot m \cdot \ln \left( \frac{m}{n} \right),
$
which completes the proof.
\end{proof}

\section{Many Buyers: $\SREV$ with No Additional Buyers}
% !TeX root = main99revenue.tex
\label{SEC:MOREBUYERS}
In this section we show that adding more buyers is not required for a number of buyers $n$ that is sufficiently larger than the number of items $m$.
Specifically, we prove the following theorem:
\begin{theorem} [Theorem~\ref{thm:SREV-many-buyers}, case $n \gg m$]\label{thm:ManyBuyersNoNeedForMore}
    For any constant $\varepsilon > 0 $ there exists a constant $\delta(\varepsilon)>0$ such that whenever $n \geq \frac{m}{\delta}$:
    $$
    (1-\varepsilon)\REV\left( \distM^{n} \right) \leq \SREV\left( \distM^{n} \right)
    $$
\end{theorem}

\subsubsection{Proof outline - a tale of three tails}
We first partition the domain of the valuation distributions into sub-domains, using the Subdomain Stiching Lemma (Lemma~\ref{lem:subdomainStiching}).
In our analysis, we condition on the event (or sub-domain) $\event$ that describes which items are in the tail or core for which buyer (we still don't know their values within the tail/core).
For each event we define a bipartite graph with buyer-nodes $[n]$ and item-nodes $[m]$, where an edge $\{i, j\}$ between buyer $i$ and item $j$ exists if and only if item $j$ is in buyer $i$'s tail.
Crucially, the revenue is separable across connected components (because in different connected components, disjoint sets of buyers are interested in disjoint sets of items).

We now establish a ``marginal mechanism lemma'' (Lemma~\ref{lem:marginalSubB}) that shows that the optimal revenue from $\distM^\event$
is bounded from above by the sum of expected item values in the core in event $\event$, plus the revenue from selling items in each connected component to the buyers in that connected component (conditioned on them being in the tail).

For a connected component that contains only one item (i.e., the buyers in this connected component have only this item in their tails), the lemma provides an even tighter bound. Namely, we can restrict attention to mechanisms that sell the item to buyers in the connected component with probability at most $p$, where $p$ is the upper bound on the ex-ante probability that the optimal mechanism (for event $\event$) sells this item. Then, with probability $1-p$ we can resell the item to other buyers.
Essentially, this means that we can use most buyers and some items to (almost) recover the contribution from the tail, and then use the rest of the buyers and items to (almost) recover the contribution from the core.
Note that while we partition the buyers in advance, the partitioning of the items is only done after seeing which items are bought by the first subset of buyers. Note also that the mechanism of partitioning the buyers is for analysis purpose only: once we establish that it achieves a good revenue by selling each item separately, we can only improve the revenue by running Myerson's optimal mechanism.

Since we consider a setting with many buyers, the expected number of
buyers that have a specific item in their tail (henceforth termed
{\em interested buyers}) is concentrated. This implies a core-tail
decomposition lemma (Lemma~\ref{lem:coreDecompositionMoreBuyers})
that bounds the optimal revenue using four terms: one core, and
three tails. The three tails are (1) the contribution from selling
item $j$ whenever the number of interested buyers does not exceed
the expected number of interested buyers by much, (2) the
contribution from selling item $j$ whenever the number of interested
buyers {\em does} exceed the expected number of interested buyers by
a significant constant, and (3) the contribution from selling items
in connected components with at least two items.

We bound the three tails in Section~\ref{subsec:manyBuyers:tail}.

In Lemma~\ref{lem:fewAgentsInItemTail} we show that the first tail is almost completely recovered
by the revenue from selling items separately to all but a small fraction of the buyers, and restricting the ex-ante probability of sale as described above.
Note that the tail is not fully recovered, but this is a sufficient guarantee due to the probability of sale restriction (see above). In particular, the items can be resold to other buyers with sufficiently high probability.

In Lemma~\ref{lem:ManyBuyersAjLarge} we show that the second tail is a tiny fraction of the revenue from selling items separately.
This is done by arguing that the event where the number of interested buyers significantly exceeds the expectation is extremely rare.

In Lemma~\ref{lem:moreBuyersSrevTailBound} we show that the third tail is also a tiny fraction of the revenue from selling items separately.
This part is more involved, since we need to bound the expected size of a connected component, as well as some of the higher moments. Here we use some simple ideas from percolation theory. See Subsection~\ref{sub:Tail3} for details.

Finally, we conclude the proof in
Section~\ref{sub:ManyAdditivebuyersMoreBuyers} by integrating the
above arguments into the core-tail decomposition.

\subsection{The item-buyer bipartite graph} \label{subsec:itemBuyerGraph}
We adopt all the notation from Section~\ref{sec:coreTail}, and introduce some more notation.
For a single dimensional distribution $\val \gets \dist$ and cutoff $\cutoff$, let $\hat{\dist}$ denote the distribution of $\val \vert \val > \cutoff$.
Fix an event $\event$ as described in Section~\ref{sec:coreTail}.
Consider the bipartite graph $H_{\event} = ([n], [m], E[\event])$ over buyers and items.
The set of edges is defined to be $E[\event] = \{ \{i, j\} : j \in \event^i\}$,
i.e., for each item $j$ in the tail of $i$, ($\vals_j^i > \cutoff_j$) there exists an edge $\{i, j\}$.
For each node $k$, let $\event[k]$ denote $k$'s neighbors in $H_\event$, i.e., for an item $j \in [m]$,  $\event[j] = \{i : j \in \event^i \}$,
and for a buyer $i \in [n]$, $\event[i] = \event^i$.
Note that $H_\event$ is a random graph where each edge $\{i, j\}$ exists independently w.p. $\Pr\left[ \vals^i_j > \cutoff_j\right]$.
Let $P^{\event}$ be the partition of the graph $H_{\event}$ to connected components.
For a connected component $X \in P^{\event}$, let $X_n = X \cap [n]$ be the nodes in $X$ associated with buyers, and similarly $X_m = X \cap [m]$ are the nodes in $X$ associated with items.
For a connected component $X \in P^\event$, let $\distM^X_T$ be the product distribution
over the buyers $X_n$, where each buyer $i$ has only values $\vals_j^i$ for items $j$ in $\event[i]$, and each such $\vals_j^i$ is drawn from the conditional distribution $\hat{\dist}_j$, i.e.,
$\dist^X_T =
\times_{i \in X_n}\distM^{\event[i]}_T$.

\subsection{Core-tail decomposition.}  \label{SEC:MOREITEMSCORETAIL}

We introduce a new ``marginal mechanism'' lemma (Lemma~\ref{lem:marginalSubB}) that is suitable for the case $n \gg m$.
Fix some $\event$ as in Section~\ref{sec:coreTail},
and let $\M$ be an optimal mechanism w.r.t. $\distM^{\event}$, i.e.,
$$
    \REV_{\M}(\distM^{\event})
    =
    \REV(\distM^{\event}).
$$
Let $\alloc$ be the mechanism's allocation function, and let
$\piBar_j = \sum_{i \in \event[j]} \expect{\vals \gets \distM^{\event}}{ \alloc_j^i(\vals)}$.
Recall that for a connected component $X \in P^{\event}$, $X_n = X \cap [n]$ and similarly $X_m = X \cap [m]$.
Also, recall that $\hat{\dist}_j$ is the distribution of the random variable $\val \gets \dist_j$ conditioned on $\val > \cutoff_j$.
\begin{lemma} \label{lem:marginalSubB}
    (Marginal Mechanism on Sub-Domain)
    Fix $\event$, and let $\piBar$ be as above.
    Let $P^{\event}_2$ be all the connected components $X$ in $P^{\event}$ so that $\cardinality{X_m} \geq 2$.
    Then
\begin{align*}
    \REV(\distM^{\event})
\leq
\sum_{j \in [m]}\left( \REV_{\piBar_j}\left(\hat{\dist}_j^{\cardinality{\event[j]}}\right) + \cutoff_j \cdot (1 - \piBar_j) \right)
+
\sum_{X \in P^{\event}_2}       \REV(\distM^{X}_T)
\end{align*}
I.e., the revenue from $\distM^\event$ is bounded by the revenue from selling items separately to buyers in the tail at $\piBar_j$ ex-ante probability of sale per item $j$, plus the contribution to the core - the cutoff values times the ex-ante probability of being in the core ($1- \piBar_j$) -  plus the revenue from connected components with more than one item.
\end{lemma}
\begin{proof}
Recall that mechanism $\M$ is so that
$
    \REV_{\M}(\distM^{\event})
    =
    \REV(\distM^{\event}),
$
and $\alloc_j^i(\cdot)$ is buyer $i$'s allocation function for item $j$ in mechanism $\M$.
Let $\pay^i(\cdot )$ be buyer $i$'s payment function in mechanism $\M$.
Recall that
for a set $S \subseteq [n]\times [m]$,
we use $\vals_S = \{\vals_j^i \}_{(i, j) \in S}$
and
$\vals_{-S} = \{\vals_j^i \}_{(i, j) \not \in S}$.
Let $S^i = \{ \{i, j\}: j \in \event[i]\}$.

Consider the mechanism $\M_{\event}$ that allows buyer $i$ to purchase only the items in $\event[i]$, and performs as follows:
first the mechanism samples $\hat{\vals} \gets \distM^{\event}$, reports to each buyer $i$ only the values $\{ \hat{\vals}_j^i \}_{j \not \in \event[i]}$, and solicits bids $\vals_{S^i}$.
Then the mechanism simulates $\M$ with the bids $(\vals_{\event} , \hat{\vals}_{-\event}) \triangleq \left(\{\vals_{S^i}\}_{i} , \{\hat{\vals}_{-S^i}\}_{i} \right)$,
while refunding each buyer $i$ by $\sum_{j \not \in \event[i]}{\alloc^{i}_{j}(\vals_{\event} , \hat{\vals}_{-\event}) \cdot \hat{\vals}^i_{j} }$.
From each buyer $i$'s perspective, each other buyer $t$'s values and samples are drawn from $\dist^{\event^t}$.
Also, the report $\hat{\vals}^i_{j \not \in \event[i]}$ is drawn from $\distM^{\event^i}_C$.
Therefore, buyer $i$'s expected payment for submitting $\vals_{S^i}$ is:
\begin{align*}
\pay^i_{\event}(\vals_{S^i})
= \expect{\hat{\vals} \gets \distM^\event }{
    \pay^i(\vals_{S^i} , \hat{\vals}_{-S^i} ) - \sum_{j \not \in \event[i]}{\alloc^{i}_{j}(\vals_{S^i} , \hat{\vals}_{-S^i}) \cdot \hat{\vals}^i_{j} }
    },
\end{align*}
and each item $j \in \event[i]$ is allocated to $i$ w.p. $\expect{\hat{\vals} \gets \distM^\event }{\alloc_j^i(\vals_{S^i} , \hat{\vals}_{-S^i} )}$.
Truthfulness of $\M_{\event}$ follows by that $\M$ is truthful.
By independence across items and buyers, the revenue of $\M_{\event}$ is
\begin{align*}
    \REV_{\M_{\event}}(\distM^{\event})
    =
    \expect{\vals \gets \distM^\event}{\sum_{i \in [n]}\pay^i_{\event}(\vals_{S^i})  }
    =
    \REV_{\M}(\distM^{\event}) - \expect{\vals \gets \distM^{\event}}{\sum_{i}\sum_{j \not \in \event[i]}{\alloc^{i}_{j}(\vals) \cdot \vals^i_{j} }}
\end{align*}
Note that $j \not \in \event[i]$ implies that $\vals^i_j \leq \cutoff_j$ and mechanism $\M_{\event}$ sells each item $j$ w.p. at most
$
    \piBar_j
$, therefore
\begin{align} \label{eq:refundBound}
\expect{\vals \gets \distM^{\event}}{\sum_{i}\sum_{j \not \in \event[i]}{\alloc^{i}_{j}(\vals) \cdot \vals^i_{j} }}
\leq
\sum_{j\in [m]} \cutoff_j \cdot \sum_{i \not \in \event[j]}\expect{\vals \gets \distM^{\event}}{{\alloc^{i}_{j}(\vals) }}
\leq
\sum_{j\in [m]}\cutoff_j \cdot (1 - \piBar_j)
\end{align}

Fix a connected component $X \in P^{\event}$.
Recall that the mechanism $\M_{\event}$ allows a buyer $i \in X_n$ to purchase only the items in $\event[i] \subseteq X_m$.
Consider the mechanism $\M_{X}$ that sells only to buyers in $X_n$ and allows each buyer $i \in X_n$ to purchase only items from $\event[i]$,
by (privately) sampling $\hat{\vals} \gets \distM^{\event}$ and executing $\M_{\event}$ with the bids of the buyers in $X_n$, and with the samples $\{\hat{\vals}^i\}_{i \not \in X_n}$.
Then this mechanism's revenue is
$\REV_{\M_{X}}(\distM^{X}_T) =
\expect{\vals \gets \distM^\event}{\sum_{i \in X_n}\pay^i_{\event}(\vals_{S^i})  }$,
and since connected components $\{X\}$ form a partition $\{X_n\}$ over the buyers, we get that
\begin{align*}
    \REV_{\M_{\event}}(\distM^\event)
    =
    \sum_{X \in P^{\event}} \REV_{\M_{X}}(\distM^{X}_T),
\end{align*}
Let $P^{\event}_1$ be all the connected components $X \in P^\event$ so that $\cardinality{X_m} = 1$.
For every $X \in P^{\event}_1$, $\M_X$ sells the single item $j$ that is in $X_m$ to the buyers $\event[j]$, therefore
$$
\REV_{\M_X}(\distM^{X}_T)
=
\REV_{\M_X}\left(\hat{\dist}_j^{\cardinality{\event[j]}}\right)
\leq
\REV_{\piBar_j}\left(\hat{\dist}_j^{\cardinality{\event[j]}}\right),
$$
Therefore,
\begin{align*}
\REV(\distM^{\event})
& \leq
\sum_{X \in P^{\event}_1}   \REV_{\M_{X}}(\distM^{X}_T)
+
\sum_{X \in P^{\event}_2}   \REV_{\M_{X}}(\distM^{X}_T)
+
\sum_{j\in [m]}\cutoff_j \cdot (1 - \piBar_j) \\
& \leq
\sum_{j \in [m]} \left(\REV_{\piBar_j}\left(\hat{\dist}_j^{\cardinality{\event[j]}}\right) + \cutoff_j \cdot (1 - \piBar_j) \right)
+
\sum_{X \in P^{\event}_2}       \REV(\distM^{X}_T)
\end{align*}
The first inequality is by Equation~\ref{eq:refundBound}.
The second inequality is by summing over all $j \in [m]$ and not just items in $P^{\event}_1$, and by
upper bounding the revenue of $\M_X$ for $X \in P^{\event}_2$ by the optimal revenue for the same setting, i.e.,
$
\REV_{\M_{X}}(\distM^{X}_T)
\leq
\REV(\distM^{X}_T) .
$
This completes the proof
\end{proof}
Recall that given $\val \gets \dist_j$, we use $\hat{\dist}_j$ to denote the distribution of $\val \vert \val > \cutoff_j$.
The following lemma upper bounds the revenue from $k$ buyers with values that are conditioned to be in the tail, with
the revenue of more buyers from the original distribution (i.e., without conditioning on tail item values).
\begin{lemma} \label{lem:revFromTailMoreBuyers}
    For item $j$ and cutoff $\cutoff_j$, let $\delta  = 2\cdot  \Pr_{\val \gets \dist_j}\left[\val > \cutoff_j \right]$.
    Then
    for any $k \leq 1/\delta$,
    $$
    \REV_{}\left(\hat{\dist}_j^{k}\right)
    \leq
    4 \cdot k \cdot \REV_{}(\dist_j^{1/\delta})
    $$
\end{lemma}
\begin{proof} Chaining Corollary~\ref{cor:RevSubmodular}, Lemma~\ref{lem:condRevIndRev}, and Lemma~\ref{lem:addbuyers} gives
    \begin{align*}
    \REV_{}\left(\hat{\dist}_j^{k}\right)
    \leq
    k \cdot \REV_{}\left(\hat{\dist}_j\right) 
    =
    \frac{2 k\cdot \REV(\val \cdot \ind [\val > \cutoff_j])}{\delta} 
    & \leq
    \frac{4 \cdot \delta \cdot k\cdot \REV\left(\dist^{1/\delta}\right)}{\delta}, 
    \end{align*}
    as required.
\end{proof}

We repeat the notation used in the following core-tail decomposition
lemma.
For a single dimensional distribution $\val \gets \dist$ and cutoff $\cutoff$, let $\hat{\dist}$ denote the distribution of $\val \vert (\val > \cutoff)$.
Recall that every realization of $\vals \gets \distM^n$ is in
some event $\event$. Then $\event[j]$ is the set of buyers that have
item $j$ is in their tail. $X^j$ denotes the connected component
that contains item $j$ in the graph defined in
Section~\ref{subsec:itemBuyerGraph}  (the event $\event$ will be
understood from the context). For a connected component $X$, we use
$X_n = X \cap [n]$ and $X_m = X \cap [m]$. Also, $\piBar_j = \sum_{i
\in \event[j]} \expect{\vals \gets \distM^{\event}}{
\alloc_j^i(\vals)}$ where $\alloc$ is the allocation function of an
optimal mechanism w.r.t. $\distM^\event$.
\begin{lemma}
    \label{lem:coreDecompositionMoreBuyers}
    (Core-tail decomposition)
    For cutoffs $\cutoff_j$ so that $\Pr_{\val \gets \dist_j}\left[\val > \cutoff_j \right] \geq 1/(2n),$
    and a cutoff $\cutoff^*$
    \begin{align*}
    \REV\left( \distM^{n} \right)
    \leq
    &
    \sum_{j \in [m]}\cutoff_j \cdot \expect{}{1 - \piBar_j} && \mbox{Core} \\
    & +
    \sum_{j \in [m]}\expect{}{
        \REV_{\piBar_j}\left(\hat{\dist}_j^{\cardinality{\event[j]}}\right) \cdot \ind \left[\cardinality{\event[j]} \leq \cutoff^* \right]
    }
        && \mbox{Tail 1} \\
    & +
    4 \cdot \sum_{j \in [m]} \REV\left(\distM^{n}_j\right)
        \cdot
        \expect{}{ \cardinality{\event[j]} \cdot
            \ind \left[\cardinality{\event[j]} > \cutoff^* \right]
        }
        && \mbox{Tail 2}
    \\
    & +
    4 \cdot \sum_{j \in [m]} \REV\left(\distM^{n}_j\right)
    \cdot
    \expect{}{\cardinality{X^j_m} \cdot \cardinality{X^j_n}^2  \cdot \ind[\cardinality{X^j_m} \geq 2] }
    && \mbox{Tail 3}
    \end{align*}
\end{lemma}
\begin{proof}
    Chaining Lemma~\ref{lem:subdomainStiching} and Lemma~\ref{lem:marginalSubB} gives
    \begin{align*}
    \REV\left( \distM^{n} \right)
    \leq &
    \sum_{j \in [m]}\expect{}{\cutoff_j \cdot (1 - \piBar_j)}  \\
    &
    +
    \expect{}{
        \sum_{j \in [m]} \REV_{\piBar_j}\left(\hat{\dist}_j^{\cardinality{\event[j]}}\right)
    }
     \\
    & +
    \expect{}{\sum_{X \in P^{\event}_2}     \REV(\distM^{X}_T) }.
    \end{align*}

    For every $j$,
    Lemma~\ref{lem:revFromTailMoreBuyers} with $k = \cardinality{\event[j]}$,
    combined with the assumption that $2\Pr_{\val \gets \dist_j}\left[\val > \cutoff_j \right] \geq 1/n$ gives
    \begin{align} \label{eq:coreDecompositionMoreBuyers:1}
        \REV\left(\hat{\dist}_j^{\cardinality{\event[j]}} \right)
        \leq
        4 \cdot \cardinality{\event[j]} \cdot \REV(\dist_j^{n}).
    \end{align}
    Also,
    by Lemma~\ref{lem:revSrev} it holds that $\REV(\distM^{X}_T) \leq \cardinality{X_m} \cdot \cardinality{X_n} \cdot \SREV(\distM^{X}_T)$,  therefore
    \begin{align*}
    \SREV(\distM^{X}_T)
    =
    \sum_{j \in X_m} \REV\left(\hat{\dist}_j^{\cardinality{\event[j]}} \right)
    \leq
    \sum_{j \in X_m} 4 \cdot \cardinality{\event[j]} \cdot \REV(\dist_j^{n})
    \end{align*}
and after applying Inequality~(\ref{eq:coreDecompositionMoreBuyers:1}) to cases when $\cardinality{\event[j]} \leq \cutoff^*$
 we get that
\begin{align*}
\REV\left( \distM^{n} \right)
\leq &
\sum_{j \in [m]}\cutoff_j \cdot \expect{}{1 - \piBar_j} \\
& +
\sum_{j \in [m]}\expect{}{
    \left( \REV_{\piBar_j}\left(\hat{\dist}_j^{\cardinality{\event[j]}}\right) \cdot \ind \left[\cardinality{\event[j]} \leq \cutoff^* \right]  \right)
} \\
& +
\sum_{j \in [m]}\expect{}{
    4 \cdot \cardinality{\event[j]} \cdot \REV(\dist_j^{n}) \cdot \ind \left[\cardinality{\event[j]} > \cutoff^* \right]
} \\
& +
\expect{}{\sum_{X \in P^{\event}_2}     \cardinality{X_m} \cdot \cardinality{X_n} \cdot \sum_{j \in X_m} 4 \cdot \cardinality{\event[j]} \cdot \REV(\dist_j^{n})  }.
\end{align*}
By definition $j \in X \in P^{\event}_2$ actually says that $\cardinality{X^j_m} \geq 2$.
Also, clearly $\event[j] \subseteq X^j_n$ ($j$'s neighbors are in $j$'s connected component),
 therefore the proof follows by linearity of expectation.
\end{proof}

\newcommand{\epsCutoff}{\varepsilon_1}
\newcommand{\epsSREVtail}{\mathcal{C}}
\newcommand{\epsAboveAj}{\varepsilon_2}
\newcommand{\epsLessBuyers}{\delta}
\newcommand{\epsBelowAj}{\varepsilon_3}
\newcommand{\epsAwayBuyers}{\varepsilon_4}

\paragraph{Cutoff setting.}
Fix some $0 < \epsCutoff, \epsAboveAj, \epsBelowAj, \epsAwayBuyers < 1$.
Throughout this section,
    for every item $j$ fix $\cutoff_j$ to be the value for which
    $\Pr_{\val \gets \dist_j}\left[\val > \cutoff_j\right]
    =
    \frac{1}{n \cdot \epsCutoff}
    .$
Clearly the expected number of buyers that have item $j$ in their tail is $\frac{1}{\epsCutoff}$.
We set $\cutoff^*= \frac{(1+\epsAboveAj)}{\epsCutoff}$, i.e., we use $\epsAboveAj$ as a deviation from this expected size.

We will first sell items separately at high prices to a $1-\epsAwayBuyers$ fraction of the buyers.
In this case the expected number of buyers that have item $j$ in their tail is $\frac{1-\epsAwayBuyers}{\epsCutoff}$.
We will use $\epsBelowAj$ as a deviation from this expected size.

Whenever an item $j$ is not sold to on of the $1-\epsAwayBuyers$ fraction of the buyers,
we will sell the item to one of the remaining $\epsAwayBuyers$ buyers at price $\cutoff_j$.
This will suffice to to almost fully recover the contribution from the core.

\subsection{Core.}
\begin{lemma} \label{lem:moreBuyersCoreBound}
    For every item $j$,
    $
    \cutoff_j \cdot \expect{}{(1 - \piBar_j)}
    \leq
    \expect{}{(1 - \piBar_j)} \left(1-e^{-{\epsAwayBuyers}/{\epsCutoff}}\right)^{-1}\REV(\dist_j^{\epsAwayBuyers \cdot n})
    $
\end{lemma}
\begin{proof}
    Directly implied by Lemma~\ref{lem:addbuyersAlmost}.
\end{proof}

\subsection{Tail.} \label{subsec:manyBuyers:tail}
\newcommand{\epsApxTail}{\varepsilon_5}

The following lemma essentially combines Lemma~\ref{lem:restProbHighVals} with a concentration bound, to show that Tail~1 can be almost fully recovered by selling items separately to a $1-\epsAwayBuyers$ fraction of the buyers.
Let $\epsApxTail = \epsApxTail(\epsCutoff, \epsAboveAj, \epsAwayBuyers, \epsBelowAj)$ be so that
$1+\epsApxTail = \frac{1+\epsAboveAj}{(1-\epsAwayBuyers)\cdot (1-\epsBelowAj)} \cdot \frac{1}{1 - \exp\left(- \frac{(1-\epsAwayBuyers) \cdot \epsBelowAj^2}{2 \cdot \epsCutoff}  \right) }$.
It is only important to note that $\epsApxTail \rightarrow 0$ as $\epsCutoff \rightarrow  0$.
\begin{lemma} \label{lem:fewAgentsInItemTail}
    (Tail 1).
    Fix an item $j$.
    For $k \leq \frac{(1+\epsAboveAj)}{\epsCutoff}$,
    \begin{align*}
        \REV_{p}\left(\hat{\dist}_j^{k}\right)
        \leq
        (1+\epsApxTail)\cdot  \REV_{p}(\dist_j^{(1-\epsAwayBuyers)\cdot n})
    \end{align*}
\end{lemma}
\begin{proof}
By Lemma~\ref{lem:restMonotonicity}, adding more buyers does not decrease revenue, i.e.,
    \begin{align*}
            \REV_{p}\left(\hat{\dist}_j^{k}\right)
            & \leq
            \REV_{p}\left(\hat{\dist}_j^{\frac{(1+\epsAboveAj)}{\epsCutoff}}\right)
    \end{align*}

Fix $1-\epsilon = (1-\epsAwayBuyers) \cdot (1-\epsBelowAj)$, and fix $\epsLessBuyers = 1 -\frac{1-\epsilon}{1+\epsAboveAj}$.
We set $\delta$ so that by decreasing a $\delta$ fraction of the buyers we remain
with an amount that is under $\frac{1 - \epsAwayBuyers }{\epsCutoff}$, i.e.,  we have that
$\frac{1+\epsAboveAj}{\epsCutoff} (1-\epsLessBuyers) = \frac{1-\epsilon}{\epsCutoff}$,
and that
$\frac{1}{1-\delta} = \frac{1+\epsAboveAj}{1-\epsilon}$.
Therefore, by Corollary~\ref{cor:RevSubmodular}
\begin{align*}
\REV_{p}\left(\left(\hat{\dist}_j \right)^{\frac{(1+\epsAboveAj)}{\epsCutoff}}\right)
\leq
\frac{1+\epsAboveAj}{1-\epsilon} \cdot \REV_{p}\left(\hat{\dist}_j^{\frac{1-\epsilon}{\epsCutoff}}\right).
\end{align*}
Also,
observe that
\begin{align*}
    \REV_{p}\left(\hat{\dist}_j^{\frac{1-\epsilon}{\epsCutoff}}\right)
    & \leq
    \frac{\REV_{p}\left(\dist_j^{(1-\epsAwayBuyers)\cdot n}\right) }{\Pr\left[ \cardinality{\{i \in [(1-\epsAwayBuyers)\cdot n] : \vals^i_j > \cutoff_j \}} \geq \frac{1-\epsilon}{\epsCutoff} \right] } && \mbox{Lemma~\ref{lem:restProbHighVals}} \\
    & \leq
    \frac{\REV_{p}\left(\dist_j^{(1-\epsAwayBuyers)\cdot n}\right)}{1 - e^{- \frac{(1-\epsAwayBuyers) \cdot \epsBelowAj^2}{2 \cdot \epsCutoff}  } } && \mbox{Chernoff bounds}
\end{align*}
To prove the last inequality, recall the definition of $\epsilon$ and
 observe that $\cardinality{\{i : \vals^i_j > \cutoff_j \}}$ is a sum of $(1-\epsAwayBuyers)\cdot n$ i.i.d. indicators that equal $1$ w.p. $\frac{1}{n \cdot \epsCutoff}$.
Therefore its expectation is $\mu = \frac{1-\epsAwayBuyers}{\epsCutoff}$ and
$$
\Pr\left[ \cardinality{\{i : \vals^i_j > \cutoff_j \}} < \frac{(1-\epsAwayBuyers)\cdot (1-\epsBelowAj)}{\epsCutoff} \right] \leq e^{- \frac{(1-\epsAwayBuyers) \cdot \epsBelowAj^2}{2 \cdot \epsCutoff}  },
$$
which completes the proof.
\end{proof}

The following lemma shows that the coefficient of $\SREV(\distM^n)$ in Tail 2 tends to $0$ as $\epsCutoff$ tends to $0$.
\newcommand{\epsTailTwo}{\varepsilon_6}
Let $\epsTailTwo = \frac{4}{\epsCutoff} \cdot \exp\left( - \frac{\epsAboveAj^2 }{8 \cdot \epsCutoff } \right) $.
\begin{lemma}  \label{lem:ManyBuyersAjLarge}
    (Tail 2)
    Fix item $j$.
$
\expect{}{\cardinality{\event[j]} \cdot
    \ind \left[\cardinality{\event[j]} > \cutoff^* \right]
}
 \leq
\epsTailTwo
$
\end{lemma}
\begin{proof}
By law of total expectation
\begin{align*}
    \expect{}{\cardinality{\event[j]} \cdot
        \ind \left[\cardinality{\event[j]} > \cutoff^*  \right]
    }
    & \leq
    \sum_{k = \cutoff^*}^{n}  \expect{}{ \cardinality{\event[j]}
        \bigg \vert \cardinality{\event[j]} = k
    } \cdot \Pr[\cardinality{\event[j]} = k] \\
    & =
    \sum_{k = \cutoff^*}^{n}  k \cdot \Pr[\cardinality{\event[j]} = k] .
\end{align*}
To apply Chernoff bounds, observe that
$
    \Pr[\cardinality{\event[j]} \geq k]
    =
    \Pr[\cardinality{\event[j]} \geq \frac{1}{\epsCutoff}(1+\epsCutoff \cdot k -1 )],
$
so for
$\delta
=
\epsCutoff \cdot k - 1,
$
and $\mu = \frac{1}{\epsCutoff}$. We get
\begin{align*}
        \Pr[\cardinality{\event[j]} \geq \mu (1+ \delta)]
        & \leq
        \exp\left( - \frac{\delta^2}{2+\delta} \cdot \mu \right) \\
        & =
        \exp\left( - \frac{(\epsCutoff \cdot k - 1 )^2}{1+ \epsCutoff \cdot k } \cdot \frac{1}{\epsCutoff} \right) \\
        & \leq
        \exp\left( - \frac{(\epsCutoff \cdot k - 1 )^2}{2 \cdot \epsCutoff \cdot k } \cdot \frac{1}{\epsCutoff} \right) \\
        & =
        \exp\left( - \frac{(k - \frac{ 1}{\epsCutoff } )^2}{2 \cdot k } \right) \\
        & \leq
        \exp\left( - \frac{(\frac{\epsAboveAj}{2} \cdot k)^2}{2 \cdot k } \right) \\
        & =
        \left(\exp\left( - \frac{\epsAboveAj^2 }{8 } \right) \right)^{k}.
\end{align*}
The last inequality holds because $\epsAboveAj \leq 1$ implies that $1+\epsAboveAj \geq \frac{1}{1-\frac{\epsAboveAj}{2}}$,
therefore by rearranging $k  \geq \frac{(1 + \epsAboveAj)}{\epsCutoff} \geq \frac{1}{\epsCutoff \cdot (1-\frac{\epsAboveAj}{2})}$
we get $k - \frac{\epsAboveAj}{2}\cdot k  \geq  \frac{1}{\epsCutoff}$. In total
\begin{align*}
        \sum_{k = \cutoff^*}^{n}  k \cdot \Pr[\cardinality{\event[j]} = k]
        =
        \sum_{ k = 1}^{n}  k \cdot \left(\exp\left( - \frac{\epsAboveAj^2 }{8 } \right) \right)^{k}
        -
        \sum_{k = 1}^{\cutoff^*}  k \cdot \left(\exp\left( - \frac{\epsAboveAj^2 }{8 } \right) \right)^{k}   \\
\end{align*}
Recall that $\sum_{k = 1}^t{k \cdot z^k} = z \frac{1-(t+1)z^t + tz^{t+1}}{(1-z)^2}$.
Let $z = \exp\left( - \frac{\epsAboveAj^2 }{8 } \right)$.
Then the above equals
$
    \frac{z }{(1-z)^2}\left((\cutoff^*+1)z^{\cutoff^*} - {\cutoff^*}z^{\cutoff^*+1}\right)
     =
    \frac{z^{\cutoff^*+1} }{(1-z)^2}\left(1 + \cutoff^*(1-z) \right)
    \leq
    2 \cdot \cutoff^* \cdot z^{\cutoff^*+1}.
$
In total we get
\begin{align*}
        \expect{}{\cardinality{\event[j]} \cdot
            \ind \left[\cardinality{\event[j]} > \cutoff^* \right]
        }
        & \leq
        2 \cdot \cutoff^* \cdot \exp\left( - \frac{\epsAboveAj^2 \cdot \left(1+\epsAboveAj + \epsCutoff \right) }{8 \cdot \epsCutoff } \right)
        & \leq
        \frac{4}{\epsCutoff} \cdot \exp\left( - \frac{\epsAboveAj^2 }{8 \cdot \epsCutoff } \right)
\end{align*}
\end{proof}

\subsection{The 3$^{\text{rd}}$ tail (a non-trivial connected component)}\label{sub:Tail3}

In the following lemmas are goal is to show that the coefficient of $\SREV(\distM^n)$ in Tail 3 tends to $0$ as as $\epsCutoff$ tends to $0$.
This will be proved in Lemma~\ref{lem:moreBuyersSrevTailBound}.
We will develop the proof in steps.
Recall that for a connected component $X \in P^\event$, $X_n = X \cap [n]$ and $X_m = X \cap [m]$.
Lemma~\ref{lem:moreBuyersSrevTailBound:2} upper bounds the term we wish to bound into two terms.
Note that in both terms, at a loss of a constant factor, we free ourself from the dependence on the number of buyers that have at most one item
in their tail (and are in $j$'s connected component).
\begin{lemma} \label{lem:moreBuyersSrevTailBound:2}
    For every $\event$, let $S \subseteq [n]$ denote the buyers $i$ with at least two items in their tail, i.e.,  $\event[i] \geq 2$.
    For every item $j$, let $\hat{\ind} = \ind[\cardinality{X^j_m} \geq 2] $, then
    \begin{align*}
    \expect{}{\cardinality{X^j_m} \cdot \cardinality{X^j_n}^2 \cdot \hat{\ind}}
    \leq
    \frac{3}{\epsCutoff} \expect{}{\left(\cardinality{X^j_m}+ \cardinality{X^j_n \cap S}\right)^3 \cdot \hat{\ind}}
    + \frac{1}{\epsCutoff^2} \cdot \expect{}{\cardinality{X^j_m}\cdot \hat{\ind}}
    \end{align*}
\end{lemma}
\begin{proof}
For some $S \subseteq [n]$,
let $Q = \{Q^i \in 2^{[m]} , \cardinality{Q^i} \geq 2\}_{i \in S}$
denote the event that $\event[i] = Q^i$ for every $i \in S$ (and $\cardinality{\event[i]} \leq 1$ for every $i \not \in S$).
Note that conditioned on $Q$, for every $j$, $X^j_m$ is fully determined --
and hence, also $\hat{\ind}$ (defined above) --
because buyers $i$ with $\cardinality{\event[i]} \leq 1$ cannot change the connectivity of items.
Therefore, for every item $j$, to upper bound
$
    \expect{}{\cardinality{X^j_m} \cdot \cardinality{X^j_n}^2 \cdot \hat{\ind} \bigg  \vert Q }
$
it suffices to upper bound $\expect{}{\cardinality{X^j_n}^2  \bigg  \vert Q } $.
Fix $Q$, and let $S_j = S \cap X^j_n$, i.e.,  $S_j$ are all buyers $i$ that are in item $j$'s connected component and $\cardinality{\event[i]} \geq 2$.
Observe that conditioned on $Q$, the set $S_j$ is determined.
Also, let $R_j = \expect{}{\cardinality{\{i: \event[i] =\{j\} \} }  \big  \vert Q}$.
Then $X^j_n = S_j \mathbin{\dot{\cup}} \{i: \event[i] =\{j\} \} $ which implies that
\begin{align*}
\expect{}{\cardinality{X^j_n}^2  \bigg  \vert Q }
 =
\expect{}{\left( \cardinality{S_j} + \cardinality{\{\event[i] =\{j\} \}}\right)^2  \bigg  \vert Q }
 =
\cardinality{S_j}^2 + 2 \cardinality{S_j} R_j + R_j^2
\end{align*}
And also
\begin{align*}
R_j
& =
\sum_{i \not \in S}\expect{}{ \ind\left[ \event[i] =\{j\}  \right]  \bigg  \vert \cardinality{\event[i]} \leq 1  } \\
&
\leq
\sum_{i \in [n]} \Pr \left[ \event[i] =\{j\}   \bigg  \vert \cardinality{\event[i]} \leq 1 \right] \\
& =
\sum_{i \in [n]}
\frac{\frac{1}{n \cdot \epsCutoff } \cdot \left(1 - \frac{1}{n \cdot \epsCutoff} \right)^{m-1}}{\left(1 - \frac{1}{n \cdot \epsCutoff}\right)^m + \frac{m}{n \cdot \epsCutoff} \cdot \left(1 - \frac{1}{n \cdot \epsCutoff}\right)^{m-1} } \\
& \leq
\frac{1}{\epsCutoff }.
\end{align*}
Also,
$
\cardinality{S_j}^2 + \frac{2 \cdot \cardinality{S_j}}{\epsCutoff}  + \frac{1}{\epsCutoff^2}
\leq
\frac{3}{\epsCutoff}\cardinality{S_j}^2 + \frac{1}{\epsCutoff^2}.
$
Taking expectation over $Q$,
we get that
\begin{align*}
\expect{}{\cardinality{X^j_m} \cdot \cardinality{X^j_n}^2 \cdot \hat{\ind}  }
&
\leq
\expect{}{\cardinality{X^j_m}\cdot \hat{\ind} \left(\frac{3}{\epsCutoff} \cardinality{S_j}^2 + \frac{1}{\epsCutoff^2}\right) }
\leq
\frac{3}{\epsCutoff} \expect{}{\left(\cardinality{X^j_m}+ \cardinality{S_j}\right)^3 \cdot \hat{\ind} }
+ \frac{1}{\epsCutoff^2} \cdot \expect{}{\cardinality{X^j_m} \cdot \hat{\ind} }
\end{align*}
As required.
\end{proof}

Clearly, showing that $\expect{}{\left(\cardinality{X^j_m}+ \cardinality{X^j_n \cap S}\right)^3}$ is small implies that
$\expect{}{\cardinality{X^j_m}} $ is small.
In Lemma~\ref{lem:boundExpectationXj} we first upper bound the latter as the proof is shorter and contains most of the ideas to bound the former.
For this, we use simple ideas from percolation theory for bounding the size of a connected component in a random graph.
\begin{lemma} \label{lem:boundExpectationXj}
    For every item $j$,
    $
    \expect{}{\cardinality{X^j_m}}  \leq    \frac{1}{1 - \frac{m}{n \cdot \epsCutoff^2}}
    $
\end{lemma}
\begin{proof}
First observe that all edges in $H_\event$ are i.i.d., therefore by symmetry across buyers $i$ there exists
$
\hat{x}_n \triangleq
\expect{}{\cardinality{X^i_m}}
$
and by symmetry across items $j$ there exists
$
\hat{x}_m \triangleq
\expect{}{\cardinality{X^j_m}}.
$

Fix $\event$ and item $j$.
Consider the graph
$H_{\event}^{-j} = ([n], [m], E[\event] \setminus \{\{i, j\}\}_{i \in [n]} )$, i.e., the graph $H_\event$ with all edges $\{i, j\}$ removed.
Let $Y^k$ denote $k$'s connected component in in the graph $H_{\event}^{-j}$.
Then
$
\cardinality{X^j_m} \leq 1 + \sum_{i \in \event[j]}\cardinality{Y^i_m} = 1 + \sum_{i \in [n]}\ind[i \in \event[j]]\cdot\cardinality{Y^i_m}
$
because every item $k \in X^j_m$ has at least one $i \in \event[j]$ that is connected to it without passing through $j$.
By law of total expectation w.r.t. $\event[j]$,
$
\hat{x}_m
\leq
1 + \sum_{i \in [n]}\expect{}{\expect{}{\ind[i \in \event[j]] \cdot \cardinality{Y^i_m} \bigg \vert \event[j]}}
$
Observe that for every $i \in [n]$, the random set $Y^i_m$ is independent of $\cardinality{\event[j]}$
(because the edges $\{i', j\}$ never appear in $H_\event^{-j}$).
Also, for every $\event$ and every $k$ it holds that $Y^k_m \subseteq X^k_m$.
Therefore in total
$
\hat{x}_m
\leq
1 + \hat{x}_n \cdot \sum_{i \in [n]}\expect{}{\expect{}{\ind[i \in \event[j]]  \bigg \vert \event[j]}}
1 + \hat{x}_n \cdot \expect{}{\cardinality{\event[j]}},
$
and since $\expect{}{\cardinality{\event[j]}} = \frac{1}{\epsCutoff}$ we get that
$$
\hat{x}_m
\leq
1 +  \frac{\hat{x}_n }{\epsCutoff}.
$$

Similarly, fix $\event$ and buyer $i$.
Consider the graph
$H_{\event}^{-i} = ([n], [m], E[\event] \setminus \{\{i, j\}\}_{j \in [m]} )$, i.e., the graph $H_\event$ with all edges $\{i, j\}$ removed.
Now redefine $Y^k$ to denote $k$'s connected component in the graph $H_{\event}^{-i}$.
Then
$
\cardinality{X^i_m} \leq \sum_{k \in [m]} \ind[k \in \event[i]] \cdot \cardinality{Y^k_m}
$
because every item in $X^i_m$ as at least one $k \in \event[i]$ that is connected to it without passing through $i$.
By law of total expectation
$
\hat{x}_n \leq \sum_{k \in [m]} \expect{}{\expect{}{\ind[k \in \event[i]] \cdot \cardinality{Y^k_m}}\bigg \vert \event[i]},
$
and observe that as before $Y^k_m$ is independent of $\event[i]$, and $Y^k_m \subseteq X^k_m$,
therefore
$
\hat{x}_n
\leq
\hat{x}_m \cdot \sum_{k \in [m]} \expect{}{\expect{}{\ind[k \in \event[i]] \bigg \vert \cardinality{\event[i]}}}
=
\hat{x}_m \cdot \expect{}{\cardinality{\event[i]}}
,$
and
since $\expect{}{\cardinality{\event[i]}} = \frac{m}{n \cdot \epsCutoff} $ we get that
$$
\hat{x}_n
\leq
\hat{x}_m \cdot \frac{m}{n \cdot \epsCutoff},
$$
and in total we conclude that
$
\hat{x}_m
\leq
1 +  \frac{\hat{x}_n }{\epsCutoff}
\leq
1 +  \frac{\hat{x}_m \cdot \frac{m}{n \cdot \epsCutoff}, }{\epsCutoff},
$
i.e.,
\begin{align*}
    \hat{x}_m
    \leq
    \frac{1}{1 - \frac{m}{n \cdot \epsCutoff^2}}
\end{align*}
\end{proof}

The more technical Lemma~\ref{lem:boundExpectationXj3} handles two new challenges compared to Lemma~\ref{lem:boundExpectationXj}: (i) counting also buyers in the connected component that have at least two items in their tail, and (ii) the random variable is raised to the power of $3$.
Both challenges are mostly technical. The first is overcome by using the same proof technique as Lemma~\ref{lem:boundExpectationXj} to bound the number of nodes in a connected component after removing buyers with at most one item in their tail, and the second is overcome by standard Cauchy-Schwarz arguments.
\begin{lemma} \label{lem:boundExpectationXj3}
    For every $\event$, let $S \subseteq [n]$ denote the buyers $i$ with at least two items in their tail, i.e.,  $\event[i] \geq 2$.
    There exists a (universal) constant $c$ so that for every item $j$,
    $$
    \expect{}{\left(\cardinality{X^j_m} + \cardinality{X^j_n \cap S}\right)^3}
    \leq
    \frac{1 + c \cdot \frac{m}{n \cdot \epsCutoff^2}}{1 - c \cdot \frac{m}{n \cdot \epsCutoff^2}}
    $$
\end{lemma}
\begin{proof}
    We will actually bound $\expect{}{\left(\cardinality{X^j_m} + \cardinality{X^j_n \cap S}\right)^4}$.

    Fix $\event$ and item $j$.
    For each $\event$, let $\mathcal{H}_\event$ be the bipartite graph with nodes $[n]$ and $[m]$ and the edges
    $$
    \bigcup_{i : \cardinality{\event[i]} \geq 2} \{\{i, j\} :  j \in \event[i]\}.
    $$
    I.e., the bipartite graph that contains an edge $\{i, j\}$ for every $j \in \event[i]$ such that $\cardinality{\event[i]} \geq 2$.
    Define $\event'[k]$  to be the neighbors of node $k$ in $\mathcal{H}_\event$.
    Let $x[k]$ denote the number of nodes in node $k$'s connected component in the graph $\mathcal{H}_\event$.

    Let
    $\mathcal{H}_{\event}^{-j}$ denote the graph $\mathcal{H}_\event$ after removing all edges connected to $j$, i.e., without the edges
    $\{\{i, j\}\}_{i \in [n]}$.
    Let $y[k]$ denote the number of nodes in node $k$'s connected component in $\mathcal{H}_\event^{-j}.$
    Observe that by symmetry across items, it holds that
    for all items $j$ there exists
    $
    \hat{x}_m \triangleq
    \expect{}{\left(x[j]\right)^4}
    $
    and by symmetry across buyers, it holds that for all buyers $i$ there exists
    $
    \hat{x}_n \triangleq
    \expect{}{\left(x[i]\right)^4}.
    $
    Observe that with this new notation our goal is now to bound
    $\hat{x}_m$.
    Observe that
    $
    x[j]
    \leq
    1 + \sum_{i \in \event'[j]}y[i]
    $
    because when upper bounding $x[j]$, we count one for $j$, and every node in $j$'s connected component is reachable via some neighbor $i$, without going through an edge that touches $j$, and hence is counted in $y[i]$. Therefore,
    \begin{align*}
    (x[j])^4
    = &
    \left(1 + \sum_{i \in [n]}\ind[ i \in \event'[j] ] \cdot y[i] \right)^4 \\
    =
    1
    & + 4 \cdot \sum_{i \in [n]}\ind[ i \in \event'[j] ] \cdot y[i]  \\
    & + 6 \cdot \sum_{i, i' \in [n]}\ind[ i, i' \in \event'[j] ] \cdot y[i] \cdot y[i'] \\
    & + 4 \cdot \sum_{i, i', i'' \in [n]}\ind[ i, i', i'' \in \event'[j] ] \cdot y[i] \cdot y[i'] \cdot y[i''] \\
    & + \sum_{i, i', i'', i''' \in [n]}\ind[ i, i', i'',i''' \in \event'[j] ] \cdot y[i] \cdot y[i'] \cdot y[i'']\cdot y[i''']
    \end{align*}

    Taking expectation and by law of total expectation w.r.t. $\event'[j]$, we get that
    \begin{align*}
    \hat{x}_m
    \leq
    1
    & + 4 \cdot \sum_{i \in [n]}\expect{}{\expect{}{\ind[ i \in \event'[j]] \cdot y[i] \bigg \vert \event'[j]}} \\
    & + 6 \cdot \sum_{i, i' \in [n]}\expect{}{\expect{}{\ind[ i, i' \in \event'[j] ] \cdot y[i] \cdot y[i'] \bigg \vert \event'[j]}} \\
    & + 4 \sum_{i, i', i'' \in [n]}\expect{}{\expect{}{\ind[ i, i', i'' \in \event'[j] ] \cdot y[i] \cdot y[i'] \cdot y[i''] \bigg \vert \event'[j]}} \\
    & + \sum_{i, i', i'', i''' \in [n]}\expect{}{\expect{}{\ind[ i, i', i'', i''' \in \event'[j] ] \cdot y[i] \cdot y[i'] \cdot y[i''] \cdot y[i'''] \bigg \vert \event'[j]}}
    \end{align*}
    Observe that for every $i$ the random variable $y[i]$ is independent of $\event'[j]$,
    hence, also independent of $\ind[ i' \in \event'[j]]$ for all $i'$.
    Also, for every $i$, $y[i] \leq x[i]$ because $i$'s connected component in $\mathcal{H}_\event^{-j}$ is a subset of $i$'s connected component in $\mathcal{H}_\event$.
    Therefore, we get that
    \begin{align} \label{eq:xmUpperBound}
    \hat{x}_m
    \leq
    1
    & + 4 \cdot \sum_{i \in [n]}\expect{}{x[i]} \cdot \expect{}{\expect{}{\ind[ i \in \event'[j]] \bigg \vert \event'[j]}} \nonumber \\
    & + 6 \cdot \sum_{i, i' \in [n]}\expect{}{x[i] \cdot x[i']} \cdot \expect{}{\expect{}{\ind[ i, i' \in \event'[j] ] \bigg \vert \event'[j]}} \nonumber \\
    & + 4 \sum_{i, i', i'' \in [n]}\expect{}{x[i] \cdot x[i'] \cdot x[i'']} \cdot \expect{}{\expect{}{\ind[ i, i', i'' \in \event'[j] ] \bigg \vert \event'[j]}} \nonumber \\
    & + \sum_{i, i', i'', i''' \in [n]}\expect{}{x[i] \cdot x[i']\cdot x[i'']\cdot x[i'''] }  \cdot \expect{}{\expect{}{\ind[ i, i', i'', i''' \in \event'[j] ]  \bigg \vert \event'[j]}}
    \end{align}
    Observe that by symmetry across buyers, we also have that for all buyers $i$ there exists  $x^\dagger = \expect{}{(x[i])^2}$.
    By Cauchy-Schwarz's inequality:
    \begin{align*}
    \expect{}{x[i] \cdot x[i']}
    \leq
    \sqrt{\expect{}{(x[i])^2} \cdot \expect{}{(x[i'])^2}} = x^\dagger \leq \hat{x}_n
    \end{align*}
    Also,  applying Cauchy-Schwarz's inequality twice combined with Jensen's inequality:
    \begin{align*}
    \expect{}{x[i] \cdot x[i'] \cdot x[i''] }
    &\leq
    \sqrt{\expect{}{(x[i])^2 \cdot (x[i'])^2}} \cdot \sqrt{\expect{}{(x[i'])^2}} \\
    &
    \leq
    \left(\expect{}{(x[i])^4} \cdot \expect{}{(x[i'])^4} \cdot \expect{}{(x[i'])^4}\right)^{1/4} \\
    & \leq
    \hat{x}_n
    \end{align*}
    And finally, by applying Cauchy-Schwarz's inequality twice:
    \begin{align*}
    \expect{}{x[i] \cdot x[i']\cdot x[i'']\cdot x[i'''] }
    & \leq
    \sqrt{ \expect{}{(x[i] \cdot x[i'])^2} \cdot \expect{}{ (x[i'']\cdot x[i'''])^2 } } \\
    & \leq
    \left( \expect{}{(x[i]^4)} \cdot \expect{}{(x[i'])^4} \cdot \expect{}{ (x[i''])^4} \cdot \expect{}{(x[i'''])^4 } \right)^{1/4} \\
    & =
    \hat{x}_n
    \end{align*}
    Plugging into Inequality~(\ref{eq:xmUpperBound}) we get:
    \begin{align*}
    \hat{x}_m
    &\leq
    1 + \hat{x}_n \cdot \left(
    4 \expect{}{\cardinality{\event'[j]}}
    + 6 \expect{}{{\cardinality{\event'[j]}}^{2}}
    + 4 \expect{}{{\cardinality{\event'[j]}}^{3}}
    + \expect{}{{\cardinality{\event'[j]}}^{4}}
    \right) \\
    &\leq
    1 + 15 \cdot \hat{x}_n \cdot \expect{}{\cardinality{\event'[j]}^4}
    \end{align*}

    Observe that applying the same argument as above, but replacing "buyers" with "items", i.e.,
    defining $\mathcal{H}_{\event}^{-i}$ to be the graph $\mathcal{H}_\event$ after removing all th edges connected to {\em buyer} $i$,
    defining $\event'[k]$ to be $k$'s neighbors in $\mathcal{H}_{\event}^{-i}$, and observing
    that
    $
    x[i] \leq
    1 + \sum_{j \in [m]} \ind[j \in \event'[i]] \cdot y[j]
    $, and applying the same argument as in Inequality~\ref{eq:xmUpperBound} but using law of total expectation w.r.t. $\event'[i]$,
    etc., gives that
    \begin{align*}
    \hat{x}_n
    &\leq
    1 + 15 \cdot \hat{x}_m \cdot \expect{}{\cardinality{\event'[i]}^4}.
    \end{align*}
    Therefore, in total we get that
    \begin{align*}
    \hat{x}_m
    \leq
    1 +
    15 \cdot \left( 1 + 15 \cdot \hat{x}_m \cdot \expect{}{\cardinality{\event'[i]}^4}\right) \cdot \expect{}{\cardinality{\event'[j]}^4}
    \end{align*}
    I.e.,
    \begin{align} \label{eq:xmUpperBound:1}
    \hat{x}_m
    \leq
    \frac{1+15 \cdot \expect{}{\cardinality{\event'[j]}^4}}{1- 15 \cdot \expect{}{\cardinality{\event'[i]}^4}\cdot \expect{}{\cardinality{\event'[j]}^4}}
    \end{align}
    Therefore, it remains to bound $\expect{}{\cardinality{\event'[j]}^4}$ and $\expect{}{\cardinality{\event'[i]}^4}$.
    Fix item $j$, and let $\ind[i] = \ind[j \in \event[i]]$, $\ind'[i] = \ind[\exists k \neq j : k \in \event[i] ].$
    Observe that for all $i$, $\expect{}{\ind[i]} = \frac{1}{n \cdot \epsCutoff}$ and
    $\expect{}{\ind'[i]} = 1- \left(1 - \frac{1}{n \cdot \epsCutoff}\right)^{m-1} \leq \frac{m}{n \cdot \epsCutoff}$.
    With this notation, it holds that
    $
    \cardinality{\event'[j]}
    =
    \sum_{i \in [n]}\expect{}{\ind[i] \cdot \ind'[i]}.
    $
    Observe that all the indicators are independent, therefore, when calculating $\expect{}{\cardinality{\event'[j]}^4}$,
    whenever indices are different we can take the product expectation, and whenever indices equal
    we can use for every indicator $\ind$ that $\ind^2 = \ind$.
    Let $k$ denote the number of indices that are different, then
    \begin{align*}
    \expect{}{\cardinality{\event'[j]}^4}
    = &
    \expect{}{\left(\sum_{i \in [n]}\expect{}{\ind[i] \cdot \ind'[i]}\right)^4} \\
    = &
    \frac{m}{n \cdot \epsCutoff^2} +\sum_{k = 1}^{4}\binom{4}{k}\binom{n}{k}\left(\frac{m}{n^2 \cdot \epsCutoff^2}\right)^k \\
    \leq &
    \frac{m}{n \cdot \epsCutoff^2} +\sum_{k = 1}^{4}\binom{4}{k}\frac{1}{k!}\left(\frac{m}{n \cdot \epsCutoff^2}\right)^k \\
    \leq &
    \frac{m}{n \cdot \epsCutoff^2}\left(1 +\sum_{k = 1}^{4}\binom{4}{k}\frac{1}{k!}\right)
    \end{align*}
    Where the last inequality is by assuming $\frac{m}{n \cdot \epsCutoff} \leq 1$.
    A simple calculation shows that
    $$
    \expect{}{\cardinality{\event'[j]}^4}
    \leq
    9 \cdot \frac{m}{n \cdot \epsCutoff^2}
    $$

    Similarly, fix buyer $i$ and redefine as follows: $\ind[j] = \ind[j \in \event[i]], \ind'[j] = \ind[ \exists k \neq j: k \in \event[i].$
    Observe that as before (summing over items instead of buyers), we get that
    \begin{align*}
    \expect{}{\cardinality{\event'[i]}^4}
    & =
    \expect{}{\left(\sum_{j \in [m]}\expect{}{\ind[j] \cdot \ind[j']}\right)^4} \\
    & =
    \frac{m}{n \cdot \epsCutoff^2} + \sum_{k = 1}^{4}\binom{4}{k}\binom{m}{k}\left(\frac{m}{n^2 \cdot \epsCutoff^2}\right)^k \\
    & \leq
    \frac{m}{n \cdot \epsCutoff^2} + \sum_{k = 1}^{4}\binom{4}{k}\frac{1}{k!}\left(\frac{m}{n \cdot \epsCutoff}\right)^{2k} \\
    & \leq
    \frac{m}{n \cdot \epsCutoff} \cdot \left(\frac{1}{\epsCutoff} + \sum_{k = 1}^{4}\binom{4}{k}\frac{1}{k!} \right)
    \end{align*}
    Where the last inequality is by assuming $\frac{m}{n \cdot \epsCutoff} \leq 1$.
    A simple calculation shows that
    $$
    \expect{}{\cardinality{\event'[i]}^4}
    \leq
    \frac{m\left(1 + 8 \cdot \epsCutoff\right)}{n \cdot \epsCutoff^2}
    \leq
    \frac{9m}{n \cdot \epsCutoff^2} .
    $$
    Plugging into Inequality~\ref{eq:xmUpperBound:1} we get that
    \begin{align*}
    \hat{x}_m \leq \frac{1 + 135\cdot \frac{m}{n\cdot \epsCutoff^2}}{1 - 1215 \cdot \frac{m}{n \cdot \epsCutoff^2}}
    \end{align*}
As required.
\end{proof}

We are now ready to bound $\expect{}{\cardinality{X^j_m} \cdot \cardinality{X^j_n}^2  \cdot \ind[\cardinality{X^j_m} \geq 2] }.$
\newcommand{\epsTailThree}{\varepsilon_7}
Let $\epsTailThree = \epsTailThree(m, n, \epsCutoff)$ be so that for the universal constant $c$ from Lemma~\ref{lem:boundExpectationXj3},
$\epsTailThree =    \frac{4}{n \cdot \epsCutoff^4} +
\frac{m}{n}\left( \frac{6 \cdot c }{\epsCutoff^3\cdot\left(1 - c \cdot \frac{m}{n \cdot \epsCutoff^2}\right)}
+
\frac{1}{\epsCutoff^4\cdot \left(1 - \frac{m}{n \cdot \epsCutoff^2}\right)} \right).$
\begin{lemma}  \label{lem:moreBuyersSrevTailBound}
    (Tail 3)
    For every $j$:
    \begin{align*}
    \expect{}{\cardinality{X^j_m} \cdot \cardinality{X^j_n}^2 \cdot \ind[\cardinality{X^j_m} \geq 2] }
    & \leq
    \epsTailThree
    \end{align*}
\end{lemma}
\begin{proof}
    First observe that trivially $\cardinality{X^j_m} \geq 1$ for all $\event$ (because $j$ is in its own connected component).
    By Lemma~\ref{lem:boundExpectationXj}
    we have that $
    \expect{}{\cardinality{X^j_m}}  \leq    \frac{1}{1 - \frac{m}{n \cdot \epsCutoff^2}}.
    $
    Furthermore
    $
    \Pr[\cardinality{X^j_m} = 1]
    =
    1 - \Pr[\exists i \in [n]: j, k \in \event[i]]
    \geq
    1 - \frac{1}{n \cdot \epsCutoff^2},
    $
    because $\Pr[\exists i \in [n]: j, k \in \event[i]] \leq \sum_{i \in [n]} \frac{1}{(n \cdot \epsCutoff)^2} = \frac{1}{n \cdot \epsCutoff^2} $.

    Therefore,
\begin{align*}
\frac{1}{1 - \frac{m}{n \cdot \epsCutoff^2}}
\geq
    \expect{}{\cardinality{X^j_m}}
    & =
    \expect{}{\cardinality{X^j_m} \cdot \ind [\cardinality{X^j_m} = 1]}
    +
    \expect{}{\cardinality{X^j_m} \cdot \ind [\cardinality{X^j_m} \geq 2]} \\
    & =
    \Pr[\cardinality{X^j_m} = 1]
    +
    \expect{}{\cardinality{X^j_m} \cdot \ind [\cardinality{X^j_m} \geq 2]}  \\
    & \geq
    1 - \frac{1}{n \cdot \epsCutoff^2} +
    \expect{}{\cardinality{X^j_m} \cdot \ind [\cardinality{X^j_m} \geq 2]}
\end{align*}
Which implies that
\begin{align*}
\expect{}{\cardinality{X^j_m} \cdot \ind [\cardinality{X^j_m} \geq 2]}
 \leq
\frac{1}{n \cdot \epsCutoff^2} + \frac{\frac{m}{n \cdot \epsCutoff^2}}{1 - \frac{m}{n \cdot \epsCutoff^2}}.
\end{align*}
Similarly, by Lemma~\ref{lem:boundExpectationXj3} (recall that for every $\event$, $S \subseteq [n]$ denotes the buyers $i$ with $\event[i] \geq 2$).

    $$
    \expect{}{\left(\cardinality{X^j_m} + \cardinality{X^j_n \cap S}\right)^3}
    \leq
    \frac{1 + c \cdot \frac{m}{n \cdot \epsCutoff^2}}{1 - c \cdot \frac{m}{n \cdot \epsCutoff^2}}.
    $$
    and
    $\left(\cardinality{X^j_m} + \cardinality{X^j_n \cap S}\right)^3$ is also trivially at least $1$.
    Also, observe that $\cardinality{X^j_n \cap S} > 1$ implies $\cardinality{X^j_m} > 1$,
    and $\cardinality{X^j_m} = 1$ implies  $\cardinality{X^j_n \cap S} = 0$. In words -
    $\cardinality{X^j_m} + \cardinality{X^j_n \cap S}$
    is more than one if and only if there exists a buyer with item $j$ and another item in her tail, i.e.,
    $
    \Pr[\cardinality{X^j_m} + \cardinality{X^j_n \cap S} = 1] = \Pr[\cardinality{X^j_m} =1] \geq 1- \frac{1}{n\cdot \epsCutoff}.
    $
    Therefore, together with Lemma~\ref{lem:boundExpectationXj3} we conclude that for the Lemma's constant $c$,
    \begin{align*}
    \frac{1 + c \cdot \frac{m}{n \cdot \epsCutoff^2}}{1 - c \cdot \frac{m}{n \cdot \epsCutoff^2}}
    \geq &
    \expect{}{\left(\cardinality{X^j_m} + \cardinality{X^j_n \cap S}\right)^3} \\
    = &
    \expect{}{\left(\cardinality{X^j_m} + \cardinality{X^j_n \cap S}\right)^3 \cdot \ind [\cardinality{X^j_m} + \cardinality{X^j_n \cap S} = 1]} \\
    &
    +
    \expect{}{\left(\cardinality{X^j_m} + \cardinality{X^j_n \cap S}\right)^3 \cdot \ind [\left(\cardinality{X^j_m} + \cardinality{X^j_n \cap S}\right)^3 \geq 2]} \\
    = &
    \Pr[\cardinality{X^j_m} + \cardinality{X^j_n \cap S} = 1]
    +
    \expect{}{\left(\cardinality{X^j_m} + \cardinality{X^j_n \cap S}\right)^3 \cdot \ind [\left(\cardinality{X^j_m} + \cardinality{X^j_n \cap S}\right)^3 \geq 2]}  \\
    & \geq
    1 - \frac{1}{n \cdot \epsCutoff^2} +
    \expect{}{\left(\cardinality{X^j_m} + \cardinality{X^j_n \cap S}\right)^3 \cdot \ind [\left(\cardinality{X^j_m} + \cardinality{X^j_n \cap S}\right)^3 \geq 2]},
    \end{align*}
    which implies that
    \begin{align*}
    \expect{}{\left(\cardinality{X^j_m} + \cardinality{X^j_n \cap S}\right)^3 \cdot \ind [\left(\cardinality{X^j_m} + \cardinality{X^j_n \cap S}\right)^3 \geq 2]}
    \leq
    \frac{1}{n \cdot \epsCutoff^2} + \frac{2 c \cdot \frac{m}{n \cdot \epsCutoff^2}}{1 - c \cdot \frac{m}{n \cdot \epsCutoff^2}}
    \end{align*}
    Combined with Lemma~\ref{lem:moreBuyersSrevTailBound:2}, we conclude that
    \begin{align*}
        \expect{}{\cardinality{X^j_m} \cdot \cardinality{X^j_n}^2 \cdot \ind[\cardinality{X^j_m} \geq 2] }
        & \leq
            \frac{3}{\epsCutoff} \left( \frac{1}{n \cdot \epsCutoff^2} + \frac{2 c \cdot \frac{m}{n \cdot \epsCutoff^2}}{1 - c \cdot \frac{m}{n \cdot \epsCutoff^2}}  \right)
            +
            \frac{1}{\epsCutoff^2} \cdot \left( \frac{1}{n \cdot \epsCutoff^2} + \frac{\frac{m}{n \cdot \epsCutoff^2}}{1 - \frac{m}{n \cdot \epsCutoff^2}} \right). \\
        & \leq
            \frac{4}{n \cdot \epsCutoff^4} +
            \frac{m}{n}\left( \frac{6 \cdot c }{\epsCutoff^3\cdot\left(1 - c \cdot \frac{m}{n \cdot \epsCutoff^2}\right)}
            +
            \frac{1}{\epsCutoff^4\cdot \left(1 - \frac{m}{n \cdot \epsCutoff^2}\right)} \right).
    \end{align*}
    As required.
\end{proof}

\subsection{Proof of Theorem~\ref{thm:ManyBuyersNoNeedForMore}}
\label{sub:ManyAdditivebuyersMoreBuyers}
\begin{proof}[proof of Theorem~\ref{thm:ManyBuyersNoNeedForMore}]
Combine the core-tail decomposition
(Lemma~\ref{lem:coreDecompositionMoreBuyers}) with
Lemmas~\ref{lem:moreBuyersCoreBound},
Lemma~\ref{lem:fewAgentsInItemTail} with $p = \piBar_j$,
Lemma~\ref{lem:ManyBuyersAjLarge}, and
Lemma~\ref{lem:moreBuyersSrevTailBound} to get:

    \begin{align*}
    \REV\left( \distM^{n} \right)
    \leq
    &
    \left(1-e^{-{\epsAwayBuyers}/{\epsCutoff}}\right)^{-1} \cdot \sum_{j \in [m]} \expect{}{(1 - \piBar_j)} \cdot \REV(\dist_j^{\epsAwayBuyers \cdot n})
    \\
    & +
    (1+\epsApxTail)\cdot \sum_{j \in [m]}\expect{}{
                  \REV_{\piBar_j}(\dist_j^{(1-\epsAwayBuyers)\cdot n})
    }
\\
    & +
    4 \cdot \epsTailTwo \cdot \sum_{j \in [m]} \REV\left(\distM^{n}_j\right)
    \\
    & +
    4 \cdot \epsTailThree \cdot  \sum_{j \in [m]} \REV\left(\distM^{n}_j\right)
    \end{align*}

Consider a seller that sells the items separately.
Each item $j$ is first auctioned to $(1-\epsAwayBuyers)\cdot n$ buyers using an optimal mechanism
with restricted ex-ante probability of sale of $\piBar_j$.
Whenever the mechanism does not allocate the item, the seller proceeds and sells the item to the remaining $\epsAwayBuyers \cdot n$ buyers optimally.
Clearly, the revenue of the suggested mechanism is at least
$$
\REV_{\piBar_j}(\dist_j^{(1-\epsAwayBuyers)\cdot n})    + \REV(\dist_j^{\epsAwayBuyers \cdot n}) \cdot (1-\piBar_j)
$$
(it is not exactly this, as the mechanism that sells to $(1-\epsAwayBuyers)\cdot n$ buyer might sell with ex-ante probability that is {\em lower} than $\piBar_j$,
in which case the suggested mechanism will make even more than the above.

The above mechanism, by definition, does not make more than $\REV(\dist_j^n)$ revenue, hence,
    \begin{align*}
    \REV\left( \distM^{n} \right)
    \leq
    \left(
    \max\{
    \left(1-e^{-{\epsAwayBuyers}/{\epsCutoff}}\right)^{-1} ,  (1+\epsApxTail)\}
    +
    4 \cdot (\epsTailTwo + \epsTailThree )
    \right)  \cdot \sum_{j \in [m]} \REV(\dist_j^{n})
    \end{align*}
As required.
\end{proof}

\section{Single Buyer: $\max\{\SREV,\BREV\}$ with a Constant Number of Buyers}
\label{SEC:SREVBREV}
In this section we show that the better of selling separately and selling the grand bundle to a constant number of buyers extracts 99\%
of the optimal revenue.
Specifically we prove that:
\begin{theorem}     \label{thm:oneAdditivebuyerSrevBrev}
    (Theorem~\ref{thm:SREVBREV})
    For any constant $\varepsilon> 0$ there exists a constant $\delta(\varepsilon) > 0 $ such that:
    \begin{align*}
    \REV(\distM) \leq (1 + \varepsilon) \max\{\SREV(\distM^{1/\delta}), \BREV(\distM^{1/\delta}) \}
    \end{align*}
\end{theorem}

\subsubsection*{Proof outline}
We use the same core-tail decomposition as in
Section~\ref{SEC:SREV}, albeit we set the cutoffs differently for
the remaining analysis. We bound the tail in
Section~\ref{subsec:tailSrevBrev}, in a manner similar to
Section~\ref{subsec:SREV:tail}, showing in
Lemma~\ref{lem:tailofSrevBrev} that the tail is a tiny fraction of
the revenue from selling items separately to a constant number of
buyers. Due to the new cutoffs the analysis slightly different.

We bound the core in Section~\ref{subsec:SrevBrevCore}.
An optimistic approach would attempt to use the same cutoffs as in Section~\ref{SEC:SREV}, and hope that as in \cite{babaioff2014simple}, the variance of the sum of values in the core is bounded by $c \cdot \SREV(\distM)^2$ for some constant $c$.
Unfortunately, their analysis relies on the cutoffs having the same value.
When items are non-identically distributed, this is not the case for the cutoffs from Section~\ref{SEC:SREV} which have the same {\em quantile}.
Luckily, this tension is manageable by the following intuition: set the same cutoff for all items,
and see which items $H$ exceed the cutoff with constant probability.
Increase the cutoff until the number of such items $\cardinality{H}$ is a constant.
By independence, a buyer also has high values for {\em all } items in $H$ with constant probability,
hence, a good bundle price for the remaining items, plus a good bundle price for items in $H$, is a good bundle price for all items!
This intuition is made formal in two lemmata (\ref{lem:smallCoreForLowValues} and \ref{lem:largeCoreForLowValues}).

Finally, we complete the proof in
Section~\ref{sub:AdditivebuyerSrevBrevProof} by combining the core-tail
decomposition with the above arguments.

\paragraph{Cutoff setting.}
Fix some $\varepsilon_0, \varepsilon_1 \leq 1$.
For each item $j$ let $c_j$ be so that $\Pr\left[\vals_j > c_j\right] = \varepsilon_1$,
and set $\cutoff_j = \max\{\varepsilon_0 \cdot \SREV(\distM), c_j\}$.
We call the set of items for which the cutoff equals $\varepsilon_0 \cdot \SREV(\distM)$ the set of ``low items'' $\low = \{j : \cutoff_j = \varepsilon_0 \cdot \SREV(\distM)\}$, and the remaining set $\high = \{j : \cutoff_j = c_j > \varepsilon_0 \SREV(\distM) \}$ we call the set of ``high items''.

\subsection{Tail} \label{subsec:tailSrevBrev}

The following lemma is the analog of Lemma~\ref{lem:TailUpperBound}, and similarly it shows that $\SREV$ with a constant number of buyers can extract {\em much more revenue} than the tail (the differences are due to the different cutoffs).

\begin{lemma} \label{lem:tailofSrevBrev}
    $
    \sum_{A}{p_A}{\REV(\distM^A_T)}
    \leq
    4(1+\varepsilon_0^{-1})\varepsilon_1 \cdot \SREV(\distM^{
        1/(2\varepsilon_1)
        }).
    $
\end{lemma}
\begin{proof}
    By repeating the first part of the proof of lemma~\ref{lem:tailToSeperate} we get:
    \begin{align*}
    \sum_{A}{p_A}{\REV(\distM^A_T)}
    \leq &
    \sum_{j\in [m]}{\REV(\vals_j \cdot \ind\left[\vals_j >\cutoff_j\right] )\sum_{A: A\ni j}{\cardinality{A}\frac{p_A}{\Pr\left[\vals_j > \cutoff_j \right]}}}
    \end{align*}
    As in proposition~1 in \cite{babaioff2014simple}, we observe that
    the rightmost sum is the expected size of the set in the tail, conditioned on $j$ being in the tail, i.e., $1+\sum_{k \neq j}{\Pr[\vals_k > \cutoff_k]} $. For every item $j$, by definition of $\cutoff_j$ and by Lemma~\ref{lem:tail1},  it holds that
    $\Pr\left[\vals_j > \cutoff_j \right] \leq \Pr\left[\vals_j > \varepsilon_0\cdot \SREV(\distM) \right]  \leq \frac{\REV(\dist_j)}{\varepsilon_0 \cdot \SREV(\distM)}$.
    Therefore, $1+\sum_{k \neq j}{\Pr[\vals_k > \cutoff_k]} \leq 1 + \varepsilon_0^{-1}$.
    Hence,
    $\sum_{A}{p_A}{\REV(\distM^A_T)}
    \leq
    (1+\varepsilon_0^{-1})\SREV(\distMtail)$.
    By definition of $\cutoff_j$, it holds that $\Pr\left[\vals_j > \cutoff_j \right] \leq \varepsilon_1$.
    Therefore by applying
    Lemma~\ref{lem:addbuyers} with $\delta = 2\varepsilon_1$ for every $j$, we get that:
    $
    \sum_{A}{p_A}{\REV(\distM^A_T)}
    \leq
    (1+\varepsilon_0^{-1})\cdot 4 \cdot \varepsilon_1 \sum_{j\in [m]}{\REV(\dist_j^{1/(2\varepsilon_1)})}
    $
    as required.
\end{proof}

\subsection{Core} \label{subsec:SrevBrevCore}
In this section we bound the contribution from the core.
As previously mentioned, we separate all items $[m]$ to low items $\low$ and high items $\high$.
For low items, we consider two cases.
(1)
The contribution of low items to the core is bounded by a tiny fraction of $\SREV$.
In this case $\SREV$ with a constant number of bidders almost fully recovers the contribution from the core, as shown in Lemma~\ref{lem:smallCoreForLowValues}.
(2) Otherwise, a concentration bound shows that selling all low items in a bundle can recover almost their entire contribution to the core with a constant probability.
In this case $\BREV$ with a constant number of buyers almost fully recovers the contribution from the core, as shown in Lemma~\ref{lem:largeCoreForLowValues}.
To separate the cases, fix some $\varepsilon_2 >0$.
Also, let $\varepsilon_4 = \varepsilon_4(\varepsilon_1, \delta)$ be so that $1+\varepsilon_4 = (1-e^{-\varepsilon_1 / \delta})^{-1}$.
Note that $\varepsilon_4$ is increasing in $\delta$.

\begin{lemma}\label{lem:smallCoreForLowValues}
    If it holds that:
    \begin{align} \label{eq:CoreValsLow}
    \sum_{j\in \low}{\expect{\vals_j \gets \dist_j}{\vals_j \cdot \ind\left[\vals_j \leq \cutoff_j \right] }} \leq \varepsilon_2\cdot \SREV(\distM),
    \end{align}
    then for any constant $\delta>0$ it holds that
    $
    \VAL(\distMcore)
    \leq
    \left(1+  \varepsilon_2 + \varepsilon_4\right)\SREV\left( \distM^{1/\delta} \right).
    $
\end{lemma}

\begin{proof}
    Given Inequality~(\ref{eq:CoreValsLow}),
    it remains to bound
    $
    \sum_{j\in \high}{\expect{\vals_j \gets \dist_j}{\vals_j \cdot \ind\left[\vals_j \leq \cutoff_j \right] }}.
    $
    Applying Lemma~\ref{lem:addbuyersAlmost} gives
    $
    (1-e^{-\varepsilon_1 / \delta})^{-1} \REV(\dist_j^{1/\delta})
    \geq
    \cutoff_j
    \geq
    \expect{\vals_j \gets \dist_j}{\vals_j \cdot \ind\left[\vals_j \leq \cutoff_j \right] }.
    $
    Recall that $\distM^{\high}$ denotes $\distM$ restricted to the items in set $\high$. Summing over all $j\in \high$ we get that
    \begin{align} \label{eq:SrevCoreBound}
     (1+\varepsilon_4)\SREV\left( \distM^{1/\delta} \right)
     \geq
     (1+\varepsilon_4) \SREV\left(\left( \distM^{\high} \right)^{1/\delta} \right )
     \geq
     \sum_{j\in \high}{\expect{\vals_j \gets \dist_j}{\vals_j \cdot \ind\left[\vals_j \leq \cutoff_j \right] }}
    \end{align}
    Combining inequalities~(\ref{eq:CoreValsLow})~and~(\ref{eq:SrevCoreBound}) completes the proof.
\end{proof}

The following lemma shows that when the contribution of low items to
the core is significant, then selling all items in a bundle to a
constant number of buyers almost fully recovers all the contribution
from the core. Let $\varepsilon_5 = \varepsilon_5(\varepsilon_0,
\varepsilon_1, \varepsilon_2, \delta)$ be so that $1 + \varepsilon_5
=  \left(1-\exp\left(-\varepsilon_1^{(\varepsilon_1 \cdot
\varepsilon_0)^{-1}}/ 2\delta\right)\right)^{-1} \cdot
\left(1-\frac{2\sqrt{\varepsilon_0}}{\varepsilon_2} \right)^{-1}$.
Note that $\varepsilon_5$ can be made very small by maintaining $
\delta \ll \varepsilon_1 \ll \varepsilon_0 \ll \varepsilon_2^2 \ll
1$.

\begin{lemma} \label{lem:largeCoreForLowValues}
    If it holds that:
    \begin{align} \label{eq:CoreValsHigh}
    \sum_{j\in \low}{\expect{\vals_j \gets \dist_j}{\vals_j \cdot \ind\left[\vals_j \leq \cutoff_j \right] }} > \varepsilon_2\cdot \SREV(\distM)
    \end{align}
    then for any $0 < \delta \leq 1$,
    $$
    \VAL(\distMcore)
    \leq
    (1+\varepsilon_5) \BREV\left(\distM^{1/\delta} \right)
    $$
\end{lemma}
\begin{proof}
    Recall that for a low item $j$, $\cutoff_j = \varepsilon_0 \cdot \SREV(\distM)$.
    For simplicity, let $R_j = \REV(\vals_j \cdot \ind\left[\vals_j \leq \cutoff_j \right])$.
    By lemma~\ref{lem:boundVar},
    $$
    \var(\vals_j \cdot \ind\left[\vals_j \leq \cutoff_j \right])
    \leq
    (2\frac{\cutoff_j}{R_j} -1 ) \cdot R_j^2
    \leq
    2\varepsilon_0 \SREV(\distM) \cdot R_j.
    $$
    By independence across items, and as $\sum_{j\in \low}R_j \leq \sum_{j\in \low}{\REV(\dist_j)} \leq \SREV(\distM) $ we get that:
    \begin{align} \label{eq:varBound}
    \var\left(\sum_{j\in \low}{\vals_j \cdot \ind\left[\vals_j \leq \cutoff_j \right]}\right)
    =
    \sum_{j\in \low}{\var\left(\vals_j \cdot \ind\left[\vals_j \leq \cutoff_j \right]\right)}
    \leq
    2\varepsilon_0 \SREV(\distM) \sum_{j\in \low}{R_j}
    \leq
    2\varepsilon_0 \SREV(\distM)^2
    \end{align}

    By Chebyshev's inequality, for  any $\varepsilon_3 >0$:
    \begin{align} \label{eq:chebyshev}
    \Pr\left[\sum_{j \in \low}{\vals_j \cdot \ind\left[\vals_j \leq \cutoff_j \right]} \leq (1-\varepsilon_3)\expect{}{\sum_{j \in \low}{\vals_j \cdot \ind\left[\vals_j \leq \cutoff_j \right]}}\right]
    \leq &
    \frac{\var\left(\sum_{j \in \low}{\vals_j \cdot \ind\left[\vals_j \leq \cutoff_j \right]}\right)}{\varepsilon_3^2\cdot \expect{}{\sum_{j \in \low}{\vals_j \cdot \ind\left[\vals_j \leq \cutoff_j \right]}}^2} \nonumber \\
    < &
    \frac{2\varepsilon_0\SREV(\distM)^2}{\varepsilon_3^2 \varepsilon_2^2\cdot \SREV(\distM)^2}  \nonumber \\
    = &
    \frac{2\varepsilon_0}{\varepsilon_3^2 \varepsilon_2^2}
    \end{align}
    Where the second inequality follows by applying Inequality~(\ref{eq:CoreValsHigh})
    to the denominator and Inequality~(\ref{eq:varBound}) to the numerator.

    Set
    $P_{\low}  = (1-\varepsilon_3)\expect{}{\sum_{j \in \low}{\vals_j \cdot \ind\left[\vals_j \leq \cutoff_j \right]}}$,
    and
    set $P_{\high} =\sum_{j\in \high}{\cutoff_j}$.
    For ease of exposition, set $\varepsilon_3 = \sqrt{\frac{4\varepsilon_0}{\varepsilon_2^2} }$.
    Since $\vals_j \geq \vals_j\cdot \ind\left[\vals_j \leq \cutoff_j \right]$ for every $j$, Inequality~(\ref{eq:chebyshev}) implies that
    $
    \Pr\left[\sum_{j \in \low}{\vals_j} > P_{\low}\right]
    \geq
    \frac{1}{2},
    $
    and by definition of high items
    $
    \Pr[\sum_{j\in \high}{\vals_j} > \sum_{j\in \high}{\cutoff_j}]
    \geq
    \prod_{j\in \high}\Pr[\vals_j > \cutoff_j]
    =
    \varepsilon_1^{\cardinality{\high}}
    $.
    Therefore the probability that the sum of all item values exceeds $P_{\low} + P_{\high}$ is at least:
    \begin{gather*}
    \Pr\left[\sum_{j}{\vals_j} > P_{\low} + P_{\high} \right]
    \geq
    \Pr\left[\sum_{j \in \low}{\vals_j } > P_{\low} , \sum_{j \in \high}{\vals_j } > P_{\high} \right]
    \geq
    \frac{1}{2} \cdot \varepsilon_1^{\cardinality{\high}}
    \end{gather*}
    where the last inequality follows by independence across items.
    To bound $\cardinality{\high}$, observe that
   $
        \SREV(\distM)
        \geq
        \sum_{j\in \high}{\varepsilon_1\cdot \cutoff_j}
        >
        \varepsilon_1 \cdot \sum_{j \in \high}{\varepsilon_0 \cdot \SREV(\distM)}
  $
  since we can sell every high item at price $\cutoff_j$ (and every low item at $0$).
  Dividing by $\SREV(\distM)$ and rearranging implies that $\cardinality{\high} < (\varepsilon_0 \cdot \varepsilon_1)^{-1}$.
  Since $\varepsilon_1 \leq 1$ we get that $\varepsilon_1^{\cardinality{\high}} \geq \varepsilon_1^{(\varepsilon_1 \cdot \varepsilon_0)^{-1}}$.
    Observe that
    $P_{\low}+P_{\high} \geq
    (1-\varepsilon_3)\expect{}{\sum_{j \in [m]}{\vals_j \cdot \ind\left[\vals_j \leq \cutoff_j \right]}}
    $.

    Apply Lemma~\ref{lem:addbuyersAlmost} to
    the random variable of the sum $\sum_{j\in [m]}{\vals_j }$,
    with the cutoff
    $\cutoff_{\alpha} =(1-\varepsilon_3)\expect{}{\sum_{j \in [m]}{\vals_j \cdot \ind\left[\vals_j \leq \cutoff_j \right]}}$,
    and the probability $\alpha = \frac{1}{2} \cdot \varepsilon_1^{(\varepsilon_1 \cdot \varepsilon_0)^{-1}}$
    of the sum to exceed $\cutoff_{\alpha}$.
    We get that:
    $$
    \BREV\left(\distM^{1/\delta} \right)
    \geq
    (1-\exp(-\alpha/\delta)) \cdot \cutoff_{\alpha},
    $$
as required.
\end{proof}

\subsection{Proof of Theorem~\ref{thm:oneAdditivebuyerSrevBrev}}\label{sub:AdditivebuyerSrevBrevProof}
\begin{proof}[Proof of Theorem~\ref{thm:oneAdditivebuyerSrevBrev}]
    Combining the core-tail decomposition lemma (Lemma~\ref{lem:coreDecompositionManyBuyers}) for $n=1$ with
    Lemma~\ref{lem:tailofSrevBrev} gives:
    \begin{align} \label{eq:oneAdditivebuyerSrevBrev:1}
    \REV(\distM)
    \leq
    \VAL(\distMcore)
    +
    4(1+\varepsilon_0^{-1})\varepsilon_1 \cdot \SREV(\distM^{1/(2\varepsilon_1)})
    \end{align}
    Combining lemmas~\ref{lem:smallCoreForLowValues}~and~\ref{lem:largeCoreForLowValues} we get that:
    \begin{gather} \label{eq:coreSrevBrev}
    \VAL(\distMcore) =
    1+ \max\{\varepsilon_2 + \varepsilon_4, \varepsilon_5\} \cdot \max\{\SREV\left( \distM^{1/\delta} \right), \BREV\left(\distM^{1/\delta} \right)
    \}.
    \end{gather}
    Adding buyers does not decrease revenue, therefore by combining Inequalities~(\ref{eq:oneAdditivebuyerSrevBrev:1})~and~(\ref{eq:coreSrevBrev}), and taking $\delta < 2\varepsilon_1$ we get:
    \begin{align*}
    \REV(\distM)
    \leq
    \left( 1 + 4(1+\varepsilon_0^{-1})\varepsilon_1 + \max\{\varepsilon_2 + \varepsilon_4, \varepsilon_5\} \right)
    \cdot \max\{\SREV\left( \distM^{1/\delta} \right), \BREV\left(\distM^{1/\delta} \right) \}.
    \end{align*}
    Finally, we can set $4(1+\varepsilon_0^{-1})\varepsilon_1 + \max\{\varepsilon_2 + \varepsilon_4, \varepsilon_5\} \leq \varepsilon$ for a sufficiently small (and constant) $\delta >0$.
\end{proof}

\section{Single Buyer, Regular distributions: $\BVCG$ with a constant number of buyers}
% !TeX root = main99revenue.tex
\label{SEC:BREV}
In this section we prove the following theorem, which immediately implies Theorem~\ref{thm:BREVregular}.
Let BVCG be the mechanism that sells the grand bundle via the VCG mechanism.
\begin{theorem}     \label{thm:oneAdditivebuyerBrevRegular}
    For any constant $\varepsilon>0$, there exists a constant $\delta> 0$ such that:
    \begin{gather*}
        \REV(\distM)
        \leq
        (1+\varepsilon)\REV_{\BVCG}(\distM^{1/\delta}).
    \end{gather*}
\end{theorem}

Single dimensional regular distributions are appealing since they have a ``well behaved tail'' property that is exploited mostly by prior-independent mechanisms.
Unfortunately, even when every $\vals_j \gets \dist_j$ is regular, the {\em grand bundle value} $\sum_{j}{\vals_j}$ need not be regular,  except for very specific cases.
We show that, as one would hope, even though $\sum_{j}{\vals_j}$ is not distributed regularly by itself, the underlying regularity of each $\vals_j$ would still maintain some ``well behaved'' properties.

We use the same core-tail decomposition as in Section~\ref{SEC:SREV}, albeit we set the cutoffs differently for the remaining analysis.
We bound the tail in Section~\ref{subsec:regularTail}, in a manner similar to Section~\ref{subsec:SREV:tail}, showing
in Lemma~\ref{lem:tailToSeparateRegular} that the tail is bounded by a constant factor times the revenue from selling items separately (to a single buyer) using tail prices, 
and in Lemma~\ref{lem:coreToTailRegular} we use the regularity condition to 
show that this revenue is a tiny fraction of the contribution from the core.

We bound the core in Section~\ref{sub:BREV-core}.
We show 
in Lemma~\ref{lem:largeCoreForLowValuesRegular}
that 
a concentration bound suggests a bundle price that almost matches the contribution from the core, and sells with constant probability.
Hence, we can show that two out of a constant number of buyers will be willing to buy at this bundle price with probability almost $1$ (i.e., the second highest value - the revenue of VCG, is higher than the bundle price).

Finally, we complete the proof in Section~\ref{sub:oneAdditivebuyerBrevRegular} by combining the core-tail decomposition with the above arguments.

\paragraph{Cutoff setting.}
Fix some $\varepsilon_0, \varepsilon_1 \leq 1$.
For each item $j$ let $c_j$ be so that $\Pr\left[\vals_j > c_j\right] = \varepsilon_1$.
Fix some $\cutoff >0 $ to be decided later.
Recall that $\distMtail$ denotes the product distribution of the random variables $\{\vals_j \cdot \ind\left[\vals_j >\cutoff_j \right]\}_{j \in [m]}$.
If we set $\cutoff_j = \max\{\cutoff, c_j\}$ for every item $j$, then $\SREV\big( \distMtail \big)$ decreases as $\cutoff$ increases.
Therefore, there exists some $\cutoff$ so that $\cutoff = \varepsilon_0 \cdot \SREV\big( \distMtail \big)$.
This is the value we choose for $\cutoff$.
Call the set of items for which the cutoff is $\cutoff$ the set of ``low items'' $\low = \{j : \cutoff_j = \cutoff\}$, and call the remaining set $\high = \{j : \cutoff_j = c_j > \cutoff \}$ the set of ``high items''.

\subsection{Tail} \label{subsec:regularTail}
Lemma~\ref{lem:tailToSeparateRegular} is an analog of Lemma~\ref{lem:tailToSeperate} which relates the contribution from the tail to $\SREV(\distMtail)$. In this section we use \SREV only for the purpose of analysis (our final mechanism only sells the grand bundle).
\begin{lemma} \label{lem:tailToSeparateRegular}
    $\sum_{A \subseteq [m]}{p_A\REV(\distM^A_T)} \leq (1+\varepsilon_0^{-1})\SREV\big( \distMtail \big)$
\end{lemma}
\begin{proof}
    Repeating the proof of lemma~\ref{lem:tailToSeperate} gives:
        \begin{align*}
        \sum_{A}{p_A}{\REV(\distM^A_T)}
        \leq &
        \sum_{j\in [m]}{\REV(\vals_j \cdot \ind\left[\vals_j >\cutoff_j\right] )\sum_{A: A\ni j}{\cardinality{A}\frac{p_A}{\Pr\left[\vals_j > \cutoff_j \right]}}} \\
        \end{align*}
        Again, we observe that
        the rightmost sum is the expected size of the set in the tail, conditioned on $j$ being in the tail, i.e., $1+\sum_{k \neq j}{\Pr[\vals_k > \cutoff_k]}$.
        Then since $\cutoff_j \geq \varepsilon_0 \cdot \SREV\big( \distMtail \big)$, 
        \begin{align*}
        \sum_{j}\Pr\left[\vals_j > \cutoff_j \right]
        & =
        \sum_{j}\Pr\left[\vals_j \cdot \ind\left[\vals_j > \cutoff_j \right] > \varepsilon_0 \cdot \SREV\big( \distMtail \big) \right] \\ 
        &= 
        \sum_{j}\Pr\left[\vals_j \cdot \ind\left[\vals_j > \cutoff_j \right] > 
	        \REV(\vals_j \cdot \ind\left[\vals_j > \cutoff_j \right]) \cdot \frac{\varepsilon_0 \cdot \SREV\big( \distMtail \big)}{\REV(\vals_j \cdot \ind\left[\vals_j > \cutoff_j \right])}  \right]
		        \\ 
        & \leq 
        \sum_{j}\frac{\REV(\vals_j \cdot \ind\left[\vals_j > \cutoff_j \right])}{\varepsilon_0 \cdot \SREV\big( \distMtail \big)} \\ 
   		&  = \frac{1}{\varepsilon_0}.
        \end{align*}
        Where the inequality follows by applying Lemma~\ref{lem:tail1} to every $\vals_j \cdot \ind\left[\vals_j > \cutoff_j \right]$.
\end{proof}

In the following lemma we show that the contribution from the core is much larger than the contribution from the tail. Thus it will be sufficient to approximate the contribution from the core (which we do in Section~\ref{sub:BREV-core}).
This step uses the regularity assumption.
For ease of exposition, for $k$ that satisfies $\varepsilon_1 \cdot \varepsilon_2^{-k} < 1$, let $\largeConstReg = \largeConstReg(\varepsilon_1, \varepsilon_2, k)  = k \cdot  (1 - \varepsilon_2) \cdot (1- \varepsilon_1 \cdot \varepsilon_2^{-k})$. 
Note that $\largeConstReg$ is almost $k$ when $\varepsilon_1 \ll \varepsilon_2^{k} \ll 1$.
\begin{lemma} \label{lem:coreToTailRegular}
    $
    \VAL(\distMcore)
    \geq
    \largeConstReg \cdot \SREV\big( \distMtail \big)
    $
\end{lemma}

\begin{proof}
Recall that 
$
\VAL(\distMcore)
= 
\sum_{A \subseteq [m]}{p_A \cdot \VAL(\distM_{A}^{C})}
=
\sum_{A \subseteq [m]}{p_A \sum_{j \in \comp{A}} {\frac{\expect{}{\vals_j \cdot \ind\left[\vals_j \leq \cutoff_j \right]}}{\Pr\left[\vals_j \leq \cutoff_j \right]}}}
$.
By Lemma~\ref{lem:CoreToTailRegular}, for every $j$ we have that
$
\expect{}{\vals_j \cdot \ind\left[\vals_j \leq \cutoff_j \right]}
\geq
\largeConstReg  \cdot \REV(\vals_j \cdot \ind\left[\vals_j > \cutoff_j \right])
$, 
therefore:
\begin{align*}
\VAL(\distMcore)
\geq 
\largeConstReg \cdot \sum_{A \subseteq [m]}{p_A \sum_{j \in \comp{A}} {\frac{\REV(\vals_j \cdot \ind\left[\vals_j > \cutoff_j \right])}{\Pr\left[\vals_j \leq \cutoff_j \right]}}}
\end{align*}
Finally, observe that 
\begin{align*}
\sum_{A \subseteq [m]}{p_A \sum_{j \in \comp{A}} {\frac{\REV(\vals_j \cdot \ind\left[\vals_j > \cutoff_j \right])}{\Pr\left[\vals_j \leq \cutoff_j \right]} }}
=
\sum_{j \in [m]}{\REV(\vals_j \cdot \ind\left[\vals_j > \cutoff_j \right])\cdot \frac{\sum_{A : j \in \comp{A}}{p_A}}{\Pr\left[\vals_j \leq \cutoff_j \right]}}
=
\SREV\big( \distMtail \big)
\end{align*}
because $\sum_{A:j \in \comp{A}}{p_A}$ is the total probability of $j$ being in the core, i.e., exactly $\Pr\left[\vals_j \leq \cutoff_j \right]$.
\end{proof}

\subsection{Core}\label{sub:BREV-core}
As in Section~\ref{SEC:SREVBREV}, we separate items to high and low items, and reason about the concentration of the sum of low items.
Unlike Section~\ref{SEC:SREVBREV}, we show that the contribution of low items to the core is always sufficiently large.

The proof of Lemma~\ref{lem:largeCoreForLowValuesRegular} is as follows: 
a concentration bound shows that selling all low items in a bundle can recover almost their entire contribution to the core with a
constant probability. 
Therefore, by adding a constant number of buyers, at least two would be willing to buy at the suggested bundle price, with probability almost $1$, and VCG for the bundle of all items extracts at least as much revenue.

Let 
$\varepsilon_4 
= 
\varepsilon_4(\varepsilon_0, \varepsilon_1, \varepsilon_2, \delta)$ be so that   
$1 + \varepsilon_4 
= \left(1-\exp \left(-\varepsilon_1^{(\varepsilon_1 \cdot \varepsilon_0)^{-1}} / 2 \delta \right) \right)^{-2} \cdot (1-\sqrt{\varepsilon_0})^{-1}$.
Note that $\varepsilon_4$ can be made very small by maintaining $\delta \ll \varepsilon_0, \varepsilon_1 \ll 1$.

\begin{lemma} \label{lem:largeCoreForLowValuesRegular}
	For any constant $1 \geq \delta >0$ it holds that 
	\begin{align*}
	\VAL\big(\distMcore\big)
	\leq
	(1+\varepsilon_4)
	\REV_{\BVCG}(\distM^{1/\delta})
	\end{align*}
\end{lemma}
\begin{proof}
Let $P_{\low} = \sum_{j \in \low}{\expect{\vals_j \gets \dist_j}{\vals_j \cdot \ind\left[\vals_j \leq \cutoff_j \right]}}$,
and $P_{\high} = \sum_{j \in \high}{\expect{\vals_j \gets \dist_j}{\vals_j \cdot \ind\left[\vals_j \leq \cutoff_j \right]}}$.
We will show a lower bound on $\BREV$ by analyzing the auction that sells the grand bundle for price $(1-\varepsilon_3)(P_{\low} +P_{\high})$.
In particular, we will show that with constant probability, each buyer has value at least $P_{\high}$ for the high items, and at least $P_{\low} -\varepsilon_3(P_{\low} +P_{\high})$ for the low items.
Therefore, for a sufficiently large but constant number of buyers, we expect that at least one of them will buy the grand bundle for price $(1-\varepsilon_3)(P_{\low} +P_{\high})$.

    For every low item $j \in \low$, recall that $\cutoff_j = \varepsilon_0 \cdot \SREV\big( \distMtail \big)$.
The variance on its core values is therefore bounded by:
    \begin{align*}
    \var(\vals_j \cdot \ind\left[\vals_j  \leq \cutoff_j \right])
& \leq \expect{\vals_j \gets \dist_j}{\vals_j^2 \cdot \ind\left[\vals_j \leq \cutoff_j \right]} 
 \leq \expect{\vals_j \gets \dist_j}{\vals_j \cdot \ind\left[\vals_j \leq \cutoff_j \right]} \cutoff_j \\
& = \expect{\vals_j \gets \dist_j}{\vals_j \cdot \ind\left[\vals_j \leq \cutoff_j \right]} \varepsilon_0 \cdot \SREV\big( \distMtail \big)  \end{align*}

By independence across items we get that:
\begin{align} \label{eq:varBoundRegular}
\var\left(\sum_{j\in \low}{\vals_j \cdot \ind\left[\vals_j \leq \cutoff_j \right]}\right)
=
\sum_{j\in \low}{\var\left(\vals_j \cdot \ind\left[\vals_j \leq \cutoff_j \right]\right)}
\leq
\varepsilon_0 \SREV\big( \distMtail \big) \cdot P_{\low}.
\end{align}

By Lemma~\ref{lem:CoreToTailRegular} applied to every item $j$ we have that
$$ P_{\low} +P_{\high} = \sum_{j \in [m]}{\expect{\vals_j \gets \dist_j}{\vals_j \cdot \ind\left[\vals_j \leq \cutoff_j \right]}} \geq K \cdot \SREV\big( \distMtail \big).$$

Combining the last two inequalities, we have that
\begin{gather} \label{eq:CoreValsHighRegular}
\sum_{j\in \low}{\var(\vals_j \cdot \ind\left[\vals_j  \leq \cutoff_j \right]) }
\leq
\varepsilon_0 \cdot K^{-1} \cdot P_{\low}(P_{\low} +P_{\high}).
\end{gather}

By Chebyshev's inequality and Inequality~(\ref{eq:CoreValsHighRegular}),
\begin{align} \label{eq:chebyshevRegular}
\Pr\left[\sum_{j \in \low}{\vals_j \cdot \ind\left[\vals_j \leq \cutoff_j \right]}
    \leq
    P_{\low} -\varepsilon_3(P_{\low} +P_{\high})\right]
\leq &
\frac{\var\left(\sum_{j \in \low}{\vals_j \cdot \ind\left[\vals_j \leq \cutoff_j \right]}\right)}
    {\varepsilon_3^2 (P_{\low} +P_{\high})^2} \nonumber \\
\leq &
\frac{\varepsilon_0 \cdot P_{\low}(P_{\low} +P_{\high})}
    {K \cdot \varepsilon_3^2 (P_{\low} +P_{\high})^2} \nonumber \\
< &
 \frac{\varepsilon_0}{\varepsilon_3^2 \cdot K}
\end{align}
For ease of exposition, set $\varepsilon_3 = \sqrt{2 \varepsilon_0}$.
Since $\vals_j \geq \vals_j\cdot \ind\left[\vals_j \leq \cutoff_j \right]$ for every $j$, Inequality~(\ref{eq:chebyshevRegular}) implies that
\begin{gather} \label{eq:lowItemsRegularChebychev}
    \Pr\left[\sum_{j \in \low}{\vals_j}
    >
    P_{\low} -\varepsilon_3(P_{\low} +P_{\high})\right]
    \geq
    1- \frac{1}{2\largeConstReg} 
    \geq 
    \frac{1}{2}.
\end{gather}
Also, by the definition of high items:
\begin{gather} \label{eq:highItemsRegular}
\Pr[\sum_{j\in \high}{\vals_j} > P_{\high}]
\geq 
\Pr[\sum_{j\in \high}{\vals_j} > \sum_{j\in \high}{\cutoff_j}]
\geq
\prod_{j\in \high}\Pr[\vals_j > \cutoff_j]
=
\varepsilon_1^{\cardinality{\high}}
>
\varepsilon_1^{(\varepsilon_1 \cdot \varepsilon_0)^{-1}}.
\end{gather}
To show the last inequality, suppose we sell every high item at the price $\cutoff_j$ (and every low item at $0$), then
	\begin{gather}\label{eq:continuity}
	\SREV\big( \distMtail \big)
	\geq 
	\sum_{j\in \high}{\varepsilon_1\cdot \cutoff_j} > \varepsilon_1 \cdot \sum_{j \in \high}{\varepsilon_0 \cdot \SREV\big( \distMtail \big)}.
\end{gather}
	Dividing by $\SREV\big( \distMtail \big)$ and rearranging implies that $\cardinality{\high} < (\varepsilon_0 \cdot \varepsilon_1)^{-1}$.

Therefore the probability that the sum of all item values exceeds $(1-\varepsilon_3)(P_{\low} +P_{\high})$ is at least:
\begin{align*}
    \Pr\left[\sum_{j\in [m]}{\vals_j} > (1-\varepsilon_3)(P_{\low} +P_{\high})\right]
    \geq &
    \Pr\left[\sum_{j\in \low }{\vals_j} > P_{\low} - \varepsilon_3(P_{\low}+ P_{\high}), \sum_{j \in \high}{\vals_j} > P_{\high}\right] \\
    = &
    \Pr\left[\sum_{j\in \low }{\vals_j} > P_{\low} - \varepsilon_3(P_{\low}+ P_{\high})\right]\Pr\left[\sum_{j \in \high}{\vals_j} > P_{\high}\right] \\
    \geq &
    \frac{1}{2}\cdot (\varepsilon_1^{(\varepsilon_1 \cdot \varepsilon_0)^{-1}}), 
\end{align*}
where the equality follows by independence across items, and the last inequality follows by Inequalities~(\ref{eq:lowItemsRegularChebychev})~and~(\ref{eq:highItemsRegular}).
Consider the random variable of the sum $\vals = \sum_{j \in [m]}{\vals_j}$, the cutoff
$\cutoff_{\alpha} = (1-\varepsilon_3)(P_{\low} +P_{\high}) = (1-\varepsilon_3)\VAL\big(\distMcore \big)$, and the probability of the sum to exceed the cutoff
$\alpha = \frac{1}{2}\cdot (\varepsilon_1^{(\varepsilon_1 \cdot \varepsilon_0)^{-1}})$.
By Lemma~\ref{lem:VCGAlmost} get that:
$
\REV_{\BVCG }(\distM^{2/\delta}) 
\geq 
(1-\exp(-2\alpha /\delta))^2 \cdot (1-\varepsilon_3) \cdot \VAL\big(\distMcore \big),
$
as required.
\end{proof}

\subsection{Proof of Theorem~\ref{thm:oneAdditivebuyerBrevRegular}}\label{sub:oneAdditivebuyerBrevRegular}
\begin{proof}[Proof of Theorem~\ref{thm:oneAdditivebuyerBrevRegular}]
    Combining Lemma~\ref{lem:coreToTailRegular} with Lemma~\ref{lem:tailToSeparateRegular}
    we get that
    \begin{gather} \label{eq:coreIsSmallRegular}
    \VAL(\distMcore)
    \geq
    \frac{K}{(1+\varepsilon_0^{-1})} \sum_{A \subseteq [m]}{p_A\REV(\distM^A_T)}
    \end{gather}
    Combining the core-tail decomposition lemma (Lemma~\ref{lem:coreDecompositionManyBuyers}), Inequality~(\ref{eq:coreIsSmallRegular}), 
    and Lemma~\ref{lem:largeCoreForLowValuesRegular} gives:
    \begin{align} \label{eq:oneAdditivebuyerBrevRegular:1}
       \REV(\distM) 
       \leq
       \left(1 +
       \frac{1+\varepsilon_0^{-1}}{K}\right)
       \VAL\big( \distMcore \big)
       \leq
       \left(1 +
       \frac{1+\varepsilon_0^{-1}}{K}\right)
       (1+\varepsilon_4)
       \REV_{\BVCG}(\distM^{1/\delta})
       .
    \end{align}
	Note that we can increase $\largeConstReg$ independently of $\varepsilon_0$'s value, which completes the proof.
\end{proof}
\fi

\bibliographystyle{alpha}
\bibliography{AGT}

\newcommand{\etalchar}[1]{$^{#1}$}
\begin{thebibliography}{EFF{\etalchar{+}}17b}

\bibitem[BCKW10]{briest2010pricing}
Patrick Briest, Shuchi Chawla, Robert Kleinberg, and S~Matthew Weinberg.
\newblock Pricing randomized allocations.
\newblock In {\em Proceedings of the twenty-first annual ACM-SIAM symposium on
  Discrete Algorithms}, pages 585--597. Society for Industrial and Applied
  Mathematics, 2010.

\bibitem[BDHS15]{BateniDHS15}
MohammadHossein Bateni, Sina Dehghani, MohammadTaghi Hajiaghayi, and Saeed
  Seddighin.
\newblock Revenue maximization for selling multiple correlated items.
\newblock In {\em Algorithms - {ESA} 2015 - 23rd Annual European Symposium,
  Patras, Greece, September 14-16, 2015, Proceedings}, pages 95--105, 2015.

\bibitem[BILW14]{babaioff2014simple}
Moshe Babaioff, Nicole Immorlica, Brendan Lucier, and S~Matthew Weinberg.
\newblock A simple and approximately optimal mechanism for an additive buyer.
\newblock In {\em Proceedings of the 2014 IEEE 55th Annual Symposium on
  Foundations of Computer Science}, pages 21--30. IEEE Computer Society, 2014.

\bibitem[BK96]{BulowK96}
Jeremy Bulow and Paul Klemperer.
\newblock Auctions versus negotiations.
\newblock {\em The American Economic Review}, pages 180--194, 1996.

\bibitem[CD11]{CaiD11}
Yang Cai and Constantinos Daskalakis.
\newblock On minmax theorems for multiplayer games.
\newblock In {\em Proceedings of the Twenty-Second Annual {ACM-SIAM} Symposium
  on Discrete Algorithms, {SODA} 2011, San Francisco, California, USA, January
  23-25, 2011}, pages 217--234, 2011.

\bibitem[CDO{\etalchar{+}}15]{ChenDOPSY15}
Xi~Chen, Ilias Diakonikolas, Anthi Orfanou, Dimitris Paparas, Xiaorui Sun, and
  Mihalis Yannakakis.
\newblock On the complexity of optimal lottery pricing and randomized
  mechanisms.
\newblock In {\em {IEEE} 56th Annual Symposium on Foundations of Computer
  Science, {FOCS} 2015, Berkeley, CA, USA, 17-20 October, 2015}, pages
  1464--1479, 2015.

\bibitem[CDP{\etalchar{+}}14]{ChenDPSY14}
Xi~Chen, Ilias Diakonikolas, Dimitris Paparas, Xiaorui Sun, and Mihalis
  Yannakakis.
\newblock The complexity of optimal multidimensional pricing.
\newblock In {\em Proceedings of the Twenty-Fifth Annual {ACM-SIAM} Symposium
  on Discrete Algorithms, {SODA} 2014, Portland, Oregon, USA, January 5-7,
  2014}, pages 1319--1328, 2014.

\bibitem[CDW16]{cai2016duality}
Yang Cai, Nikhil~R Devanur, and S~Matthew Weinberg.
\newblock A duality based unified approach to bayesian mechanism design.
\newblock In {\em Proceedings of the 48th Annual ACM SIGACT Symposium on Theory
  of Computing}, pages 926--939. ACM, 2016.

\bibitem[CH13]{CaiH13}
Yang Cai and Zhiyi Huang.
\newblock {Simple and Nearly Optimal Multi-Item Auctions}.
\newblock In {\em the 24th Annual ACM-SIAM Symposium on Discrete Algorithms
  (SODA)}, 2013.

\bibitem[CHK07]{chawla2007algorithmic}
Shuchi Chawla, Jason~D Hartline, and Robert Kleinberg.
\newblock Algorithmic pricing via virtual valuations.
\newblock In {\em Proceedings of the 8th ACM conference on Electronic
  commerce}, pages 243--251. ACM, 2007.

\bibitem[CHMS10]{chawla2010multi}
Shuchi Chawla, Jason~D Hartline, David~L Malec, and Balasubramanian Sivan.
\newblock Multi-parameter mechanism design and sequential posted pricing.
\newblock In {\em Proceedings of the forty-second ACM symposium on Theory of
  computing}, pages 311--320. ACM, 2010.

\bibitem[CM16]{chawla2016mechanism}
Shuchi Chawla and J.~Benjamin Miller.
\newblock Mechanism design for subadditive agents via an ex ante relaxation.
\newblock In {\em Proceedings of the 2016 ACM Conference on Economics and
  Computation}, EC '16, pages 579--596, New York, NY, USA, 2016. ACM.

\bibitem[CMS15]{ChawlaMS10}
Shuchi Chawla, David~L. Malec, and Balasubramanian Sivan.
\newblock The power of randomness in bayesian optimal mechanism design.
\newblock {\em Games and Economic Behavior}, 91:297--317, 2015.

\bibitem[CZ17]{CaiZ16}
Yang Cai and Mingfei Zhao.
\newblock {Simple Mechanisms for Subadditive Buyers via Duality}.
\newblock In {\em Proceedings of the 49th Annual ACM SIGACT Symposium on Theory
  of Computing}, pages 170--183. ACM, 2017.

\bibitem[DDT14]{daskalakis2014complexity}
Constantinos Daskalakis, Alan Deckelbaum, and Christos Tzamos.
\newblock The complexity of optimal mechanism design.
\newblock In {\em Proceedings of the Twenty-Fifth Annual ACM-SIAM Symposium on
  Discrete Algorithms}, pages 1302--1318. Society for Industrial and Applied
  Mathematics, 2014.

\bibitem[DDT16]{DaskalakisDT16}
Constantinos Daskalakis, Alan Deckelbaum, and Christos Tzamos.
\newblock Strong duality for a multiple-good monopolist.
\newblock {\em Econometrica}, 2016.

\bibitem[EFF{\etalchar{+}}17a]{eden2016competition}
Alon Eden, Michal Feldman, Ophir Friedler, Inbal Talgam{-}Cohen, and S.~Matthew
  Weinberg.
\newblock The competition complexity of auctions: {A} bulow-klemperer result
  for multi-dimensional bidders.
\newblock In {\em Proceedings of the 2017 {ACM} Conference on Economics and
  Computation, {EC} '17, Cambridge, MA, USA, June 26-30, 2017}, page 343, 2017.

\bibitem[EFF{\etalchar{+}}17b]{EdenFFTW16a}
Alon Eden, Michal Feldman, Ophir Friedler, Inbal Talgam{-}Cohen, and S.~Matthew
  Weinberg.
\newblock A simple and approximately optimal mechanism for a buyer with
  complements: Abstract.
\newblock In {\em Proceedings of the 2017 {ACM} Conference on Economics and
  Computation, {EC} '17, Cambridge, MA, USA, June 26-30, 2017}, page 323, 2017.

\bibitem[Elk07]{elkind2007designing}
Edith Elkind.
\newblock Designing and learning optimal finite support auctions.
\newblock In {\em Proceedings of the eighteenth annual ACM-SIAM symposium on
  Discrete algorithms}, pages 736--745. Society for Industrial and Applied
  Mathematics, 2007.

\bibitem[GK16]{GoldnerK16}
Kira Goldner and Anna Karlin.
\newblock A prior-independent revenue-maximizing auction for multiple additive
  bidders.
\newblock In {\em The 12th Conference on Web and Internet Economics (WINE)},
  2016.

\bibitem[Har13]{hartline2013mechanism}
Jason~D Hartline.
\newblock Mechanism design and approximation.
\newblock 2013.

\bibitem[HN12]{hart2012approximate}
Sergiu Hart and Noam Nisan.
\newblock Approximate revenue maximization with multiple items.
\newblock In {\em Proceedings of the 13th ACM Conference on Electronic
  Commerce}, pages 656--656. ACM, 2012.

\bibitem[HN13]{HN13}
Sergiu Hart and Noam Nisan.
\newblock The menu-size complexity of auctions.
\newblock In {\em Proceedings of the 14th ACM Conference on Economics and
  Computation}, pages 565--566, 2013.

\bibitem[HR12]{HartR12}
Sergiu Hart and Philip~J. Reny.
\newblock Maximal revenue with multiple goods: Nonmonotonicity and other
  observations.
\newblock {\em Discussion Paper Series dp630, The Center for the Study of
  Rationality, Hebrew University, Jerusalem}, 2012.

\bibitem[KW12]{kleinberg2012matroid}
Robert Kleinberg and Seth~Matthew Weinberg.
\newblock Matroid prophet inequalities.
\newblock In {\em Proceedings of the forty-fourth annual ACM symposium on
  Theory of computing}, pages 123--136. ACM, 2012.

\bibitem[LP18]{LiuP2017}
Siqi Liu and Christos{-}Alexandros Psomas.
\newblock On the competition complexity of dynamic mechanism design.
\newblock In {\em SODA}, 2018.
\newblock To appear.

\bibitem[LY13]{LiY13}
Xinye Li and Andrew Chi-Chih Yao.
\newblock On revenue maximization for selling multiple independently
  distributed items.
\newblock {\em Proceedings of the National Academy of Sciences},
  110(28):11232--11237, 2013.

\bibitem[MS15]{MS15-production_costs}
Will Ma and David Simchi{-}Levi.
\newblock Reaping the benefits of bundling under high production costs.
\newblock {\em CoRR}, abs/1512.02300, 2015.

\bibitem[Mye81]{Myerson81}
Roger~B. Myerson.
\newblock {Optimal Auction Design}.
\newblock {\em Mathematics of Operations Research}, 6(1):58--73, 1981.

\bibitem[RC98]{RochetC98}
Jean-Charles Rochet and Philippe Chone.
\newblock Ironing, sweeping, and multidimensional screening.
\newblock {\em Econometrica}, 66(4):783--826, July 1998.

\bibitem[RTCY15]{RTY15}
Tim Roughgarden, Inbal Talgam-Cohen, and Qiqi Yan.
\newblock Robust auctions for revenue via enhanced competition.
\newblock Working paper; a preliminary version appeared in EC'12, 2015.

\bibitem[Rub16]{rubinstein2016computational}
Aviad Rubinstein.
\newblock On the computational complexity of optimal simple mechanisms.
\newblock In {\em Proceedings of the 2016 ACM Conference on Innovations in
  Theoretical Computer Science}, pages 21--28. ACM, 2016.

\bibitem[RW15]{rubinstein2015simple}
Aviad Rubinstein and S~Matthew Weinberg.
\newblock Simple mechanisms for a subadditive buyer and applications to revenue
  monotonicity.
\newblock In {\em Proceedings of the Sixteenth ACM Conference on Economics and
  Computation}, pages 377--394. ACM, 2015.

\bibitem[Yao15]{Yao15}
Andrew Chi-Chih Yao.
\newblock An n-to-1 bidder reduction for multi-item auctions and its
  applications.
\newblock In {\em Proceedings of the Twenty-Sixth Annual ACM-SIAM Symposium on
  Discrete Algorithms}, pages 92--109. SIAM, 2015.

\bibitem[Yao17]{Yao17-monotonicity}
Andrew~Chi{-}Chih Yao.
\newblock On revenue monotonicity in combinatorial auctions.
\newblock {\em CoRR}, abs/1709.03223, 2017.

\end{thebibliography}

\appendix
\end{document}